%% file: report.tex
\newif\ifreport\reporttrue
\documentclass[11pt, draftclsnofoot,onecolumn]{IEEEtran}
\usepackage{setspace}
\usepackage{cite}
\usepackage{graphicx}
\usepackage{caption}
\usepackage{subcaption}
\usepackage[cmex10]{amsmath}
\usepackage{mathtools}
\usepackage{amsthm}
\usepackage{amsfonts}
\usepackage{amssymb}
\usepackage[lined, boxed, linesnumbered, ruled]{algorithm2e}
\usepackage{bm,xstring}
\usepackage{fixltx2e}
\usepackage{tikz}
\usetikzlibrary{decorations.pathreplacing}

\usepackage{color}

\newcommand{\ignore}[1]{}
\usepackage[singlelinecheck=false]{caption}

\newtheorem{lemma}{Lemma}
\newtheorem{proposition}{Proposition}

\newtheorem{theorem}{Theorem}
\newtheorem{corollary}{Corollary}
\theoremstyle{definition}

\newtheorem{definition}{Definition}

\begin{document}



\title{On Delay-Optimal Scheduling in Queueing Systems with Replications}

\author{Yin Sun, C. Emre Koksal, and Ness B. Shroff \\

\thanks{Yin Sun and C. Emre Koksal are with the Department of Electrical and Computer Engineering, the Ohio State University, Columbus, OH. Email: sunyin02@gmail.com, koksal.2@osu.edu.

Ness B. Shroff is with the Departments of Electrical and Computer Engineering and Computer Science and Engineering, the Ohio State University, Columbus, OH. Email: shroff.11@osu.edu.}
}

\maketitle
\input{abstract}
\begin{IEEEkeywords}
Queueing systems, replications, delay optimality, data locality, sample-path ordering, work-efficiency ordering, stochastic ordering.
\end{IEEEkeywords}
\newpage
\input{intro}

\input{model}
\input{analysis}

\input{distributed}

\input{numerical}
\input{conclusion}
\ifreport
\appendices
\fi
\input{sample_path}

\input{appendices_lem2}
\input{appendices_lem3}

\input{appendices_lem1_0}
\input{appendices_version2}

\input{appendices_lem2_0}
\input{appendices_2version2}
\input{appendices_coro2_0}
\input{appendices_Theorem1}
\input{appendices_3}
\input{appendices_4}
\input{appendices_dis}
\bibliographystyle{IEEEtran}
\bibliography{ref,ref1}

\end{document}

%% file: abstract.tex
\begin{abstract}
In modern computer systems, jobs are divided into short tasks and executed in parallel. 
Empirical observations in practical systems suggest that the task service times are highly random and the job service time is bottlenecked by the slowest straggling task. One common solution for straggler mitigation is to replicate a task on multiple servers and wait for one replica of the task to finish early. The delay performance of replications depends heavily on the scheduling decisions of when to replicate, which servers to replicate on, and which job to serve first. So far, little is understood on how to optimize these scheduling decisions for minimizing the delay to complete the jobs.
In this paper, we present a comprehensive study on delay-optimal scheduling of  replications in both centralized and distributed multi-server systems. Low-complexity scheduling policies are designed and are proven to be delay-optimal or near delay-optimal in stochastic ordering  among all causal and non-preemptive policies. These theoretical results are established for general system settings and delay metrics that allow for arbitrary arrival processes, arbitrary job sizes, arbitrary due times, and heterogeneous servers with data locality constraints. 
Novel sample-path tools are developed to prove these results.



%
%
%
%
%
%
%
%
%
%
%
%
%
%

\end{abstract}

%% file: intro.tex

\section{Introduction}
Achieving low delay is imperative in modern computer systems. 
Google has found that increasing the delay of Web searching from 0.4 seconds to 0.9 seconds decreases the traffic and ad revenues by 20\% \cite{Linden2006}. Similar results were reported by Amazon, where every 100 milliseconds of extra response time was shown to decrease the sales by 1\% \cite{LindenStanfordtalk}. 
These results imply that prompt responses not only allow users to view more pages in the same period of time, but also provide instant gratification to motivate them to spend more time online \cite{Farber2006}. 
In addition, low delay is critically important for the 
stock market, where the fastest trading decisions are made within a few milliseconds. It was  estimated that a 1-millisecond advantage in trading can be worth 100 million dollars a year for a major brokerage firm \cite{Martin2007}. 
Therefore, 
even small changes in delay can have a significant impact on business success.

As the size and complexity of computer systems continues its significant growth, maintaining low delay becomes increasingly challenging. Long-running jobs are broken into a batch of short tasks which can be executed in parallel over many servers \cite{mapreduce}. Experience in practical systems suggests that the response times of individual servers are highly random, because of resource sharing, network congestion, cache misses, database blocking, background activities, and so on \cite{Dean:2013:Tail}. As a result, the job service delay is constrained by the slowest straggling tasks, causing a long delay tail. 

An efficient technique used to tame the delay tail is \emph{replications} \cite{Ghare:2004,Cirne2007213,vulimiri13latency,ShengboInfocom,Wang:2014,Wang2015}, which is also called \emph{redundant requests}  \cite{shah-Allerton-2013,Kristen2015,Lee-Allerton-2015} and \emph{cloning}  \cite{Ananthanarayanan11,Ananthanarayanan13}. In this technique, multiple replicas of a task are dispatched to different servers and the first completed replica is considered as the valid execution of the task. After that, the remaining replicas of this task can be cancelled to release the servers, possibly with a certain amount of cancellation delay overhead. 
The potential benefits of replications are huge.
For example, in Google's BigTable service which has a high degree of parallelism, replications can reduce the 99.9\%-th percentile delay from 1,800 milliseconds to 74 milliseconds \cite{Dean:2013:Tail}. However, in some other systems, replications may worsen the delay performance, e.g., \cite{Liang2013_2,vulimiri13latency}. In particular, the delay performance of replications depends heavily on the scheduling decisions of when to replicate, which servers to replicate on, and which job to serve first. So far, little is understood on how to optimally schedule replications for minimizing the  delay to complete the jobs.

In this paper, we study delay-optimal scheduling of replications for  centralized and distributed multi-server queueing systems, which are illustrated in Fig. \ref{fig1model_central} and Fig. \ref{fig1model_distri}, respectively. Each job brings with it a batch of tasks, and each task is of one unit of work. The jobs arrive over time according to a general arrival process, where the number, batch sizes, arrival times, and due times of the jobs are \emph{arbitrarily} given. In centralized queueing systems, the jobs arrive at a scheduler and are stored in a job queue. The scheduler determines the  assignment, replication, and cancellation of the tasks, based on the casual information (the history and current information) of the system. In distributed queueing systems, the jobs arrive at multiple parallel schedulers and each scheduler make decisions independently, subject to \emph{data locality} constraints \cite{taskplacement2011,Sparrow:2013}. More specifically, the servers  are divided into multiple server groups, and each task can be only assigned by one group of servers, each of which stores one copy of the data necessary for executing the task.\footnote{The data locality constraints considered here are \emph{hard} constraints, where the servers are not allowed to process remote tasks belonging to other groups. There exists another form of \emph{soft} data locality constraints, where a server can execute remote tasks belonging to other groups by first retrieving the necessary data and then processing the task. Therefore, remote tasks are executed at a slower speed than local tasks. In practical systems, hard data locality constraints are more appropriate for interactive Web services, which need to respond within a few seconds or even shorter time; while soft data locality constraints are usually used in offline services, where the tasks are executed at a slower time scale and hence there is sufficient time to retrieve the necessary data.} The service times of the tasks follow New-Better-than-Used (NBU) distributions or New-Worse-than-Used (NWU) distributions, and are \emph{independent} across the servers and \emph{i.i.d.} across the tasks assigned to the same server. Our goal is to seek for low-complexity scheduling policies that optimize the delay performance of the jobs.

\begin{figure}
\centering
\includegraphics[width=0.5\textwidth]{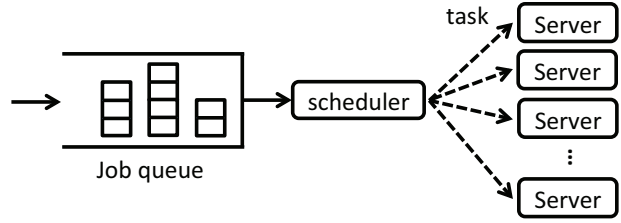}
\caption{A centralized queueing system.} \label{fig1model_central}
\end{figure}
\begin{figure}
\centering
\includegraphics[width=0.8\textwidth]{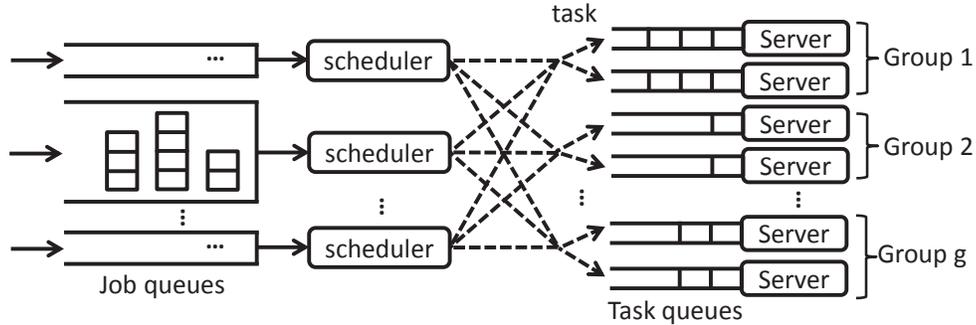}
\caption{A distributed queueing system with data locality constraints. Each task can be only assigned to a group of servers each of which stores one copy of the data necessary for executing the task.} \label{fig1model_distri}
\end{figure}
\subsection{Difficulty of Delay-Optimal Scheduling} 
Without replications, the models that we consider belong to the class of multi-class multi-server queueing systems, where delay optimality has been extremely  difficult to achieve. For example, delay minimization in deterministic scheduling problems (where the service time of each job is known) with more than one servers is $NP$-hard and has no constant competitive ratio \cite{Leonardi:1997}. Similarly, delay-optimal stochastic scheduling (where the service time of each job is random) in multi-class multi-server queueing systems is deemed to be notoriously difficult \cite{Weiss:1992,Weiss:1995,Dacre1999}. Prior attempts on solving the delay-optimal scheduling problem have met little success, except in some limiting regions such as large system limits, e.g., \cite{Ying2015}, and heavy traffic limits, e.g., \cite{Stolyar_heavy2004}. However, these results may not apply outside of these limiting regions or when the stationary distribution of the system does not exist.

In addition, replications add a further layer of difficulty to this problem. If a task is replicated on multiple servers, its service time is reduced, but at a cost of longer waiting times of other tasks. In general, it is difficult to determine whether the gain of shorter service time would exceed the loss of longer waiting times. Hence, ``to replicate or not  to replicate'' is a fundamental dilemma that needs to be resolved in order to design a delay-optimal scheduler. We note that each task has many replication modes (i.e., it can be replicated on different servers and at different time instants), which require different amounts of service time. Thus, the work conservation law \cite{Leonard_Kleinrock_book,Jose2010} does not hold in the study of replications. Hence, certain powerful and well-known  delay minimization methods, such as the achievable region approach \cite{Gittins:11,Dacre1999}, are difficult to apply to  our study.
\subsection{Summary of Main Results}
We develop a number of low-complexity scheduling policies. For arbitrarily given job parameters (including the number, batch sizes, arrival times, and due times of the jobs), these policies are proven to be delay-optimal or near delay-optimal in stochastic ordering for minimizing several classes of delay metrics among all causal and non-preemptive\footnote{We consider task-level non-preemptive policies: Processing of a task cannot be interrupted until the task is completed or cancelled; after completing or cancelling a task, the server can switch to process another task from any job.} policies. Some examples of the delay metrics considered in this paper include the average delay, maximum delay, maximum lateness, increasing and Schur convex functions of delay (e.g., the second moment of delay), etc.\footnote{To the best of our knowledge, except for the average delay, the other delay metrics are considered in the study of replications for the first time.} In particular, the proposed policies are proven to be within a constant additive delay gap from the optimum  for minimizing the mean average delay. 

The key tools in our proofs are new sample-path orderings for comparing the delay performance of different policies. These sample-path orderings are very general because they do not need to specify the queueing system model, and hence can be potentially used for establishing near delay optimality results in other scheduling problems. The interested readers are referred to Appendix \ref{sec_proofmain}.

\subsection{Organization of the Paper} We describe the model and problem formulation in Section \ref{sec_model}, together with the notations that we will use throughout the paper. Replication policies and their delay performance are analyzed for centralized queueing systems in Section \ref{sec_analysis}, and  for distributed queueing systems in Section \ref{sec_distributed}.  Numerical results are  provided in Section \ref{sec_numerical}. Finally, the conclusion is drawn in Section \ref{sec_conclusion}. The sample-path proof method is provided in Appendix \ref{sec_proofmain}.

\section{Related Work}\label{sec_related}
\subsection{Systems Work} 
The benefits of exploiting replications to reduce delay have been empirically studied for many applications \cite{mapreduce,Zaharia:2008,Ananthanarayanan:2010,Dremel2010,Dean:2013:Tail,Ananthanarayanan11,Ananthanarayanan13,Ananthanarayanan2014, vulimiri12latency,vulimiri13latency,addShengboOriginal,Liang2013_2,DTN-delay}. 
In communication networks, multiple replicated copies of a message can be sent over different routing paths to reduce delay \cite{DTN-delay,vulimiri13latency}. 
In cloud computing systems, the task execution time is highly random and it was shown that replicating straggling tasks can significantly reduce the job service delay \cite{mapreduce,Zaharia:2008,Ananthanarayanan:2010,Dremel2010,Ananthanarayanan13,Ananthanarayanan2014,Dean:2013:Tail}. In \cite{vulimiri12latency,vulimiri13latency}, the authors observed that the delay of DNS queries can be reduced by sending multiple replications of a query to multiple servers. 

Replications and more general coding techniques have been proposed to reduce communication delay in cloud storage and information retrieval systems. In \cite{addShengboOriginal,Liang2013_2}, the authors performed experiments on Amazon S3, and found that one can exploit storage redundancy to issue multiple downloading connections to reduce delay. 
Recently, coding techniques were introduced to speedup distributed algorithms such as MapReduce and machine learning \cite{Li-Allerton-2015, Lee-NIPS-2015}. Significant performance improvements were  shown through experiments on Amazon EC2  \cite{Lee-NIPS-2015}.

\subsection{Theoretical Work} 
There has been a growing interest in understanding and characterizing the delay performance of replication and coding techniques. One focus in this area is fast data retrieval in distributed storage systems. In \cite{huang-isit-2012}, Huang et al. showed that codes can reduce the queueing delay in distributed storage systems. In \cite{Joshi-2012,shah-mdsq-2012,Kumar2014}, the authors obtained bounds on the mean average delay of redundant data downloading policies. In \cite{addShengboOriginal,Liang2013_2}, Liang and Kozat provided an approximate analysis for the delay performance of redundant data downloading policies, based on their measurements on Amazon S3. 
There also exist some analytical studies on distributed computing systems.
In \cite{Wang:2014,Wang2015}, Wang et al. studied the tradeoff between delay and computing cost in cloud computing systems.
For Poisson arrivals and exponential service times, Gardner et al. \cite{Kristen2015} characterized the response time distribution for several replication policies in distributed queueing systems. Recently, the delay performance of replications was analyzed in the context of load-balancing \cite{BinInfoCom2016,Gardner2016,Gardner2016_2}, where the number of servers may potentially grow to infinity.  

In addition to the aforementioned studies that focus on delay analysis and characterization, there also exist a few works which aim to find delay-optimal scheduling policies of replications and coding. If the task service times are geometrically distributed and each job has a single task, it was shown in \cite{Borst2003} that distributing the tasks over the servers as evenly as possible is delay-optimal among all admissible policies. 
Later, delay-optimal scheduling of replications was studied for more general service time distributions, such as New-Better-than-Used (NBU) distributions and New-Worse-than-Used (NWU) distributions \cite{Righter2008,Kim2010}. In \cite{shah-Allerton-2013,Lee-Allerton-2015,Joshi2015}, the delay performance of replication policies was analyzed in a few different models, where the optimal policies were obtained for minimizing the mean average delay within a specific family of policies.
In \cite{ShengboInfocom}, Chen et al. presented some similar results with \cite{shah-Allerton-2013}, where the difference is that the delay optimality results in  \cite{ShengboInfocom} were established among all admissible policies. In  \cite{Sun2015}, (near) delay-optimal scheduling results were established for general maximum distance separable (MDS) codes in distributed storage systems. 

This paper differs from the existing works in two aspects: First, our study is carried out for very general system settings and  delay metrics, some of which are considered in the study of replications for the first time. Second, in our study, delay optimality results are established  when it is possible; and in some more general scenarios where delay optimality is inherently difficult to achieve, alternative near delay optimality results with small sub-optimality gaps are obtained.

%% file: model.tex
\section{Model and Formulation}\label{sec_model}
\subsection{Notations and Definitions}
We will use lower case letters such as $x$ and $\bm{x}$, respectively, to represent deterministic scalars and vectors.
In the vector case, a subscript will index the components of a vector, such as $x_i$.
We use $x_{[i]}$ and $x_{(i)}$, respectively, to denote the $i$-th largest and the $i$-th smallest components of $\bm{x}$.  
For any $n$-dimensional vector $\bm{x}$, let $\bm{x}_{\uparrow}=(x_{(1)},\ldots,x_{(n)})$ 
denote the increasing 
rearrangements of $\bm{x}$.  Let $\bm{0}$
denote the vector 
with all 0 
components.

Random variables and vectors will be denoted by upper case letters such as $X$ and $\bf{X}$, respectively, with the subscripts and superscripts following the same conventions as in the deterministic case. 
Throughout the paper, ``increasing/decreasing'' 
and ``convex/concave'' 
are used in the non-strict sense. LHS and RHS denote, respectively, ``left-hand side'' and ``right-hand side''.

For any $n$-dimensional vectors $\bm{x}$ and $\bm{y}$, the elementwise vector ordering $x_i\leq y_i$, $i=1,\ldots,n$, is denoted by $\bm{x} \leq \bm{y}$. Further, $\bm{x}$ is said to be \emph{majorized} by $\bm{y}$, denoted by $\bm{x}\prec\bm{y}$, if (i) $\sum_{i=1}^j x_{[i]} \leq \sum_{i=1}^j y_{[i]}$, $j=1,\ldots,n-1$ and (ii) $\sum_{i=1}^n x_{[i]} = \sum_{i=1}^n y_{[i]}$ \cite{Marshall2011}. In addition, $\bm{x}$ is said to be  \emph{weakly majorized by $\bm{y}$ from below}, denoted by $\bm{x}\prec_{\text{w}}\bm{y}$, if $\sum_{i=1}^j x_{[i]} \leq \sum_{i=1}^j y_{[i]}$, $j=1,\ldots,n$; $\bm{x}$ is said to be  \emph{weakly majorized by $\bm{y}$ from above}, denoted by $\bm{x}\prec^{\text{w}}\bm{y}$, if $\sum_{i=1}^j x_{(i)} \geq \sum_{i=1}^j y_{(i)}$, $j=1,\ldots,n$ \cite{Marshall2011}.
A function that preserves the majorization order is called a Schur convex function. Specifically, $f: \mathbb{R}^n\rightarrow \mathbb{R}$ is termed \emph{Schur convex} if $f(\bm{x})\leq f(\bm{y})$ for all $\bm{x}\prec\bm{y}$ \cite{Marshall2011}. A function $f: \mathbb{R}^n\rightarrow \mathbb{R}$ is termed \emph{symmetric} if $f(\bm{x})= f(\bm{x}_{\uparrow})$ for all $\bm{x}$. The composition of functions $\phi$ and $f$ is denoted by $\phi \circ f (\bm{x}) = \phi(f (\bm{x}))$. Define $x\wedge y=\min\{x,y\}$.

Let $\mathcal{A}$ and $\mathcal{S}$ denote sets and events, with $|\mathcal{S}|$ denoting the cardinality of $\mathcal{S}$.
For all random variable ${X}$ and events $\mathcal{A}$, let $[{X}|\mathcal{A}]$ denote a random variable with the conditional distribution of ${X}$ for given $\mathcal{A}$. A random variable ${X}$ is said to be \emph{stochastically smaller} than another random variable ${Y}$, denoted by ${X}\leq_{\text{st}}{Y}$, if $\Pr({X}>x) \leq \Pr({Y}>x)$ for all~$x\in \mathbb{R}$.
A set $\mathcal{U} \subseteq \mathbb{R}^n$ is called \emph{upper}, if $\bm{y} \in \mathcal{U}$ whenever $\bm{y}\geq \bm{x}$ and $\bm{x} \in \mathcal{U}$. 
A random vector $\bm{X}$ is said to be \emph{stochastically smaller} than another random vector $\bm{Y}$, denoted by $\bm{X}\leq_{\text{st}}\bm{Y}$, if $\Pr(\bm{X}\in \mathcal{U}) \leq \Pr(\bm{Y}\in \mathcal{U})$ for all upper sets ~$\mathcal{U}\subseteq \mathbb{R}^n$. If $\bm{X}\leq_{\text{st}}\bm{Y}$ and $\bm{X}\geq_{\text{st}}\bm{Y}$, then $\bm{X}$ and $\bm{Y}$ follow the same distribution, denoted by $\bm{X}=_{\text{st}}\bm{Y}$. We remark that $\bm{X}\leq_{\text{st}}\bm{Y}$  if, and only if $\mathbb{E}[\phi(\bm{X})] \leq \mathbb{E}[\phi(\bm{Y})]$
holds for all increasing $\phi: \mathbb{R}^n\rightarrow \mathbb{R}$ provided the expectations  exist \cite{StochasticOrderBook}.

A random variable $X$ is said to be \emph{smaller than another random variable $Y$ in the hazard rate ordering}, denoted by ${X}\leq_{\text{hr}}{Y}$, if $\Pr({X}-t>s|X>t) \leq \Pr({Y}-t>s|Y>t)$ for all~$s\geq0$ and all $t$. A random variable $X$ is said to be \emph{smaller than another random variable $Y$ in the increasing convex ordering}, denoted by ${X}\leq_{\text{icx}}{Y}$, if $\mathbb{E}[\phi({X})] \leq \mathbb{E}[\phi({Y})]$ holds for all increasing functions $\phi: \mathbb{R} \rightarrow \mathbb{R}$ provided the expectations exist \cite{StochasticOrderBook}. 

\subsection{System Model}\label{sec:model}
Consider a  system with $m$ servers, which starts to operate at time $t=0$. 
A sequence of $n$ jobs arrive at time instants $a_1,\ldots,$ $a_n$, where $n$ can be either finite or infinite and $0=a_1\leq a_2\leq\cdots\leq a_n$. 
The $i$-th incoming job, also called job $i$, brings with it a batch of $k_i$ tasks.
Each task is the smallest unit of work that can be assigned to a server. Job $i$ is completed when all $k_i$ tasks of job $i$ are completed. The maximum job size is\footnote{If $n\rightarrow\infty$, then the $\max$ operator in \eqref{eq_0} is replaced by $\sup$.}
\begin{align}\label{eq_0}
k_{\max} = \max_{i=1,\ldots,n}k_i.
\end{align}

\subsubsection{Service Time Distributions} \label{sec_Multivariate}
In practice,  the service times of the tasks are highly random due to many reasons, including resource sharing, network congestion, cache misses, database blocking, etc.  \cite{Dean:2013:Tail},  and the servers may operate at different service speeds because they have different amounts of resources, e.g., CPU, memory, I/O bandwidth \cite{Googletrace2012,C3:2015}. Motivated by this, we assume that the task service times
 are \emph{independent} across the servers and \emph{i.i.d.} across the tasks assigned to the same server.
 Let $X_l$ be a random variable representing the task service time of server $l$. The service rate of server $l$ is $\mu_l= 1/\mathbb{E}[X_l]$, which may vary across the servers. 
We consider the following classes of NBU and NWU task service time distributions.
\begin{definition}
Consider a non-negative random variable $X$ with complementary cumulative distribution function (CCDF)  $\bar{F}(x)=\Pr[X>x]$. Then, $X$ is   \textbf{New-Better-than-Used (NBU)} if for all $t,\tau\geq 0$
\begin{eqnarray}\label{eq_NBU}
\bar{F}(\tau+t)\leq \bar{F}(\tau)\bar{F}(t).
\end{eqnarray}
On the other hand, $X$ is \textbf{New-Worse-than-Used (NWU)} if $\bar{F}$ is \emph{absolutely continuous} and for all $t,\tau\geq 0$
\begin{eqnarray}\label{eq_NWU}
\bar{F}(\tau+t)\geq\bar{F}(\tau)\bar{F}(t).
\end{eqnarray}
\end{definition}

NBU distributions include increasing failure rate (IFR) distributions and log-concave distributions as special cases. 
Examples of NBU distributions include constant service time, shifted exponential distribution, geometrical distribution, Erlang distribution, etc. Recent measurements \cite{addShengboOriginal,Liang2013_2} show that the data downloading time in Amazon AWS can be approximated as a shifted exponential distribution. NWU distributions include the classes of decreasing failure rate (DFR) distributions and log-convex distributions. Examples of NWU distributions include hyperexponential distribution, Pareto type II (Lomax) distribution \cite{Arnold99}, gamma distributions with $m<1$, Weibull distribution with $c<1$, etc. In some systems \cite{Christodoulopoulos2008}, the task service time can be modeled as a hyperexponential distribution. A random variable is both NBU and NWU if, and only if, it is exponential. 

\subsubsection{Queueing Models with Replications}\label{eq_model_heter}

We consider both centralized and distributed queueing models:

\textbf{Centralized queueing model:}
In a centralized queueing system, the jobs arrive at a scheduler and are stored in a job queue, as shown in Fig. \ref{fig1model_central}. A scheduler assigns tasks to the available servers over time.

\ifreport
\begin{figure*}
\centering
\begin{subfigure}[b]{\textwidth}
\centering
\includegraphics[width=0.9\textwidth]{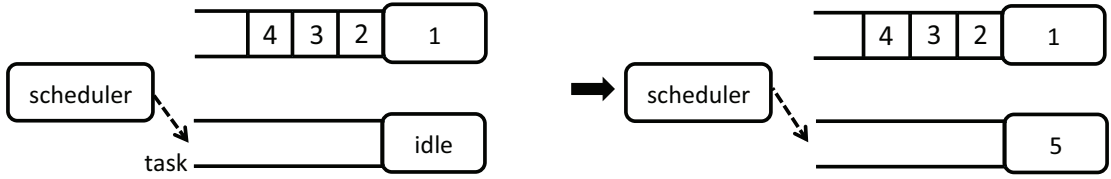}
\caption{\textbf{The Power of $d$ Choices policy:} The scheduler queries the queue lengths of $d=2$ servers, and assigns a task to the server with the shortest queue length.}
\end{subfigure}
\begin{subfigure}[b]{\textwidth}
\centering
\includegraphics[width=0.9\textwidth]{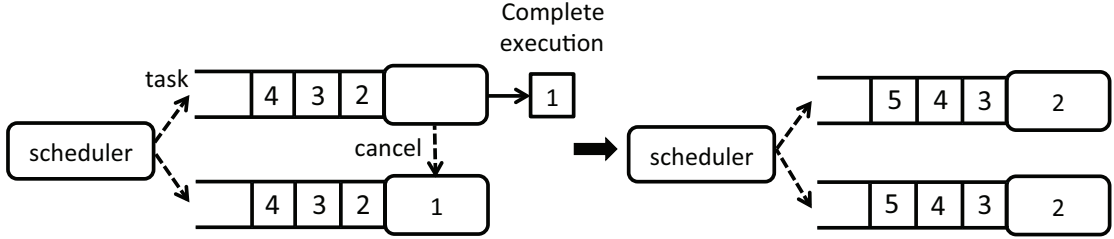}
\caption{\textbf{Cancel-After-Execution policy:} Each task is replicated in both queues and is replicated to both servers. If one copy of task 1 completes execution on one server, a message is sent to cancel the other copy of task 1, which is being executed on the other server.}
\end{subfigure}
\begin{subfigure}[b]{\textwidth}
\centering
\includegraphics[width=0.9\textwidth]{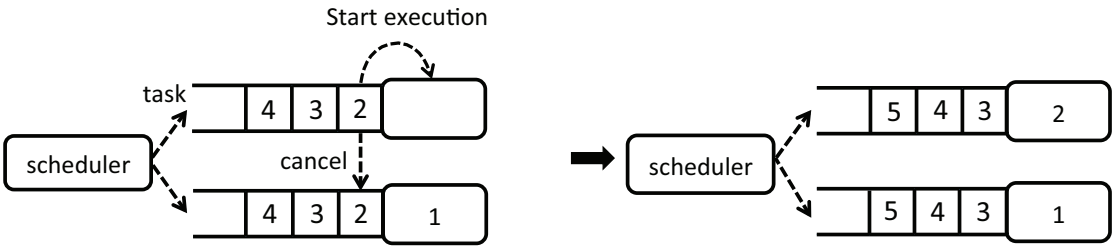}
\caption{\textbf{Cancel-Before-Execution policy:} Each task is replicated in both queues, but can be only assigned to one server. If one copy of task 2 starts execution on one server, a message is sent to cancel the other copy of task 2, which is waiting in the queue on the other server.}
\end{subfigure}
\caption{Scheduling policies for a distributed queueing system, which is a part of the distributed queueing system in Fig. \ref{fig1model_distri}.} \label{fig3_distri_replication}
\end{figure*}
\else
\fi

In order to reduce delay, a task can be \textbf{replicated} on multiple servers, possibly at different time instants.
The task is deemed completed as soon as one copy of the task is completed; after that, the other redundant copies of the task are either executed until completion or cancelled with a certain amount of cancellation delay overhead. 
 We assume that the cancellation overheads are \emph{independent} across the servers and \emph{i.i.d.}~across the  tasks cancelled on the same server. Let  $O_l$ be a random variable representing the cancellation overhead of server $l$.
Denote $\bm{X}=(X_1,\dots,X_m)$ and $\bm{O}=(O_1,\dots,O_m)$, which are  assumed to be \emph{mutually independent}.






\textbf{Distributed queueing model:}
In a distributed queueing system, the jobs arrive at a number of parallel schedulers, and are stored in  
the job queues associated to the schedulers, as depicted in Fig. \ref{fig1model_distri}. Each scheduler independently assigns tasks to the servers.
The servers  are divided into $g$ groups.  Each server has a task queue, which stores the tasks assigned from different schedulers.
The decisions of the schedulers are subject to \textbf{data locality} constraints \cite{taskplacement2011,Sparrow:2013}. More specifically, each task can be only executed by one group of servers, each of which stores one copy of the data necessary for executing the task.

There are multiple ways to assign a task to its required group of the servers. \textbf{The Power of $d$ Choices} \cite{mitzenmacher2001power,vvedenskaya1996queueing} load balancing policy is illustrated in Fig. \ref{fig3_distri_replication}(a), where the scheduler queries the queue lengths of $d$ servers in one group, and selects the server with the shortest queue length. In this policy, it may happen that some queues in the group are empty and the other queues in the group are not, which reduces the efficiency of the system.
Two alternative task assignment policies are depicted in Fig. \ref{fig3_distri_replication}(b) and Fig. \ref{fig3_distri_replication}(c), where the scheduler simultaneously places multiple copies of a task to all the servers in one group, and the servers are allowed to communicate with each other to cancel the redundant task copies. Such an approach is advocated by Google \cite{Dean:2013:Tail}. In Fig. \ref{fig3_distri_replication}(b), when one copy of a task starts execution, a message is sent to cancel the other redundant task copies. This policy is called the \textbf{Cancel-After-Execution} policy. In Fig. \ref{fig3_distri_replication}(c), when one copy of a task completes execution, a message is sent to cancel the other redundant task copies. This policy is called the \textbf{Cancel-Before-Execution} policy, which was also named ``tied-requests'' in \cite{Dean:2013:Tail} and ``late-bindling'' in \cite{Sparrow:2013}.
It was pointed out in \cite{Dean:2013:Tail,Sparrow:2013} that the Cancel-Before-Execution policy
can potentially achieve a better delay performance than the Power of $d$ Choices load balancing policy.

\ifreport
\begin{figure*}
\centering
\begin{subfigure}[b]{\textwidth}
\centering
\includegraphics[width=0.6\textwidth]{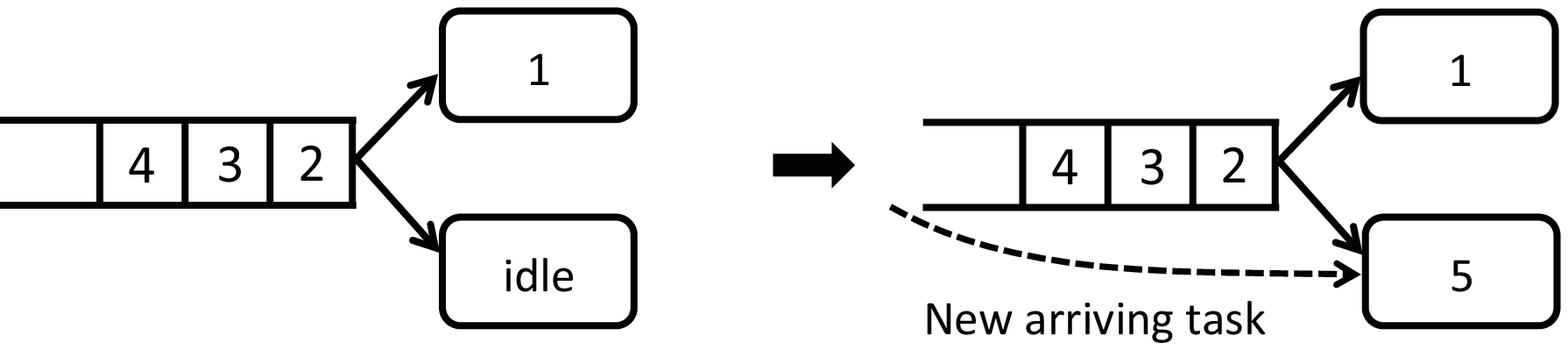}
\caption{\textbf{Assign-When-Enqueueing policy:} Each task is labelled to be served by one of the two servers when it arrives, and cannot be re-assigned to the other server. In this example, tasks 1-4 are labelled to be served by one server. If the other server is idle, only a new arriving task can be assigned to the idle server.}
\end{subfigure}
\begin{subfigure}[b]{\textwidth}
\centering
\includegraphics[width=0.6\textwidth]{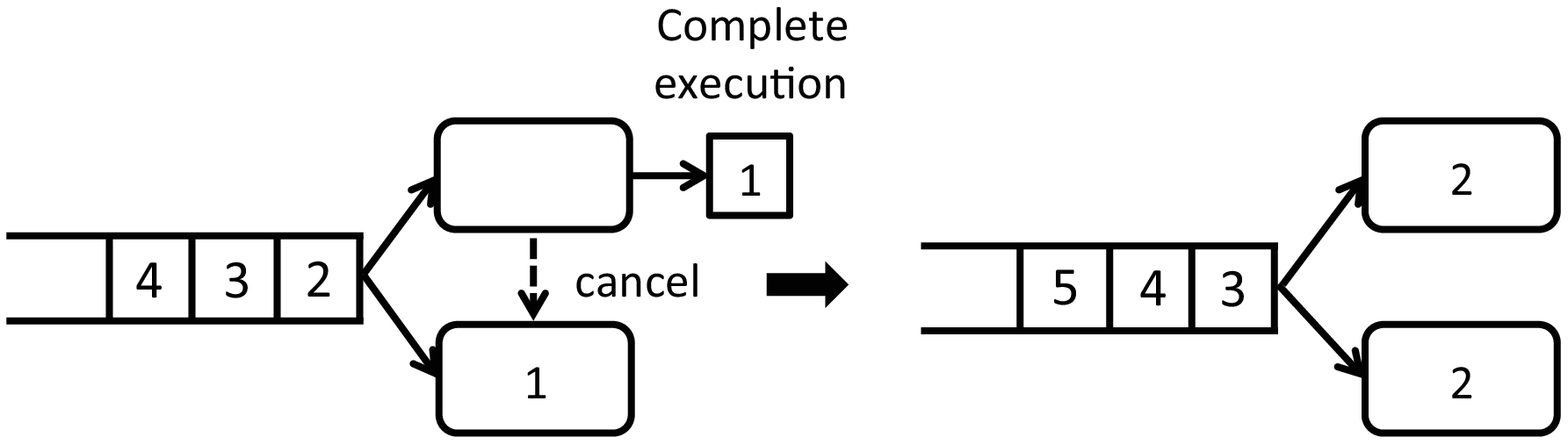}
\caption{\textbf{Replication policy:} Each task is replicated on both servers. If one copy of task 1 completes execution on one server, a message is sent to cancel the other copy of task 1, which is being executed on the other server.}
\end{subfigure}
\begin{subfigure}[b]{\textwidth}
\centering
\includegraphics[width=0.6\textwidth]{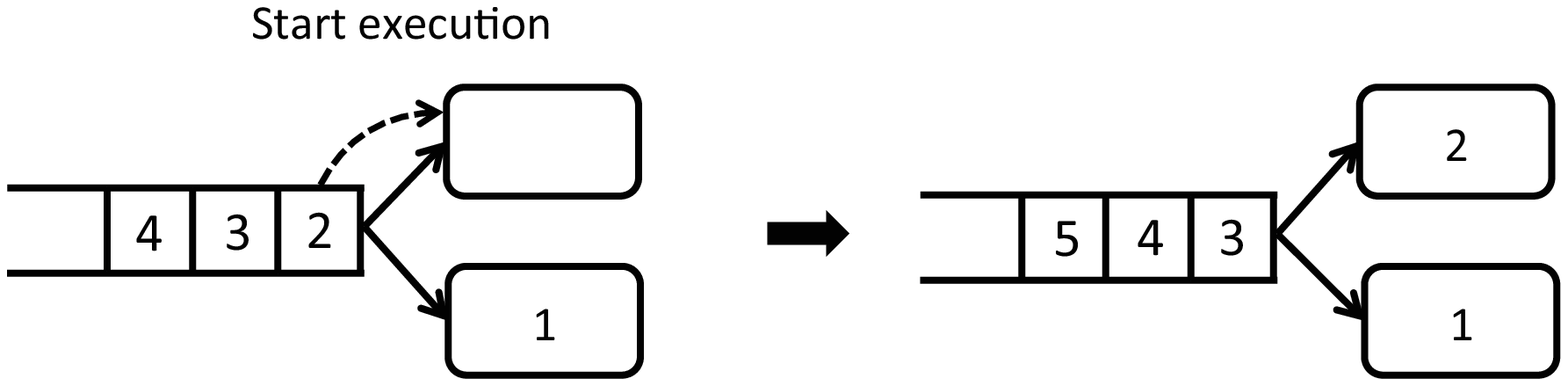}
\caption{\textbf{No-replication policy:} Each task can be only assigned to one server. If one copy of task 2 starts execution on one server, a message is sent to cancel the other copy of task 2, which is waiting in the queue on the other server.}
\end{subfigure}
\caption{Scheduling policies for a centralized queueing system. Each policy is equivalent to a scheduling policy in Fig. \ref{fig3_distri_replication} for a distributed queueing system.} \label{fig2central_replication}
\end{figure*}
\else
\fi

We note that each scheduling policy in Fig. \ref{fig3_distri_replication} for distribution queueing systems has one equivalent scheduling policy for centralized queueing systems.
In particular, The Power of $d$ Choices load balancing policy in Fig. \ref{fig3_distri_replication}(a) is equivalent to one instant of the \textbf{Assign-When-Enqueueing policy} in Fig. \ref{fig2central_replication}(a),
the Cancel-After-Execution policy in Fig. \ref{fig3_distri_replication}(b) is equivalent to the \textbf{Replication policy} in Fig. \ref{fig2central_replication}(b), and the Cancel-Before-Execution policy in Fig. \ref{fig3_distri_replication}(c) is equivalent to the \textbf{No-Replication policy} in Fig. \ref{fig2central_replication}(c). In this paper, we will establish a unified framework to study delay-optimal scheduling of replications in centralized and distributed queueing systems.

\subsection{Scheduling Policies}
A scheduling policy, denoted by $\pi$, determines the task assignments, replications, and cancellations in the system.
We consider the class of \textbf{causal} policies, in which scheduling decisions are made based on the history and current state of the system; the realization of task service time is unknown until the task is completed (unless the service time is deterministic). 
In practice, service preemption is costly and may lead to complexity and reliability issues \cite{Sparrow:2013,Borg2015}. Motivated by this, we assume that \textbf{task-level preemption is not allowed}. Hence, if a server starts to process a task, it must complete or cancel this task before switching to process another task. 
We use $\Pi$ to denote the set of \textbf{causal and non-preemptive} policies. Let us define several types of policies within $\Pi$:

A policy is said to be \textbf{anticipative}, if it has access to the parameters $(a_i,{k}_i,d_i)$ of future arriving jobs (but not other future information). For periodic and pre-planned services, future job arrivals can be predicted in advance. To cover these scenarios, we abuse the definition of causal policies a bit and include anticipative policies into the policy space $\Pi$. However, it should be emphasized that the policies that we propose in this paper are not anticipative. 

\begin{definition}
A task is termed \textbf{remaining} if it is either stored in the queue or being executed by the servers, and is termed \textbf{unassigned} if it is stored in the queue and not being executed by any server.
\end{definition}

A policy is said to be \textbf{work-conserving}, if no server is idle when there are unassigned tasks waiting in the queue.

The goal of this paper is \emph{to design low-complexity non-anticipative scheduling policies that are (near) delay-optimal among all policies in $\Pi$, even compared to the anticipative policies with knowledge about future arriving jobs}.

\subsection{Delay Metrics}\label{sec:metrics}
Each job $i$ has a \emph{due time} $d_i\in[0,\infty)$, also called \emph{due date}, which is the time that job $i$ is promised to be completed \cite{michael2012book}. 
Completion of a job after its due time is allowed, but then a penalty is incurred. Hence, the due time can be considered as a soft deadline. 

For each job $i$, $C_i$ is the job completion time, $D_i= C_i-a_i$ is the delay, $L_i = C_i-d_i$ is the {lateness} after the due time $d_i$, and $T_i = \max[C_i-d_i,0]$ is the  {tardiness} (or positive lateness). 
Define  vectors $\bm{a} =(a_1,\ldots,a_n)$,
$\bm{d} =(d_1,\ldots,d_n)$, $\bm{C}=(C_1,$ $\ldots,C_n)$, $\bm{D}=(D_1,\ldots, D_n)$,  $\bm{L}=(L_1,\ldots,$ $ L_n)$, and $\bm{C}_{\uparrow}\! =\! (C_{(1)},\ldots,C_{(n)})$. 
Let $\bm{c}=(c_1,$ $\ldots,c_n)$ and $\bm{c}_{\uparrow}\! =\! (c_{(1)},\ldots,c_{(n)})$, respectively, denote the realizations of $\bm{C}$ and $\bm{C}_{\uparrow}$. All these quantities are functions of the  scheduling policy $\pi$.


%

Several important delay metrics are introduced in the following: For any policy $\pi$, the \textbf{average delay} ${D}_{\text{avg}}: \mathbb{R}^n\rightarrow \mathbb{R}$ is defined by\footnote{If $n\rightarrow\infty$, then a $\limsup$ operator is enforced on the RHS of \eqref{eq_10}, \eqref{eq_13}, \eqref{eq_14}, and the $\max$ operator in \eqref{eq_11} and \eqref{eq_12} is replaced by $\sup$.
}
\begin{align}
{D}_{\text{avg}}(\bm{C}(\pi))=\frac{1}{n}\sum_{i=1}^{n}\left[C_i(\pi)-a_i\right]. \label{eq_10}
\end{align}
In addition, the \textbf{mean square of tardiness} ${T}_{\text{ms}}: \mathbb{R}^n\rightarrow \mathbb{R}$ is 
\begin{align}
T_{\text{ms}}(\bm{C}(\pi)) \!= \!\frac{1}{n}\sum_{i=1}^{n}\! \max[L_i(\pi),0]^2\!=\!\frac{1}{n}\sum_{i=1}^{n}\! \max\left[C_i(\pi)\!-\!d_i,0\right]^2\!\!.\label{eq_13}
\end{align}
If $d_i=a_i$, $T_{\text{ms}}$ becomes the \textbf{mean square of delay}  ${D}_{\text{ms}}$, i.e.,
\begin{align}
D_{\text{ms}}(\bm{C}(\pi)) \!= \frac{1}{n}\sum_{i=1}^{n} D_i^2(\pi)=\frac{1}{n}\sum_{i=1}^{n} \left[C_i(\pi)-a_i\right]^2,\label{eq_14}
\end{align}
since $C_i(\pi)\geq a_i$ for all $i$ and $\pi$. 

In many systems, fairness is an important aspect of the quality-of-service. We define two delay metrics related to   
\emph{min-max fairness}. The \textbf{maximum lateness}  ${L}_{\max}: \mathbb{R}^n\rightarrow \mathbb{R}$ is defined by
\begin{align}
L_{\max}(\bm{C}(\pi)) \!= \!\max_{i=1,2,\ldots,n} L_i(\pi)=\max_{i=1,2,\ldots,n} \left[C_i(\pi)-d_i\right], \label{eq_11}
\end{align}
if $d_i=a_i$, $L_{\max}$ reduces to the \textbf{maximum delay}  ${D}_{\max}$, i.e.,
\begin{align}
D_{\max}(\bm{C}(\pi))\!=\! \max_{i=1,2,\ldots,n} D_i(\pi)= \max_{i=1,2,\ldots,n} \left[C_i(\pi)-a_i\right].\!\!\label{eq_12}
\end{align}
In addition, two delay metrics related to \emph{proportional fairness} \cite{Kushner2004} are defined by
\begin{align}
&T_{\text{PF}}(\bm{C}(\pi)) \!= \! -\sum_{i=1}^n \log (\max[L_i(\pi),0]+\epsilon)=-\sum_{i=1}^n \log(\max[C_i(\pi)-d_i,0]+\epsilon),\\
&D_{\text{PF}}(\bm{C}(\pi)) \!= \! - \sum_{i=1}^n \log (D_i(\pi)+\epsilon)=-\sum_{i=1}^n \log(C_i(\pi)-a_i+\epsilon),
\end{align}
where $\epsilon$ is a positive number, which can be as small as we wish.

In general, a \textbf{delay metric} can be expressed as a function $f(\bm{C}(\pi))$ of the job completion times $\bm{C}(\pi)$, where $f: \mathbb{R}^n\rightarrow \mathbb{R}$ is increasing. In this paper, we consider three classes of delay metric functions:
\begin{align}
\mathcal{D}_{\text{sym}} &= \{f : f \text{ is symmetric and increasing}\},\nonumber\\
\mathcal{D}_{\text{Sch-1}} &= \{f:  f (\bm{x}\!+\!\bm{d}) \text{ is Schur convex and increasing in $\bm{x}$}\},\nonumber\\
\mathcal{D}_{\text{Sch-2}} &= \{f:  f (\bm{x}\!+\!\bm{a}) \text{ is Schur convex and increasing in $\bm{x}$}\}.\nonumber\end{align}
For each $f\in \mathcal{D}_{\text{Sch-1}}$, the delay metric $f(\bm{C}(\pi)) = f[(\bm{C}(\pi)-\bm{d})+\bm{d}]= f[\bm{L}(\pi)+\bm{d}]$ is Schur convex in the lateness vector $\bm{L}(\pi)$. Similarly, for each $f\in \mathcal{D}_{\text{Sch-2}}$, the delay metric $f(\bm{C}(\pi))= f[(\bm{C}(\pi)-\bm{a})+\bm{a}]= f[\bm{D}(\pi)+\bm{a}]$ is Schur convex in the delay vector $\bm{D}(\pi)$. Furthermore, every convex and symmetric function is Schur convex. Using these properties, we can get
\begin{align}
&{D}_{\text{avg}}\in\mathcal{D}_{\text{sym}}\cap \mathcal{D}_{\text{Sch-1}} \cap \mathcal{D}_{\text{Sch-2}},\nonumber \\
&L_{\max}\in\mathcal{D}_{\text{Sch-1}}, D_{\max}\in\mathcal{D}_{\text{Sch-2}},\nonumber\\
&T_{\text{ms}}\in\mathcal{D}_{\text{Sch-1}}, ~D_{\text{ms}}\in\mathcal{D}_{\text{Sch-2}}.\nonumber
\\
&T_{\text{PF}}\in\mathcal{D}_{\text{Sch-1}}, ~D_{\text{PF}}\in\mathcal{D}_{\text{Sch-2}}.\nonumber
\end{align}



\subsection{Delay Optimality and Its Approximation}\label{sec_optimality}

Define $\mathcal{I}=\{n,(a_i,{k}_i,d_i)_{i=1}^n\}$ as the parameters of the jobs, which include the number, batch sizes, arrival times, and due times of the jobs.
The job parameters $\mathcal{I}$ and random task service times are determined by two external processes, which are \emph{mutually independent} and do not change according to the scheduling policy adopted in the system.
For delay metric function $f$ and policy space $\Pi$, a policy $P\in \Pi$ is said to be \textbf{delay-optimal in stochastic ordering}, if one of the following conditions is satisfied: 
\begin{itemize}
\item[1.] For all $\pi\in \Pi$ and $\mathcal{I}$ 
\begin{align}\label{eq_optimal}
[f (\bm{C}(P))|P,\mathcal{I}]\leq_{\text{st}} [f (\bm{C}(\pi))|\pi,\mathcal{I}];
\end{align}
\item[2.] for all $\mathcal{I}$ and $t\in[0,\infty)$ 
\begin{align}
\!\!\!\!\!\!\Pr[f (\bm{C}(P))>t |P,\mathcal{I}]=\min_{\pi\in\Pi}\Pr[f (\bm{C}(\pi))>t |\pi,\mathcal{I}];
\end{align}
\item[3.] for all $\mathcal{I}$
\begin{align}\label{eq_optimal1}
\mathbb{E}[\phi \circ f (\bm{C}(P))|P,\mathcal{I}]=\min_{\pi\in\Pi}\mathbb{E}[\phi \circ f (\bm{C}(\pi))|\pi,\mathcal{I}]
\end{align}
holds for all increasing function $\phi: \mathbb{R}\rightarrow \mathbb{R}$ provided the conditional expectations in \eqref{eq_optimal1} exist.
\end{itemize}
By the definition of stochastic ordering \cite{StochasticOrderBook}, these three conditions are equivalent.
For notational simplicity, we will omit to mention that policy $\pi$ (or policy $P$) is adopted in the system as a condition of the delay performance in the rest of the paper.

In many system settings, delay optimality is extremely difficult to achieve,
even with respect to some definitions of delay optimality weaker than \eqref{eq_optimal}. This motivated us to study whether there exist policies that can come close to delay optimality. We will show that in many scenarios, near delay optimality can be achieved in the following sense:

\ifreport
\begin{figure}
\centering
\includegraphics[width=0.5\textwidth]{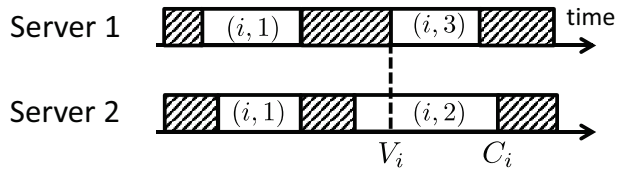} \caption{An illustration of $V_i$ and $C_i$. There are 2 servers, and job $i$ has 3 tasks denoted by $(i,1)$, $(i,2)$, $(i,3)$. Task $(i,1)$ is replicated on both servers, tasks $(i,2)$ and $(i,3)$ are assigned to the two servers separately. By time $V_i$, all tasks of job $i$ have entered the servers. By time $C_i$, all tasks of job $i$ have completed service. Therefore, $V_i\leq C_i$.}
\label{V_i}
\end{figure}
\fi

{Define $V_i$ as the earliest time that all tasks of job $i$ have started service. In other words, all tasks of job $i$ are either completed or under service at time $V_i$. One illustration of $V_i$ is provided in Fig. \ref{V_i}, from which it is easy to see 
\begin{align}\label{eq_V_i_small}
V_i\leq C_i.
\end{align}
Denote $\bm{V}=(V_1,\ldots, V_n)$, $\bm{V}_{\uparrow} \!=\! (V_{(1)},\ldots,V_{(n)})$. Let $\bm{v}=(v_1,\ldots, v_n)$ and $\bm{v}_{\uparrow} \!=\! (v_{(1)},\ldots,v_{(n)})$ be the realizations of $\bm{V}$ and $\bm{V}_{\uparrow}$, respectively. All these quantities are functions of the  scheduling policy $\pi$.
A policy $P\in \Pi$ is said to be \textbf{near delay-optimal in stochastic ordering}, if one the following three conditions is satisfied:
\begin{itemize}
\item[1.] For all $\pi\in \Pi$ and $\mathcal{I}$ 
\begin{align}\label{eq_nearoptimal}
[f (\bm{V}(P))|\mathcal{I}]\leq_{\text{st}} [f (\bm{C}(\pi))|\mathcal{I}];
\end{align}
\item[2.] for all $\mathcal{I}$ and $t\in[0,\infty)$ 
\begin{align}
\!\!\!\!\!\!\Pr[f (\bm{V}(P))>t |\mathcal{I}]&\leq \min_{\pi\in\Pi}\Pr[f (\bm{C}(\pi))>t |\mathcal{I}]\leq \Pr[f (\bm{C}(P))>t |\mathcal{I}];
\end{align}
\item[3.] for all $\mathcal{I}$
\begin{align}\label{eq_nearoptimal1}
\mathbb{E}[\phi \!\circ\! f (\bm{V}(P))|\mathcal{I}] &\leq\min_{\pi\in\Pi}\mathbb{E}[\phi \!\circ\! f (\bm{C}(\pi))|\mathcal{I}] \leq\mathbb{E}[\phi \!\circ\! f (\bm{C}(P))|\mathcal{I}]
\end{align}
holds for all increasing function $\phi: \mathbb{R}\rightarrow \mathbb{R}$ provided the conditional expectations in \eqref{eq_nearoptimal1} exist.
\end{itemize}
There exist many ways to approximate \eqref{eq_optimal}-\eqref{eq_optimal1}, and obtain various forms of near delay optimality.
We find that the form of near delay optimality in \eqref{eq_nearoptimal}-\eqref{eq_nearoptimal1} is convenient, because it is analytically 
provable and leads to tight sub-optimal delay gap, as we will see in the subsequent sections.

%% file: analysis.tex
\section{Replications in Centralized Queueing Systems}\label{sec_analysis}


In this section, we provide some delay optimality and near delay optimality results for replications in centralized queueing systems. The proofs of these results will be provided in the Appendix \ref{sec_proofmain} by using a  unified sample-path method.



\subsection{Average Delay and Related Delay Metrics}\label{sec_average_delay}

\begin{algorithm}[h]
\SetKwData{NULL}{NULL}
\SetCommentSty{small}
$Q:=\emptyset$\tcp*[r]{$Q$ is the set of  jobs in the queue}
\While{the system is ON} {
\If{job $i$ arrives}{
$\xi_i:=k_i$\tcp*[r]{job $i$ has $\xi_i$ remaining tasks~~~~~~~}
$\gamma_i:=k_i$\tcp*[r]{job $i$ has $\gamma_i$ unassigned tasks~~~~~~}
$Q:=Q \cup \{i\}$\;
}

\While{$Q\neq\emptyset$ and there are idle servers}{
Pick any idle server $l$\;
\uIf{$\sum_{i\in Q} \gamma_i> 0$}{~~~~~~~~\tcp*[h]{There exist unassigned tasks}

$j :=\arg\min\{\gamma_i : i\in Q, \gamma_i>0\}$\;
Allocate an unassigned task of job $j$ on server $l$\;
$\gamma_j:=\gamma_j - 1$\;
}
\Else(~~~~~~\tcp*[h]{All tasks are under service})
{
Pick any task and replicate it on server $l$\;\label{step}
}
}

\If{a task of job $i$ is completed}{
\For{each server $l$ processing a redundant copy of this task}{
\uIf{the time to cancel the task is shorter than the time to complete the task in the hazard rate ordering}{
Cancel this redundant task\;
}\Else{
complete this redundant task\;
}
}
$\xi_{i} := \xi_{i} -1$\;
\textbf{if} $\xi_{i}=0$ \textbf{then} {$Q:=Q/ \{i\}$\;   }
}

}
\caption{ Fewest Unassigned Tasks first with No Idleness Replication (FUT-NIR).}\label{alg1}
\end{algorithm}

If the task service times are NBU and the delay metric is within $\mathcal{D}_{\text{sym}}$ (including the average delay $D_{\text{avg}}$), we propose a class of scheduling policies called \textbf{Fewest Unassigned Tasks first with Low-Priority Replication (FUT-LPR)}. To understand this class of policies, let us introduce some definitions:

\begin{definition}
A scheduling policy is said to follow the \textbf{Fewest Unassigned Task (FUT)} first discipline, if each task assigned to the servers is from the job with the fewest unassigned tasks whenever the queue is not empty (there exist unassigned tasks in the queue). 
\end{definition}

\begin{definition} 
A scheduling policy is said to follow the \textbf{Low-Priority Replication (LPR)} discipline, if it is work-cons- erving and satisfies the following two principles:
\end{definition}

\begin{itemize}
\item[1.] {Task Replication:} If the queue is not empty (there exist unassigned tasks in the queue), then no replication is allowed; otherwise, if the queue is empty (all the tasks are under service), one can replicate the tasks arbitrarily. 

\item[2.] {Task Cancellation:} If the time to cancel a task $O_l$ is long than the remaining service time to complete the task $R_l$ in  hazard rate ordering, i.e., $R_l\leq_{\text{hr}} O_l$,
then choose to complete the task without cancellation; otherwise, one can choose either to cancel or to complete the task.


\end{itemize}

The LPR policies offer flexible choices for task replication and cancellation operations, which are quite convenient in practice.
Examples of LPR policies include the \textbf{No Replication (NR)} policy in which no replication is allowed at all, and the \textbf{No Idleness Replication (NIR)} policy which satisfies: If the queue is not empty (there exist unassigned tasks in the queue), then no replication is allowed; otherwise, if the queue is empty (all the tasks are under service), 
each idle server is allocated to process a replicated copy of \emph{any} remaining task. Therefore, {no server is idle in the NIR policy until all jobs are completed}, and hence the name. 

A scheduling policy belongs to the class of \textbf{Fewest Unassigned Tasks first with Low-Priority Replication (FUT-LPR)} policies if it simultaneously satisfies the FUT and LPR disciplines. Examples of FUT-LPR policies include the \textbf{Fewest Unassigned Tasks first policy with No Idleness Replication (FUT-NIR)} policy which is illustrated in Algorithm \ref{alg1}, and the \textbf{Fewest Unassigned Tasks first policy with No Replication (FUT-NR)} policy which can be obtained from Algorithm \ref{alg1} by removing Steps 15-16, 20-26. The delay performance of any instance of the FUT-LPR policies is characterized in the following theorem:

\begin{theorem}\label{thm1}
If the task service times are NBU, independent across the servers, i.i.d.~across the tasks assigned to the same server, then for all $\bm{O}\geq\bm{0}$, \emph{$f\in\mathcal{D}_{\text{sym}}$},  ${\pi\in\Pi}$, and $\mathcal{I}$ 
\emph{
\begin{align}\label{eq_delaygap1}
&[f(\bm{V}(\text{FUT-LPR}))|\mathcal{I}] \leq_{\text{st}} \left[f(\bm{C}(\pi))|\mathcal{I}\right].
\end{align}}
\end{theorem}
\ifreport
\fi

Let us characterize the sub-optimality delay gap of the FUT-LPR policies. For mean average delay, i.e., $f(\cdot)=D_{\text{avg}}(\cdot)$, it follows from Theorem \ref{thm1} that 
\begin{align}
\!\!\!\!\mathbb{E}[D_{\text{avg}}(\bm{V}(\text{FUT-LPR}))|\mathcal{I}] & \leq 
\min_{\pi\in\Pi}\mathbb{E}\left[D_{\text{avg}}(\bm{C}(\pi))|\mathcal{I}\right] \nonumber\\
&\leq \mathbb{E}[D_{\text{avg}}(\bm{C}(\text{FUT-LPR}))|\mathcal{I}].\!\! \label{eq_delaygap1_2}
\end{align}
The difference between the LHS and RHS of \eqref{eq_delaygap1_2} is 
\begin{align}\label{eq_delaygap1_1_1}
& \mathbb{E}[D_{\text{avg}}(\bm{C}(\text{FUT-LPR}))-D_{\text{avg}}(\bm{V}(\text{FUT-LPR}))|\mathcal{I}]\nonumber\\
=& \mathbb{E}\bigg[\frac{1}{n}\sum_{i=1}^n\big({C}_i(\text{FUT-LPR})-V_i(\text{FUT-LPR})\big)\bigg|\mathcal{I}\bigg].\!
\end{align}
Recall that $m$ is the number of servers, and $k_i$ is the number of tasks in job $i$. At time $V_i(\text{FUT-LPR})$, all tasks of job $i$ have started service. Hence, if $k_i>m$,  then job $i$ has at most $m$ incomplete tasks that are under service at time $V_i(\text{FUT-LPR})$; if $k_i\leq m$, then job $i$ has at most $k_i$ incomplete tasks that are under service at time $V_i(\text{FUT-LPR})$. Therefore, in the FUT-LPR policies, at most $k_i\wedge m = \min\{k_i,m\}$ tasks of job $i$ are completed during the time interval $[V_i(\text{FUT-LPR}), {C}_i(\text{FUT-}\text{LPR})]$. Using this and the property of NBU distributions, we can obtain

\begin{theorem}\label{lem7_NBU}
Let $\mathbb{E}[X_l]= 1/\mu_l$, and without loss of generality $\mu_1\leq \mu_2\leq \ldots\leq \mu_m$. If the task service times are NBU, independent across the servers, i.i.d.~across the tasks assigned to the same server, then  for all $\bm{O}\geq\bm{0}$ and $\mathcal{I}$
\emph{\begin{align}\label{eq_gap}
&\mathbb{E}[D_{\text{avg}}(\bm{C}(\text{FUT-LPR}))|\mathcal{I}] - \min_{\pi\in\Pi}\mathbb{E}\left[D_{\text{avg}}(\bm{C}(\pi))|\mathcal{I}\right]  \nonumber\\
\leq&\frac{1}{n}\sum_{i=1}^n\sum_{l=1}^{k_i \wedge m } \frac{1}{\sum_{j=1}^l\mu_j}\leq \frac{\ln(k_{\max}\wedge m)+1}{\mu_1},
\end{align}}
where $k_{\max}$ is the maximum job size in \eqref{eq_0} and $x\wedge y=\min\{x,y\}$.
\end{theorem}

Theorem \ref{thm1} and Theorem \ref{lem7_NBU} tell us that {for arbitrary number, batch sizes, arrival times, and due times of the jobs, as well as arbitrary cancellation overheads of the tasks, the class of FUT-LPR policies is near delay-optimal for minimizing the mean average delay within the policy space $\Pi$, even compared to the anticipative policies in $\Pi$ that can predict the parameters of future arriving jobs. 

If $\mathbb{E}[X_l]\leq 1/\mu$ for $l=1,\ldots,m$, then  the sub-optimality delay gap of the FUT-LPR policies is of the order $O(\ln(k_{\max}\wedge m))/{\mu}$. As the number of servers $m$ increases, this  sub-optimality delay gap is upper bounded by $[\ln(k_{\max})+1]/{\mu}$ which is independent of $m$. 


\begin{algorithm}[h]
\SetKwData{NULL}{NULL}
\SetCommentSty{small}
$Q:=\emptyset$\tcp*[r]{the set of  jobs in the queue}
\While{the system is ON} {
\If{job $i$ arrives}{
$\gamma_i:=k_i$\tcp*[r]{job $i$ has $\gamma_i$ unassigned tasks~~~~~~~}
$Q:=Q \cup \{i\}$\;
}

\While{$Q\neq\emptyset$ and all servers are idle}{
$j :=\arg\min\{\gamma_i: i\in Q\}$\;
Replicate a task of job $j$ on all servers\;
}

\If{a task of job $i$ is completed}{
Cancel the remaining $m-1$ replicas of this task\;
$\gamma_{i} := \gamma_{i} -1$\;
\textbf{if} $\gamma_{i}=0$ \textbf{then} {$Q:=Q/ \{i\}$\;   }
}
}
\caption{ Fewest Unassigned Tasks first with Replication (FUT-R).}\label{alg2}
\end{algorithm}

When there are multiple servers and multiple job classes, establishing tight additive bounds on the gap from the optimal delay performance is extremely difficult and has met with little success. In \cite{Weiss:1992,Weiss:1995}, closed-form upper bounds on the sub-optimality gaps of the Smith's rule and the Gittin's index rule were established for the cases that all jobs arrive at time zero.
In \cite[Corollary 2]{Dacre1999}, the authors studied the optimal control in multi-server systems with stationary arrivals of multiple classes of jobs, and used the achievable region method to obtain an additive sub-optimality delay gap, which is of the order $O(m k_{\max})$. The model and methodology in \cite{Weiss:1992,Weiss:1995,Dacre1999} are significantly different from those in this paper.

If the task service times are NWU, we propose a policy called  \textbf{Fewest Unassigned Tasks first with Replication (FUT-R)}, 
which simultaneously satisfies the FUT discipline and the following replication (R) discipline.
\begin{definition} 
A scheduling policy is said to follow the \textbf{Replication (R)} discipline, if it is work-conserving and satisfies the following two principles:
\begin{itemize}
\item[1.] {Task Replication:} When a task is assigned, it is replicated on all $m$ servers.

\item[2.] {Task Cancellation:} The cancellation overhead $\bm{O}$ is assumed to be zero, such that if one  task copy is completed on one server, the remaining $m-1$ replicated copies of this task are cancelled immediately.

\end{itemize}
\end{definition}

The FUT-R policy is described in Algorithm \ref{alg2}. Its delay performance is characterized as follows:


\begin{theorem}\label{thm2}
If (i) $k_1\leq k_2\leq \ldots\leq k_n$, (ii) the task service times are NWU, independent across the servers, i.i.d. across the tasks assigned to the same server, (iii) $\bm{O}=\bm{0}$, then for all \emph{$f\in\mathcal{D}_{\text{sym}}$},  \emph{$\pi\in\Pi$}, and $\mathcal{I}$
\emph{
\begin{align}\label{eq_delaygap3}
[f(\bm{C}(\text{FUT-R}))|\mathcal{I}] \leq_{\text{st}} [f(\bm{C}(\pi))|\mathcal{I}].
\end{align}}
\end{theorem}
\ifreport
\begin{proof}
See Appendix \ref{app2}.
\end{proof}
\fi

Hence, policy FUT-R is delay-optimal in stochastic ordering under the conditions of Theorem \ref{thm2}. 
One special case of Theorem \ref{thm2} was  obtained in Theorem 3.2 of  \cite{Righter2008} and Theorem 3 of \cite{shah-Allerton-2013}, where each job has a single task, i.e., $k_1=k_2=\ldots=k_n=1$.
We note that compared to the traditional definition of NWU distributions in reliability theory \cite[p. 1]{StochasticOrderBook}, the definition in \eqref{eq_NWU} has one additional condition on the absolute continuity of $\bar{F}$. This condition is introduced to ensure that the probability for any two servers to complete task executions at the same time is zero. In particular, if two servers complete task executions at the same time, it might be better to assign these two servers to process two distinct tasks than to replicate two copies of a task on these two servers. Similar phenomena were reported in \cite{BinInfoCom2016}.

If  the job sizes are arbitrarily given (i.e., Condition (i) of Theorem \ref{thm2} is removed), and the task service times are  exponential, the delay performance of the FUT-R policy is characterized as follows.

\begin{theorem}\label{thm2_exp}
If (i) the task service times are exponential, independent across the servers, i.i.d.~across the tasks assigned to the same server, (ii) $\bm{O}=\bm{0}$, then for all \emph{$f\in\mathcal{D}_{\text{sym}}$},  \emph{$\pi\in\Pi$}, and $\mathcal{I}$
\emph{
\begin{align}\label{eq_delaygap3_exp}
[f(\bm{V}(\text{FUT-R}))|\mathcal{I}] \leq_{\text{st}} [f(\bm{C}(\pi))|\mathcal{I}].
\end{align}}
\!\!If $E[X_l]= 1/\mu_l$, then for all $\mathcal{I}$
\emph{\begin{align}\label{eq_gap_NWU}
\mathbb{E}[D_{\text{avg}}(\bm{C}(\text{FUT-R}))|\mathcal{I}] - \min_{\pi\in\Pi}\mathbb{E}\left[D_{\text{avg}}(\bm{C}(\pi))|\mathcal{I}\right]  \leq\frac{1}{\sum_{j=1}^m\mu_j}.
\end{align}}
\end{theorem}

Hence, if $E[X_l]\leq 1/\mu$ for all $l$, the sub-optimality gap of the FUT-R policy diminishes to zero at a speed of $O(1/m)$ as $m\rightarrow\infty$. Note that the sub-optimal delay gaps in Theorem \ref{lem7_NBU} and Theorem \ref{thm2_exp} are independent of the job parameters, and hence remain constant for any traffic load.
Because exponential distribution is both NBU and NWU, Theorems \ref{thm1}-\ref{thm2_exp} are all satisfied for exponential service time distributions.

It is important to emphasize that the FUT discipline is a nice approximation of the Shortest Remaining Processing Time (SRPT) first discipline  \cite{Schrage68,Smith78}: The FUT discipline  utilizes the number of unassigned tasks of a job to approximate the remaining processing time of this job, and is within a small additive sub-optimality gap from the optimum for minimizing the mean average delay in the scheduling problems  that we consider.

\subsection{Maximum Lateness and Related Delay Metrics}\label{sec_maxlate}
Next, we consider the maximum lateness $L_{\max}$ and the delay metrics in $\mathcal{D}_{\text{Sch-1}}$. 
When the task service times are NBU, we propose a class of policies named \textbf{Earliest Due Date first with Low-Priority Replication (EDD-LPR)}, which can be obtained by combining the following EDD discipline and the LPR discipline. 

\begin{definition}
A scheduling policy is said to follow the \textbf{Earliest Due Date (EDD)} first discipline, if each task assigned to the servers is from the job with the earliest due date  whenever the queue is not empty (there exist unassigned tasks in the queue). 
\end{definition}

Two instances of the EDD-LPR policies are policy \textbf{Earliest Due Date first with No Idleness Replication (EDD-NIR)} and policy \textbf{Earliest Due Date first with No Replication (EDD-NR)}. Policy EDD-NIR can be obtained from Algorithm \ref{alg1} by revising Step 12 as $j :=\arg\min\{d_i: i\in Q, \gamma_i>0\}$. Policy EDD-NR can be obtained by revising Step 12 of Algorithm \ref{alg1} and further removing Steps 15-16, 20-26. The delay performance of policy EDD-LPR is characterized in the following theorem:
\begin{theorem}\label{thm3}
If the task service times are NBU, independent across the servers, i.i.d.~across the tasks assigned to the same server, then for all $\bm{O}\geq\bm{0}$, ${\pi\in\Pi}$, and $\mathcal{I}$ 
\emph{
\begin{align}\label{eq_delaygap1_thm3}
&[L_{\max}(\bm{V}(\text{EDD-LPR}))|\mathcal{I}] \leq_{\text{st}} \left[L_{\max}(\bm{C}(\pi))\left.\right|\mathcal{I}\right].
\end{align}}
\end{theorem}

If the job parameters $\mathcal{I}$ satisfy certain conditions, Theorem \ref{thm3} can be generalized to all delay metrics in {$\mathcal{D}_{\text{sym}}\cup\mathcal{D}_{\text{Sch-1}}$}.

\begin{theorem}\label{coro_thm3_1}
If $k_1=\ldots=k_n=1$ (or $d_1\leq d_2\leq \ldots\leq d_n$ and $k_1\leq k_2 \leq \ldots\leq k_n$) and \emph{$L_{\max}$} is replaced by any \emph{$f\in\mathcal{D}_{\text{sym}}\cup\mathcal{D}_{\text{Sch-1}}$}, Theorem  \ref{thm3} still holds.
\end{theorem}

When the task service times are NWU, we propose a policy called  \textbf{Earliest Due Date first with Replication (EDD-R)}. This policy is similar with the FUT-R policy, except that in the EDD-R policy, all servers are allocated to process $m$ replicated copies of a task from the job with the earliest due time. The EDD-R policy can be obtained from Algorithm \ref{alg2} by revising Step 8 as $j :=\arg\min\{d_i: i\in Q\}$. The delay performance of the EDD-R policy is provided as follows.

\begin{theorem}\label{thm4}
If (i) $d_1\leq d_2\leq \ldots\leq d_n$, (ii) the task service times are NWU, independent across the servers, i.i.d. across the tasks assigned to the same server, and (iii) $\bm{O}=\bm{0}$, then for all \emph{$\pi\in\Pi$} and $\mathcal{I}$
\emph{
\begin{align}\label{eq_delaygap3_thm4}
[L_{\max}(\bm{C}(\text{EDD-R}))|\mathcal{I}] \leq_{\text{st}} [L_{\max}(\bm{C}(\pi))|\mathcal{I}].
\end{align}}
\end{theorem}

If the jobs sizes satisfy certain conditions, Theorem \ref{thm4} can be generalized to all delay metrics in {$\mathcal{D}_{\text{sym}}\cup\mathcal{D}_{\text{Sch-1}}$}.
\begin{theorem}\label{coro4_1}
If $k_1\leq k_2\leq\ldots\leq k_n$ and \emph{$L_{\max}$} is replaced by any \emph{$f\in\mathcal{D}_{\text{sym}}\cup\mathcal{D}_{\text{Sch-1}}$}, Theorem  \ref{thm4} still holds.
\end{theorem}
If the task service times are exponential, the delay performance of the EDD-R policy is characterized in the following two theorems. 
\begin{theorem}\label{thm4_exp}
If (i) the task service times are exponential, independent across the servers, i.i.d.~across the tasks assigned to the same server, (ii) $\bm{O}=\bm{0}$, then for all \emph{$\pi\in\Pi$}, and $\mathcal{I}$
\emph{
\begin{align}\label{eq_delaygap4_exp}
[L_{\max}(\bm{V}(\text{EDD-R}))|\mathcal{I}] \leq_{\text{st}} [L_{\max}(\bm{C}(\pi))|\mathcal{I}].
\end{align}}
\end{theorem}

\begin{theorem}\label{coro_thm3_1_exp}
If $k_1=\ldots=k_n=1$ (or $d_1\leq d_2\leq \ldots\leq d_n$ and $k_1\leq k_2 \leq \ldots\leq k_n$) and \emph{$L_{\max}$} is replaced by any \emph{$f\in\mathcal{D}_{\text{sym}}\cup\mathcal{D}_{\text{Sch-1}}$}, Theorem  \ref{thm4_exp} still holds.
\end{theorem}


\subsection{Maximum Delay and Related Delay Metrics} \label{sec_maxdelay}
Finally, we consider maximum delay $D_{\max}$ and the delay metrics in $\mathcal{D}_{\text{Sch-2}}$. 
When the task service times are NBU, we propose a policy named \textbf{First-Come, First-Served with Low-Priority Replication (FCFS-LPR)}, which can be obtained by combining the following FCFS discipline and the LPR discipline. 

\begin{definition}
A scheduling policy is said to follow the \textbf{First-Come, First-Served (FCFS)} first discipline, if each task assigned to the servers is from the job with the earliest arrival time whenever the queue is not empty (there exist unassigned tasks in the queue). 
\end{definition}

Two instances of the FCFS-LPR policies are \textbf{Earliest Due Date first with No Idleness Replication (FCFS-NIR)} and \textbf{First-Come, First-Served with No Replication (FCFS-NR)}. Policy FCFS-NIR can be obtained from Algorithm \ref{alg1} by revising Step 12 as $j :=\arg\min\{a_i: i\in Q,\gamma_i>0\}$. Policy FCFS-NR can be obtained by revising Step 12 of Algorithm \ref{alg1} and further removing Steps 15-16, 20-26. The delay performance of FCFS-LPR is characterized as follows:

\begin{corollary}\label{thm5}
If the task service times are NBU, independent across the servers, i.i.d.~across the tasks assigned to the same server,  then for all $\bm{O}\geq\bm{0}$,  $\pi\in\Pi$, and $\mathcal{I}$
\emph{
\begin{align}\label{eq_delaygap1_thm5}
&[D_{\max}(\bm{V}(\text{FCFS-LPR}))|\mathcal{I}] \leq_{\text{st}} \left[D_{\max}(\bm{C}(\pi))\left.\right|\mathcal{I}\right].
\end{align}}
\end{corollary}

If the job sizes satisfy certain conditions, Corollary \ref{thm5} can be generalized to all delay metrics in {$\mathcal{D}_{\text{sym}}\cup\mathcal{D}_{\text{Sch-2}}$}.
\begin{corollary}\label{coro5_1}
If $k_1\leq k_2 \leq \ldots\leq k_n$ and \emph{$D_{\max}$} is replaced by any \emph{$f\in\mathcal{D}_{\text{sym}}\cup\mathcal{D}_{\text{Sch-2}}$}, Corollary  \ref{thm5} still holds.
\end{corollary}

When the task service times are NWU, we propose a policy called \textbf{First-Come, First-Served with Replication (FCFS-R)}. This policy is similar with the FUT-R policy, except that in the FCFS-R policy, all servers are allocated to process $m$ replicated copies of a task from the job with the earliest arrival time. The FCFS-R policy can be obtained from Algorithm \ref{alg2} by revising Step 8 as $j :=\arg\min\{a_i: i\in Q\}$. The delay performance of the FCFS-R policy is provided as follows.

\begin{corollary}\label{thm6}
If (i) the task service times are NWU, independent across the servers, i.i.d.~across the tasks assigned to the same server,  and (ii) $\bm{O}=\bm{0}$, then for all $\pi\in\Pi$, and $\mathcal{I}$
\emph{
\begin{align}\label{eq_delaygap3_thm6}
[D_{\max}(\bm{C}(\text{FCFS-R}))|\mathcal{I}] \leq_{\text{st}} [D_{\max}(\bm{C}(\pi))|\mathcal{I}].
\end{align}}
\end{corollary}

\begin{corollary}\label{coro6_1}
If $k_1\leq k_2\leq\ldots\leq k_n$ and \emph{$D_{\max}$} is replaced by any \emph{$f\in\mathcal{D}_{\text{sym}}\cup\mathcal{D}_{\text{Sch-2}}$}, Corollary  \ref{thm6} still hold.
\end{corollary}

Corollaries \ref{thm5}-\ref{coro6_1} follow directly from Theorems \ref{thm3}-\ref{coro4_1} by setting $d_i=a_i$ for all job $i$. Nonetheless, due to the importance of the maximum delay metric $D_{\max}$ and the FCFS queueing discipline, Corollaries \ref{thm5}-\ref{coro6_1} are of independent interests.

We note that it is difficult to obtain an additive sub-optimality delay gap for minimizing the maximum lateness $L_{\max}(\bm{C}(\pi))$ or maximum delay $D_{\max}(\bm{C}(\pi))$ that remains constant for any number of jobs $n$. This is because the maximum lateness $L_{\max}(\bm{C}(\pi))$ and maximum delay $D_{\max}(\bm{C}(\pi))$ will likely grow to infinity as the number of jobs $n$ increases, due to the maximum operator over all jobs. Hence, unlike the  average delay $D_{\text{avg}}(\bm{C}(\pi))$, the maximum lateness $L_{\max}(\bm{C}(\pi))$ and maximum delay $D_{\max}(\bm{C}(\pi))$ are unstable delay metrics as $n\rightarrow\infty$. In our future work, we will consider  stable delay metrics and try to establish tight additive  gaps from the optimal delay.

%
%
%
%




%% file: distributed.tex

\section{Replications in Distributed Queueing Systems}\label{sec_distributed}
In this section, we propose scheduling policies for replications in distributed queueing systems with data-locality constraints,\footnote{Note that if there is no data-locality constraint, each task can be assigned on all servers. In this case, a distributed queueing system can be equivalently viewed as a centralized queueing system and all results in Section \ref{sec_analysis} apply directly.} and prove that these policies are near delay-optimal for minimizing several classes of delay metrics. To the extent of our knowledge, these are the first results on delay-optimal scheduling of replications in distributed queueing systems with data-locality constraints.




\subsection{An Equivalent Distributed Queueing Model}

There are two types of data locality constraints: per-task constraints and per-job constraints \cite{Sparrow:2013}. In \textbf{per-task  data locality constraints}, the tasks of one job may have a different group of servers on which it can run; while in \textbf{per-job data locality constraints}, all tasks of one job must be executed on a predetermined group of servers. Per-task constraints are more general than per-job constraints, and are also more difficult to handle. Both types of constraints play an important role in cloud computing  \cite{taskplacement2011}. We will consider both types of data locality constraints in our study.

For the convenience of analysis, we consider a hierarchical distributed queueing model depicted in Fig. \ref{fig1model_distri_eq}. In this model, there are two levels of job queues: the job queues at the schedulers, and the sub-job queues at the server groups. At the scheduler side, each incoming job is split into $g$ sub-jobs, where the $h$-th sub-job consists of the tasks to be executed by server group $h$ and $h=1,\ldots,g$. Each sub-job is routed its corresponding server group, and stored in a local queue. Then, a local scheduler assigns tasks to the servers within the group. 
In this hierarchical distributed queueing model, each local queue for a group of servers can be considered as a centralized queueing system in Fig. \ref{fig1model_central}. 
It is important to note that this hierarchical distributed queueing model is equivalent to the original distributed queueing model in Fig. \ref{fig1model_distri}. In particular, according to the discussions in Section \ref{sec:model}, each decision of the local scheduler in the hierarchical distributed queueing model can be equivalently implemented in the original distributed queueing model, and vice versa.


Let $k_{ih}$ denote the number of tasks in $h$-th sub-job of job $i$, where $\sum_{h=1}^g k_{ih} = k_i$. If no task of job $i$ should be executed by $h$-th group of the servers, then $k_{ih}=0$. 
The arrival time and due time of each sub-job of job $i$ are $a_i$ and $d_i$, respectively. Define $\mathcal{I}_h = \{n,(a_i,d_i,k_{ih})_{i=1}^n\}$ as the parameters of the sub-jobs of server group $h$.
For sub-job $h$ of job $i$, $V_{ih}$ is the earliest time that all tasks of the sub-job have entered the servers, $C_{ih}$ is the completion time, $D_{ih}= C_{ih}-a_i$ is the delay, and $L_{ih} = C_{ih}-d_i$ is the {lateness} after the due time $d_i$. If $k_{ih}=0$,  we set $V_{ih}=C_{ih}=0$. Then, it holds that $V_{ih}\leq C_{ih}$. In addition, a job is completed when all of its $g$ sub-jobs are completed, i.e.,
\begin{align}\label{eq_sub_job_relation}
C_i = \max_{h=1,\ldots,g}C_{ih}, V_i=\max_{h=1,\ldots,g}V_{ih}, D_i = \max_{h=1,\ldots,g} D_{ih},  L_i = \max_{h=1,\ldots,g}L_{ih},
\end{align}
where $V_i$ and $C_i$  depend only on the sub-jobs with positive sizes $k_{ih}>0$. Define $\bm{V}_h = (V_{1h},\ldots,V_{nh})$ and $\bm{C}_h = (C_{1h},\ldots,C_{nh})$. 
Let  $V_{(i),h}$ and $C_{(i),h}$ denote the $i$-th smallest components of $\bm{V}_h$ and $\bm{C}_h$, respectively. 
All these quantities are functions of the adopted scheduling policy $\pi$. 

Note that under per-job data locality constraints, each job can be only executed by one predetermined  group of the servers. Therefore, for each job $i$ there exists $u(i)\in\{1,\ldots,g\}$ such that $C_{i,u(i)} = C_{i} $, $V_{i,u(i)}  = V_{i} $, and $C_{i,h} = V_{i,h}= 0$ for all other server groups satisfying $h\neq u(i)$.

Next, we will exploit the results in Section \ref{sec_analysis} to study the delay performance of replications in distributed queueing systems.


\begin{figure}
\centering
\includegraphics[width=0.9\textwidth]{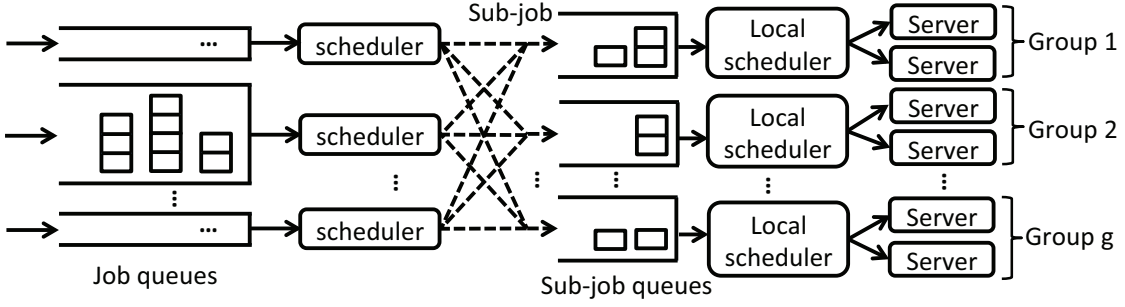}  
\caption{A hierarchical distributed queueing model, which is equivalent to the distributed queueing model in Fig. \ref{fig1model_distri}. } \label{fig1model_distri_eq} 
\end{figure} 

\subsection{Per-Task Data Locality Constraints}
We first consider delay minimization with per-task data locality constraints.
For minimizing the maximum lateness $L_{\max}(\cdot)$, we propose one policy called \textbf{Earliest Due Date first with Group-based Replication (EDD-GR)}: In the hierarchical distributed queueing model, once a job arrives at a scheduler, it is split into $g$ sub-jobs which are immediately routed to the $g$ local schedulers. If the task service times are NBU for one group of servers, the local scheduler will make decisions by following policy EDD-LPR; if the task service times are NWU for one group of servers, the local scheduler will make decisions by following policy EDD-R; if the task service times are exponential for one group of servers, the local scheduler can either choose EDD-LPR or EDD-R. According to the discussions in Section \ref{sec:model}, policy EDD-LPR (policy EDD-R) can be implemented in the original distributed queueing model  as follows: First, replicate each task to all the local task queues in the corresponding group of  servers (see Fig. \ref{fig1model_distri}), and then let the servers to communicate with each other to determine the replication and cancellation operations by following the LPR discipline (R discipline); each task assigned to the servers has earliest due date among all the unassigned tasks whenever the queue is not empty (there exist unassigned tasks in the queue). 
Hence, \emph{policy EDD-GR can be implemented distributedly}. The delay performance of policy EDD-GR is provided as follows. 
\begin{theorem}\label{thm_dis4}
If (i) the task service times are independent across the servers and i.i.d.~across the tasks assigned to the same server, (ii) each server group $h$ and its sub-job parameters $\mathcal{I}_h$ satisfy the conditions of Theorem \ref{thm3}, Theorem \ref{thm4}, or Theorem \ref{thm4_exp}, and (iii) the system is subject to per-task data locality constraints, then for all \emph{$\pi\in\Pi$} and $\mathcal{I}$
\emph{
\begin{align}\label{eq_delaygap3_thm4_dis}
[L_{\max}(\bm{V}(\text{EDD-GR}))|\mathcal{I}] \leq_{\text{st}} [L_{\max}(\bm{C}(\pi))|\mathcal{I}].
\end{align}}
\end{theorem}
\begin{proof}
See Appendix \ref{app_dis_EDD}.
\end{proof}

For minimizing the maximum delay $D_{\max}(\cdot)$, we propose one policy called  \textbf{First-Come, First-Served with Group-based Replication (FCFS-GR)}: 
If the task service times are NBU for one group of servers, the local scheduler will make decisions by following policy FCFS-LPR; if the task service times are NWU for one group of servers, the local scheduler will make decisions by following policy FCFS-R; if the task service times are exponential for one group of servers, the local scheduler can either choose FCFS-LPR or FCFS-R. Similar with policy EDD-GR, policy FCFS-GR can also be implemented distributedly in the original distributed queueing systems.
By choosing $d_i=a_i$ for all $i$, it follows from Theorem \ref{thm_dis4} that

\begin{corollary}\label{thm_dis5}
If (i) the task service times are independent across the servers and i.i.d.~across the tasks assigned to the same server, (ii) each server group $h$ and its sub-job parameters $\mathcal{I}_h$ satisfy the conditions of Corollary \ref{thm5} or Corollary \ref{thm6}, and (iii) the system is subject to per-task data locality constraints, then for all \emph{$\pi\in\Pi$} and $\mathcal{I}$
\emph{
\begin{align}\label{eq_delaygap3_thm5_dis}
[D_{\max}(\bm{V}(\text{FCFS-GR}))|\mathcal{I}] \leq_{\text{st}} [D_{\max}(\bm{C}(\pi))|\mathcal{I}].
\end{align}}
\end{corollary}


\subsection{Per-Job Data Locality Constraints}
It is difficult for us to generalize Theorem \ref{thm_dis4} and Corollary \ref{thm_dis5} and minimize other delay metrics under per-task data locality constraints. However, under per-job data locality constraints such generalizations are possible, which are discussed in the sequel.

For minimizing the delay metrics in $\mathcal{D}_{\text{sym}}$ (including the average delay $D_{\text{avg}}$), we propose one policy called \textbf{Fewest Unassigned Tasks first with Group-based Replication (FUT-GR)}: If the task service times are NBU for one group of servers, the local scheduler will make decisions by following policy FUT-LPR; if the task service times are NWU for one group of servers, the local scheduler will make decisions by following policy FUT-R; if the task service times are exponential for one group of servers, the local scheduler can either choose FUT-LPR or FUT-R. Hence, the priority of a sub-job is determined by the number of unassigned tasks in this sub-job. 
The delay performance of policy FUT-GR is provided as follows.

\begin{theorem}\label{thm2_dist}
If (i) the task service times are independent across the servers and i.i.d.~across the tasks assigned to the same server, (ii) each server group $h$ and its sub-job parameters $\mathcal{I}_h$ satisfy the conditions of Theorem \ref{thm1}, 
Theorem \ref{thm2}, or Theorem \ref{thm2_exp}, and (iii) the system is subject to per-job data locality constraints, then for all \emph{$f\in\mathcal{D}_{\text{sym}}$},  \emph{$\pi\in\Pi$}, and $\mathcal{I}$
\emph{
\begin{align}\label{eq_delaygap2_dist}
[f(\bm{V}(\text{FUT-GR}))|\mathcal{I}] \leq_{\text{st}} [f(\bm{C}(\pi))|\mathcal{I}].
\end{align}}
\end{theorem}
\begin{proof}
See Appendix \ref{app_thm2_dist}.
\end{proof}

Let us further consider the delay metrics in the set $\mathcal{D}_{\text{sym}}\cup\mathcal{D}_{\text{Sch-1}}$, for which we can obtain the following result.

\begin{theorem}\label{thm3_dist}
If (i) the task service times are independent across the servers and i.i.d.~across the tasks assigned to the same server, (ii) each server group $h$ and its sub-job parameters $\mathcal{I}_h$ satisfy the conditions of Theorem \ref{coro_thm3_1}, Theorem \ref{coro4_1}, or Theorem \ref{coro_thm3_1_exp}, and (iii) the system is subject to {per-job} data locality constraints, then for all \emph{$f\in\mathcal{D}_{\text{sym}}\cup\mathcal{D}_{\text{Sch-1}}$},  \emph{$\pi\in\Pi$}, and $\mathcal{I}$
\emph{
\begin{align}\label{eq_delaygap3_dist}
[f(\bm{V}(\text{EDD-GR}))|\mathcal{I}] \leq_{\text{st}} [f(\bm{C}(\pi))|\mathcal{I}].
\end{align}}
\end{theorem}
\begin{proof}
See Appendix \ref{app_thm3_dist}.
\end{proof}

Finally, if $d_i=a_i$ for all job $i$, it follows from Theorem 
\ref{thm3_dist} that 
\begin{corollary}\label{thm4_dist}
If (i) the task service times are independent across the servers and i.i.d.~across the tasks assigned to the same server, (ii) each server group $h$ and its sub-job parameters $\mathcal{I}_h$ satisfy the conditions of Corollary \ref{coro5_1} or Corollary \ref{coro6_1}, and (iii) the system is subject to per-job data locality constraints, then for all \emph{$f\in\mathcal{D}_{\text{sym}}\cup\mathcal{D}_{\text{Sch-2}}$},  \emph{$\pi\in\Pi$}, and $\mathcal{I}$
\emph{
\begin{align}\label{eq_delaygap4_dist}
[f(\bm{V}(\text{FCFS-GR}))|\mathcal{I}] \leq_{\text{st}} [f(\bm{C}(\pi))|\mathcal{I}].
\end{align}}
\end{corollary}

%% file: numerical.tex
\section{Numerical Results}\label{sec_numerical}
In this section, we present some numerical results to illustrate the delay
performance of different scheduling policies and validate our
theoretical results.
\subsection{Centralized Queueing Systems}
Consider a centralized queueing system consisting of 3 servers with heterogeneous service time distributions, 
The inter-arrival time of the jobs $a_{i+1}-a_i$ is exponentially distributed for even $i$; and is zero for odd $i$. Let $\lambda$ be the average job arrival rate.
\ifreport
\begin{figure}
\centering \includegraphics[width=0.6\textwidth]{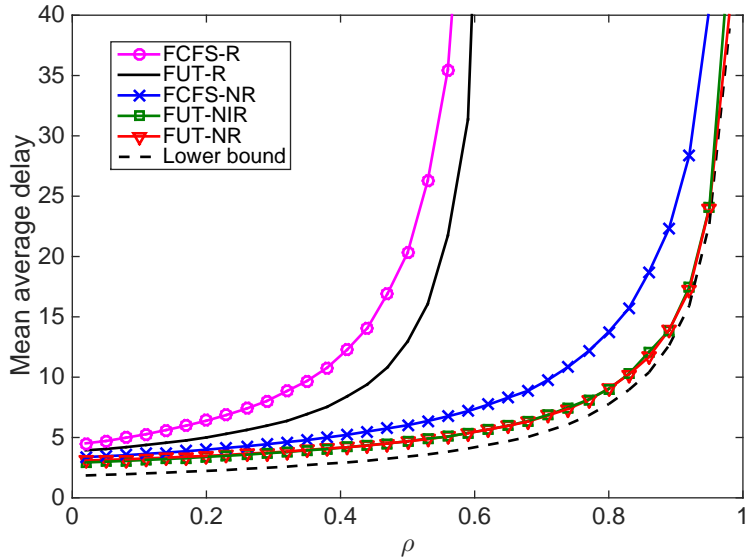} \caption{Expected average delay $\mathbb{E}[D_{\text{avg}} (\bm{C}(\pi))|\mathcal{I}]$ versus traffic intensity $\rho$ in a centralized queueing system with heterogeneous NBU service time distributions.}
\label{fig1} \vspace{-0.cm}
\end{figure}
\fi
\subsubsection{Average Delay}
Figure \ref{fig1} plots the expected average delay $\mathbb{E}[D_{\text{avg}} (\bm{C}(\pi))|\mathcal{I}]$ versus traffic intensity $\rho$ in a centralized queueing system with heterogeneous NWU service time distributions. The number of incoming jobs is $n=3000$, the job size $k_i$ is chosen to be either $1$ or $10$ with equal probability.
The task service time $X_l$ follows a shifted exponential distribution:
\begin{align}
\Pr[X_l>x] = \left\{\begin{array}{l l}1,&\text{if}~x<\frac{1}{3\mu_l};\\
\exp[-\frac{3\mu_l}{2}(x-\frac{1}{3\mu_l})],&\text{if}~x\geq \frac{1}{3\mu_l},
\end{array}\right.
\end{align}
and the service rate of the 3 servers are $\mu_1=1.4$, $\mu_2=1$, and $\mu_3=0.6$, respectively. 
The traffic intensity can be computed as $\rho = 
\lambda \frac{1+10}{2} (\mu_1+\mu_2+\mu_3) =33\lambda/2$. The cancellation overhead $O_l$ of server $l$ is exponentially distributed with rate $3\mu_l/2$. 
The ``Lower bound'' curve is generated by using $\mathbb{E}[D_{\text{avg}}(\bm{V}(\text{FUT-NR}))|\mathcal{I}]$ which, according to Theorem \ref{thm1}, is a lower bound of the optimum expected average delay.
We can observe from Fig. \ref{fig1} that the policies FCFS-NR, FUT-NR, and FUT-NIR are throughput-optimal, while the policies FCFS-R and  FUT-R have a smaller throughput region.  The average delays of policies FUT-NR and FUT-NIR are quite close to the lower bound, while the other policies are far from the lower bound.


\ifreport
\begin{figure}
\centering \includegraphics[width=0.6\textwidth]{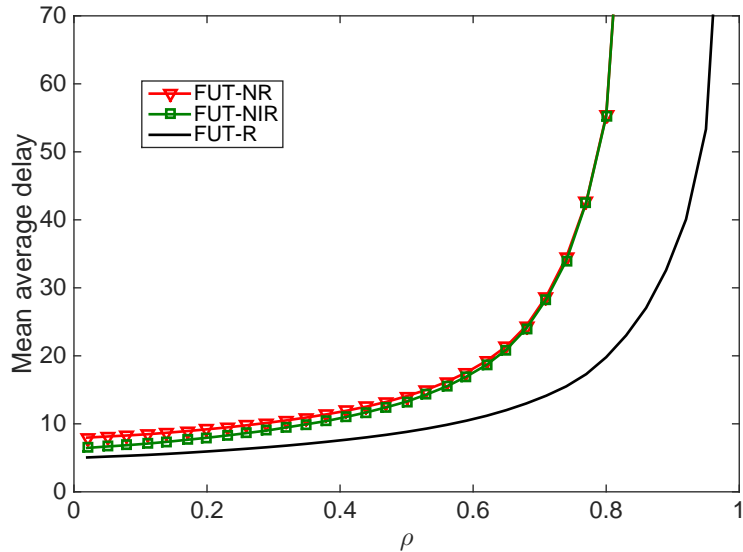} \caption{Expected average delay $\mathbb{E}[D_{\text{avg}} (\bm{C}(\pi))|\mathcal{I}]$ versus traffic intensity $\rho$ in a centralized queueing system with heterogeneous NWU service time distributions.}
\label{fig2} \vspace{-0.cm}
\end{figure}
\fi

Figure \ref{fig2} illustrates the expected average delay $\mathbb{E}[D_{\text{avg}} (\bm{C}(\pi))|\mathcal{I}]$ versus traffic intensity $\rho$ in a centralized queueing system with heterogeneous NWU service time distributions. The number of incoming jobs is $n=3000$, the job sizes are $k_i=10$ for all $i$. The task service time $X_l$ follows a Pareto type II (Lomax) distribution \cite{ParetoDis}:
\begin{align}
\Pr[X_l>x] = \left[1+\frac{x}{\sigma}\right]^{-\alpha_l},
\end{align}
where $\sigma = 14/3$, $\alpha_1 = 7$, $\alpha_2 = 5$, and $\alpha_3 = 3$. The cancellation overhead $O_l$ is zero for all servers. According to the property of Pareto type II (Lomax) distribution \cite{ParetoDis}, the maximum task service rate is $1/\mathbb{E}[\min_{l=1,2,3}X_l] = (\alpha_1+\alpha_2+\alpha_3-1)/\sigma=3$, which is achieved when each task is replicated on all 3 servers. Hence, the traffic intensity can be computed as $\rho = \lambda \times\frac{1+10}{2}\frac{1}{\mathbb{E}[\min_{l=1,2,3}X_l]} =33\lambda/2$. Because all jobs are of the same size, the FUT discipline is identical with the FCFS discipline. 
We can observe that policy FUT-R, which is identical with policy FCFS-R, is throughput-optimal, while  policy FUT-NR and policy FUT-NIR have a smaller throughput region. 
In addition, policy FUT-R achieves better delay performance than the other policies, which is in accordance with Theorem \ref{thm2}. 


\begin{figure}
\centering \includegraphics[width=0.6\textwidth]{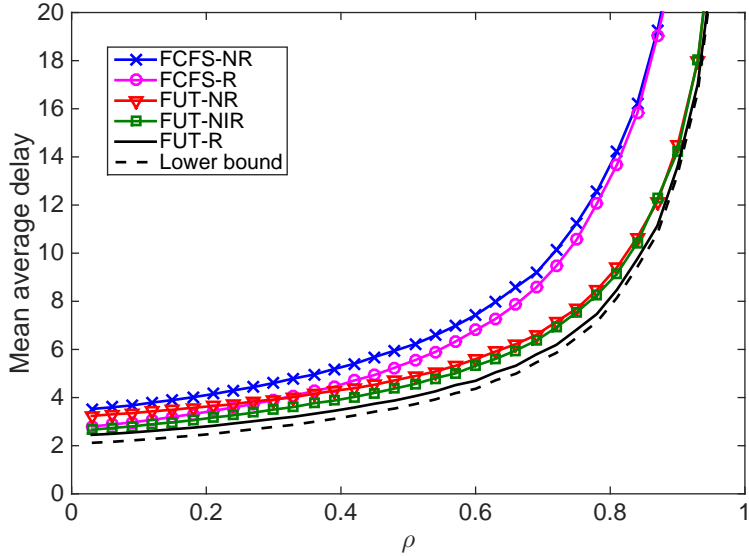} \caption{Expected average delay $\mathbb{E}[D_{\text{avg}} (\bm{C}(\pi))|\mathcal{I}]$ versus traffic intensity $\rho$ in a centralized queueing system with heterogeneous exponential service time distributions.}
\label{fig2_exp} \vspace{-0.cm}
\end{figure}

Figure \ref{fig2_exp} depicts the expected average delay $\mathbb{E}[D_{\text{avg}} (\bm{C}(\pi))|\mathcal{I}]$ versus traffic intensity $\rho$ in a centralized queueing system with heterogeneous exponential service time distributions. The number of incoming jobs is $n=3000$, the job size $k_i$ is chosen to be either $1$ or $10$ with equal probability. The service rate of the 3 servers are $\mu_1=1.4$, $\mu_2=1$, and $\mu_3=0.6$, respectively. The cancellation overhead $O_l$ is zero for all servers. The maximum task service rate is $1/\mathbb{E}[\min_{l=1,2,3}X_l] = \mu_1+\mu_2+\mu_3=3$. Hence, the traffic intensity can be computed as $\rho = \lambda \times\frac{1+10}{2}\frac{1}{\mathbb{E}[\min_{l=1,2,3}X_l]} =33\lambda/2$. The ``Lower bound'' curve is generated by using $\mathbb{E}[D_{\text{avg}}(\bm{V}(\text{FUT-R}))|\mathcal{I}]$ which, according to Theorem \ref{thm1}, is a lower bound of the optimum expected average delay. We can observe that the delay performance of policy FUT-R is close to the lower bound, compared with the other policies.



\ifreport
\begin{figure}
\centering \includegraphics[width=0.6\textwidth]{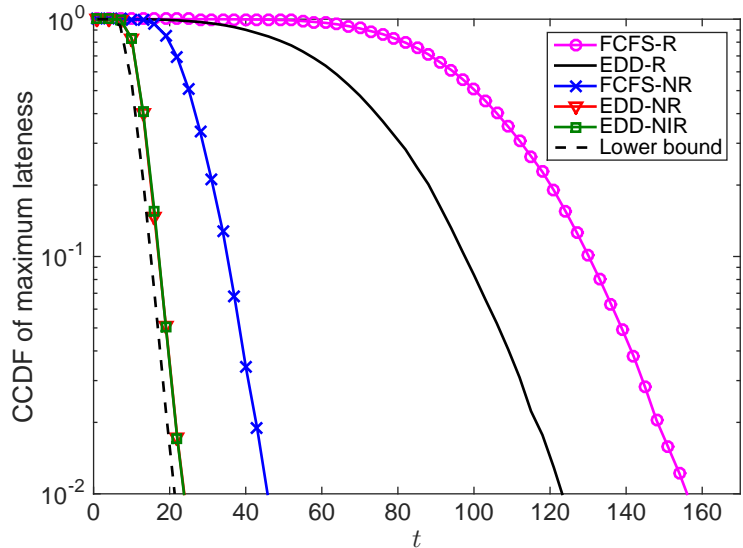} \caption{Complementary CDF of the maximum lateness $\Pr[L_{\max} (\bm{C}(\pi))>t|\mathcal{I}]$ versus $t$ in a centralized queueing system with heterogeneous NBU service time distributions.}
\label{fig3} \vspace{-0.cm}
\end{figure}
\fi

\subsubsection{Maximum Lateness}

Figure \ref{fig3} evaluates the complementary CDF of maximum lateness $\Pr[L_{\max} (\bm{C}(\pi))$ $>t|\mathcal{I}]$ versus $t$ in a centralized queueing system with heterogeneous NBU service time distributions. The number of incoming jobs is $n=100$, the job size $k_i$ is chosen to be either $1$ or $10$ with equal probability, and the due time $d_i$ is chosen to be either $a_i$ or $a_i+50$ with equal probability. 
The distributions of task service times and cancellation overheads are the same with those in Fig. \ref{fig1}. The traffic intensity is set as $\rho = 0.8$.  
The ``Lower bound'' curve is generated by using $\Pr[L_{\max}(\bm{V}(\text{EDD-NR}))>t|\mathcal{I}]$, which, according to Theorem \ref{thm3}, is a lower bound of the optimum delay performance.
We can observe that the complementary CDF of the maximum lateness of policies EDD-NR and EDD-NIR are close to the lower bound curve. In addition, the delay performance of policy FCFS-NR is better than that of FCFS-R and EDD-R. This is because policy FCFS-NR  has a larger throughput region than policy  FCFS-R and policy EDD-R. 


\ifreport
\begin{figure}
\centering \includegraphics[width=0.6\textwidth]{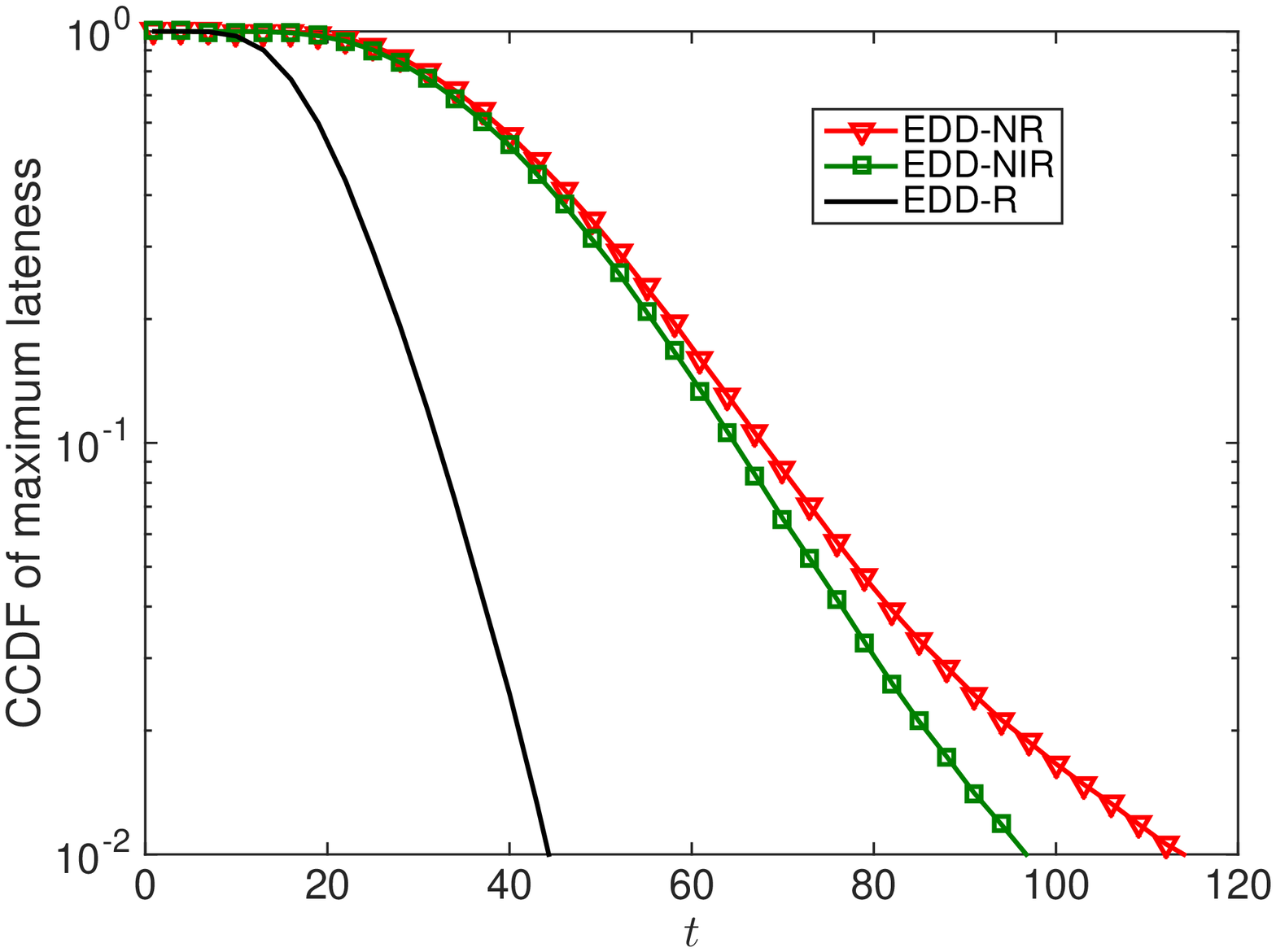} \caption{Complementary CDF of the maximum lateness $\Pr[L_{\max} (\bm{C}(\pi))>t|\mathcal{I}]$ versus $t$ in a centralized queueing system with heterogeneous service time NWU distributions.}
\label{fig4} \vspace{-0.cm}
\end{figure}
\fi

Figure \ref{fig4} shows the complementary CDF of maximum lateness $\Pr[L_{\max} (\bm{C}(\pi))>t|\mathcal{I}]$ versus $t$ for a centralized queueing system with heterogeneous NWU service time distributions. The number of incoming jobs is $n=100$, the job size $k_i$ is chosen to be either $1$ or $10$ with equal probability, and the due time $d_i$ is $a_i+5$. 
The distributions of task service times and cancellation overheads are the same with those in Fig. \ref{fig2}. The traffic intensity is set as $\rho = 0.8$. 
Because $d_1\leq d_2\leq\ldots\leq d_n$, the EDD discipline is identical with the FCFS discipline. 
We can observe that policy EDD-R, which is identical with policy FCFS-R, achieves better  performance than the other policies, which is in accordance with Theorem \ref{thm4}. 


\ifreport
\begin{figure}
\centering \includegraphics[width=0.6\textwidth]{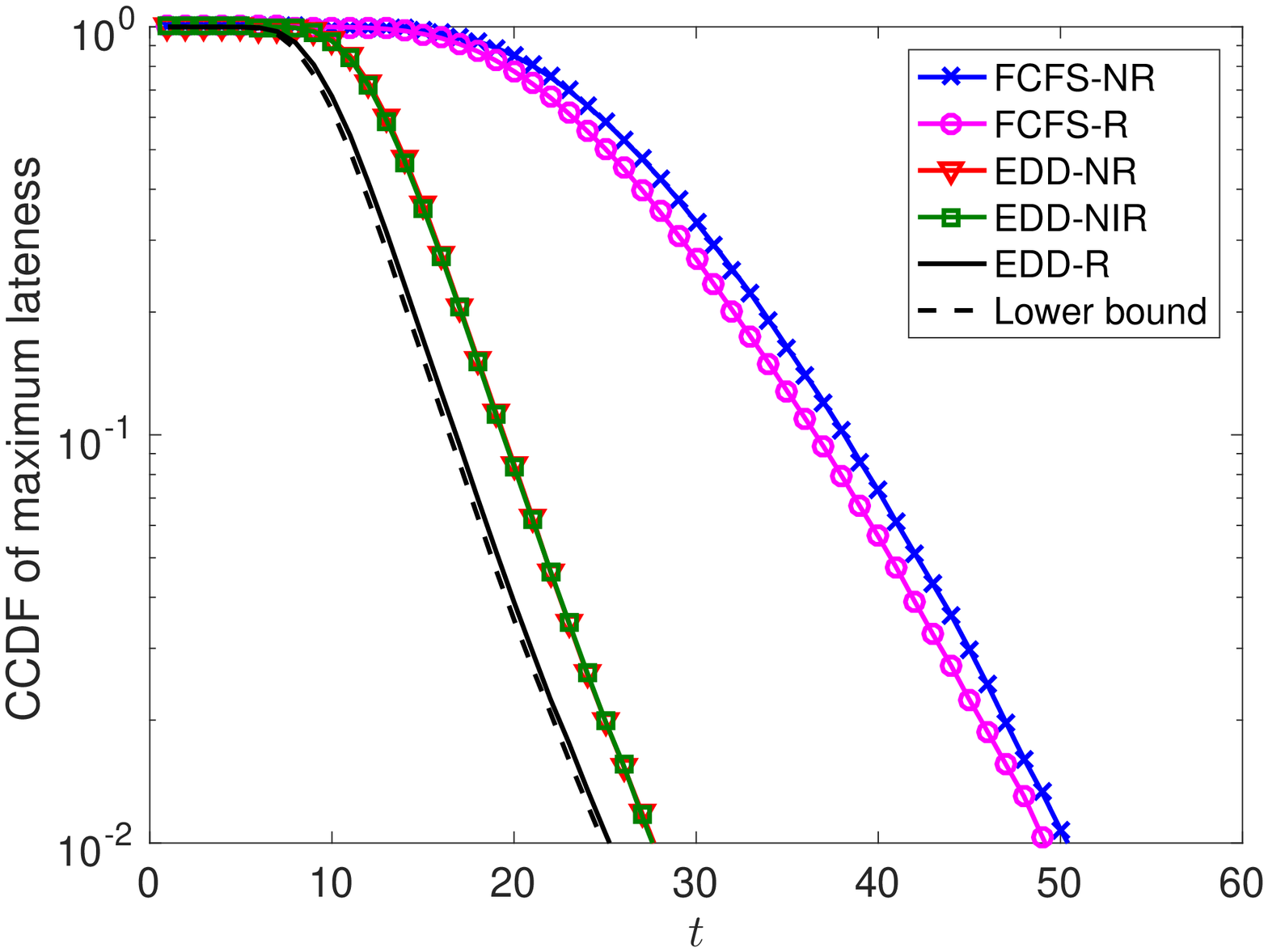} \caption{Complementary CDF of the maximum lateness $\Pr[L_{\max} (\bm{C}(\pi))>t|\mathcal{I}]$ versus $t$ in a centralized queueing system with heterogeneous exponential service time distributions.}
\label{fig4_exp} \vspace{-0.cm}
\end{figure}
\fi

Figure \ref{fig4_exp} provides the complementary CDF of maximum lateness $\Pr[L_{\max} (\bm{C}(\pi))>t|\mathcal{I}]$ versus $t$ for a centralized queueing system with heterogeneous exponential service time distributions. The number of incoming jobs is $n=100$, the job size $k_i$ is chosen to be either $1$ or $10$ with equal probability, and the  due time $d_i$ is chosen to be either $a_i$ or $a_i+50$ with equal probability. 
The distributions of task service times and cancellation overheads are the same with those in Fig. \ref{fig2_exp}. The traffic intensity is set as $\rho = 0.8$. The ``Lower bound'' curve is generated by using $\Pr[L_{\max}(\bm{V}(\text{EDD-R}))>t|\mathcal{I}]$.
We can observe that the delay performance of policy EDD-R is quite close to the lower bound curve. The delay performance of policy EDD-NR and policy EDD-NIR is a bit farther than from the lower bound curve. 
Policy FCFS-R and policy FCFS-NR have the worse performance. 
These results are in accordance with  Theorem \ref{thm3} and Theorem \ref{thm4_exp}.



\ifreport
\begin{figure}
\centering \includegraphics[width=0.6\textwidth]{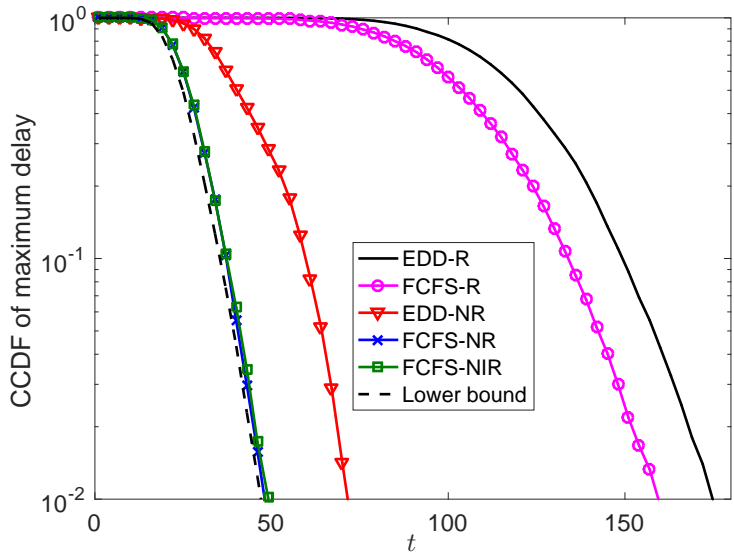} \caption{Complementary CDF of the maximum delay $\Pr[D_{\max} (\bm{C}(\pi))>t|\mathcal{I}]$ versus $t$  in a centralized queueing system with heterogeneous NBU service time distributions.}
\label{fig5} \vspace{-0.cm}
\end{figure}
\fi
\ifreport
\begin{figure}
\centering \includegraphics[width=0.6\textwidth]{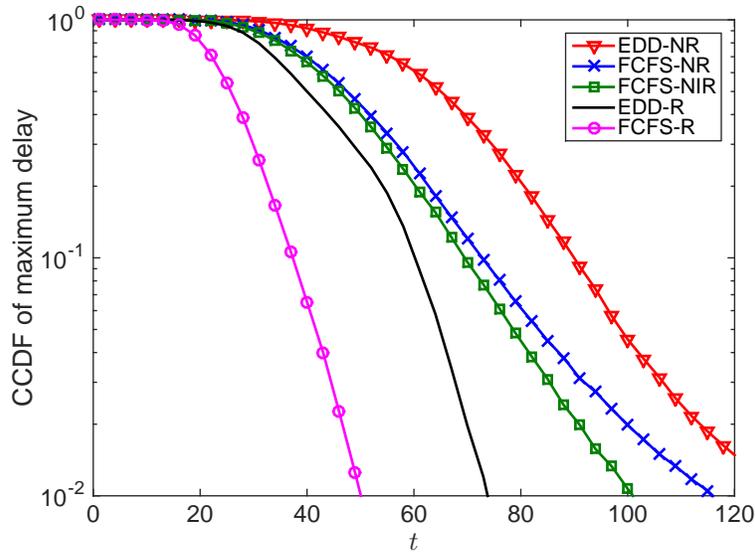} \caption{Complementary CDF of the maximum delay $\Pr[D_{\max} (\bm{C}(\pi))>t|\mathcal{I}]$ versus $t$  in a centralized queueing system with heterogeneous NWU service time distributions.}
\label{fig6} \vspace{-0.cm}
\end{figure}
\fi

\subsubsection{Maximum Delay}
Figure \ref{fig5} plots the complementary CDF of maximum delay $\Pr[D_{\max} (\bm{C}(\pi))>t|\mathcal{I}]$ versus $t$ in a centralized queueing system with heterogeneous NBU service time distributions. The system model is the same with that of Fig. \ref{fig3}.
The ``Lower bound'' curve is generated by using $\Pr[D_{\max}(\bm{V}(\text{FCFS-NR}))$ $>t|\mathcal{I}]$,  which, according to Corollary \ref{thm5}, is a lower bound of the optimum delay performance.
We can observe that the delay performance of policy FCFS-NR and policy FCFS-NIR is close to the lower bound curve. In addition, the delay performance of policy EDD-NR is better than that of FCFS-R and EDD-R because policy EDD-NR has a larger throughput region than policy  FCFS-R and policy EDD-R. 

Figure \ref{fig6} illustrates the complementary CDF of maximum delay $\Pr[D_{\max} (\bm{C}(\pi))>t|\mathcal{I}]$ versus $t$ in a centralized queueing system with heterogeneous NWU service time distributions. The system model is almost the same with that of Fig. \ref{fig4}, except that the  due time $d_i$ is chosen to be either $a_i$ or $a_i+50$ with equal probability. 
We can observe that the delay performance of policy FCFS-R is much better than that of the other policies, which validates Corollary \ref{thm6}. 

\ifreport
\begin{figure}
\centering \includegraphics[width=0.6\textwidth]{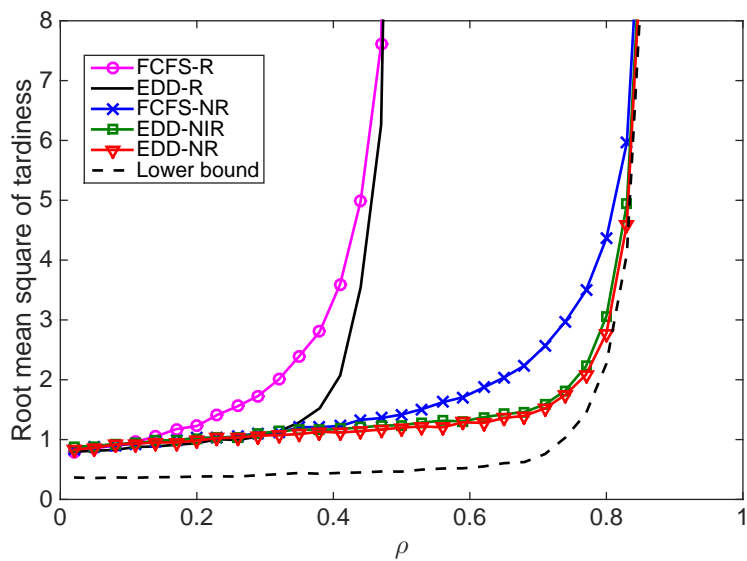} \caption{Root mean square of tardiness $\mathbb{E}[T_{\text{ms}} (\bm{C}(\pi))|\mathcal{I}]^{0.5}$ versus traffic intensity $\rho$ in a centralized queueing system with heterogeneous NBU service time distributions.}
\label{fig5_RMS_tardiness} \vspace{-0.cm}
\end{figure}
\fi

\ifreport
\begin{figure}
\centering \includegraphics[width=0.6\textwidth]{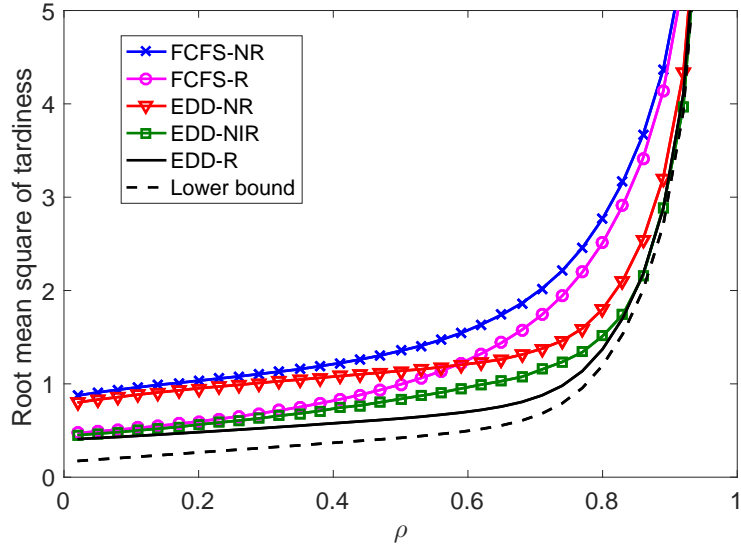} \caption{Root mean square of tardiness $\mathbb{E}[T_{\text{ms}} (\bm{C}(\pi))|\mathcal{I}]^{0.5}$ versus traffic intensity $\rho$ in a centralized queueing system with heterogeneous exponential service time distributions.}
\label{fig6_RMS_tardiness} \vspace{-0.cm}
\end{figure}
\fi

\subsection{Some Other Delay Metrics}
Figure \ref{fig5_RMS_tardiness} shows the root mean square of tardiness $\mathbb{E}[T_{\text{ms}} (\bm{C}(\pi))|\mathcal{I}]^{0.5}$ versus traffic intensity $\rho$ in a centralized queueing system with heterogeneous NBU service time distributions, where 
\begin{align}
T_{\text{ms}}(\bm{C}(\pi)) \!= \frac{1}{n}\sum_{i=1}^{n} \max[L_i(\pi),0]^2=\frac{1}{n}\sum_{i=1}^{n} \max\left[C_i(\pi)-d_i,0\right]^2.\nonumber
\end{align}
The system model is similar with that of Fig. \ref{fig5}, except that $k_i=1$ for all job $i$ and the due time $d_i$ is either $a_i$ or $a_i+10$ with equal probability.
The ``Lower bound'' curve is generated by using $\mathbb{E}[T_{\text{ms}} (\bm{V}(\text{EDD-NR}))|\mathcal{I}]^{0.5}$, which, according to Theorem \ref{coro_thm3_1}, is a lower bound of the optimum delay performance.
We can observe that the delay performance of policy EDD-NR and policy EDD-NIR is close to the lower bound curve, and is better than that of the other policies. Notice that as  $\rho\rightarrow 0$, policy EDD-R and policy EDD-NIR tends become the same policy, and hence has the same delay performance.


Figure \ref{fig6_RMS_tardiness} presents the root mean square of tardiness $\mathbb{E}[T_{\text{ms}} (\bm{C}(\pi))|\mathcal{I}]^{0.5}$ versus traffic intensity $\rho$ in a centralized queueing system with heterogeneous exponential service time distributions. The system model is similar with that of Fig. \ref{fig6}, except that $k_i=1$ for all job $i$ and the due time $d_i$ is either $a_i$ or $a_i+10$ with equal probability. 
The ``Lower bound'' curve is generated by using $\mathbb{E}[T_{\text{ms}} (\bm{V}(\text{EDD-R}))|\mathcal{I}]^{0.5}$, which, according to Theorem \ref{coro_thm3_1_exp}, is a lower bound of the optimum delay performance.
We can observe that the delay performance of policy EDD-R is close to the lower bound curve, and is better than that of the other policies.

\ifreport
\begin{figure}
\centering \includegraphics[width=0.6\textwidth]{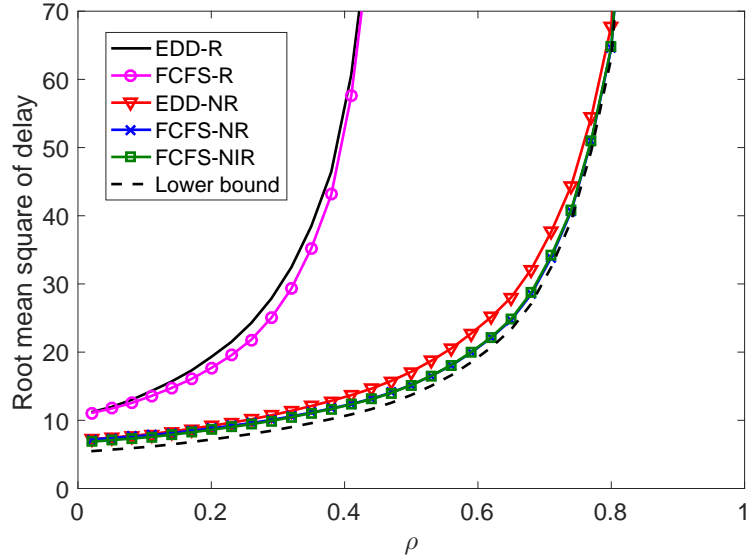} \caption{Root mean square of delay $\mathbb{E}[D_{\text{ms}} (\bm{C}(\pi))|\mathcal{I}]^{0.5}$ versus traffic intensity $\rho$ in a centralized queueing system with heterogeneous NBU service time distributions.}
\label{fig5_RMS} \vspace{-0.cm}
\end{figure}
\fi

\ifreport
\begin{figure}
\centering \includegraphics[width=0.6\textwidth]{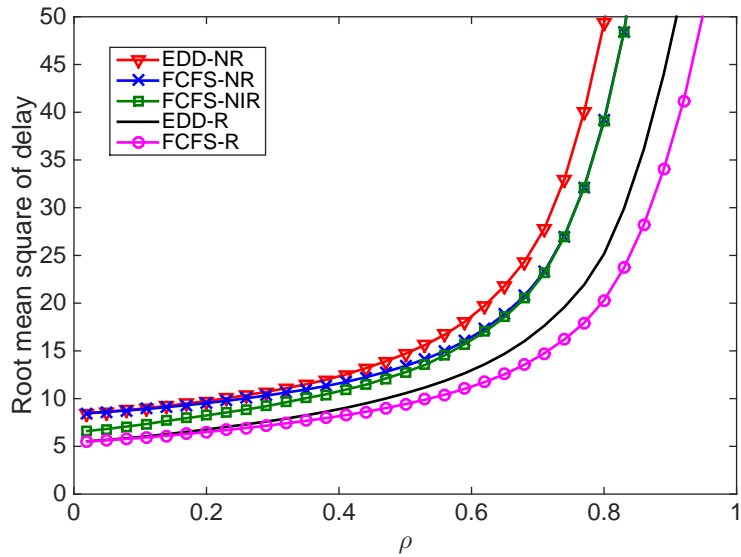} \caption{Root mean square of delay $\mathbb{E}[D_{\text{ms}} (\bm{C}(\pi))|\mathcal{I}]^{0.5}$ versus traffic intensity $\rho$ in a centralized queueing system with heterogeneous NWU service time distributions.}
\label{fig6_RMS} \vspace{-0.cm}
\end{figure}
\fi

Figure \ref{fig5_RMS} illustrates the root mean square of delay $\mathbb{E}[D_{\text{ms}} (\bm{C}(\pi))|\mathcal{I}]^{0.5}$ versus traffic intensity $\rho$ in a centralized queueing system with heterogeneous NBU service time distributions, where 
\begin{align}
D_{\text{ms}}(\bm{C}(\pi)) \!= \frac{1}{n}\sum_{i=1}^{n} D_i^2(\pi)=\frac{1}{n}\sum_{i=1}^{n} \left[C_i(\pi)-a_i\right]^2.\nonumber
\end{align}
The system model is almost the same with that of Fig. \ref{fig5}, except that $k_i=10$ for all job $i$.
The ``Lower bound'' curve is generated by using $\mathbb{E}[D_{\text{ms}} (\bm{V}(\text{FCFS-NR}))|\mathcal{I}]^{0.5}$, which, according to Corollary \ref{coro5_1}, is a lower bound of the optimum delay performance.
We can observe that the delay performance of policy FCFS-NR and policy FCFS-NIR is close to the lower bound curve, and is much better than the other policies. 


Figure \ref{fig6_RMS} presents the root mean square of delay $\mathbb{E}[D_{\text{ms}} (\bm{C}(\pi))|\mathcal{I}]^{0.5}$ versus traffic intensity $\rho$ in a centralized queueing system with heterogeneous NWU service time distributions. The system model is almost the same with that of Fig. \ref{fig6}, except that $k_i=10$ for all job $i$. 
We can observe that the delay performance of policy FCFS-R is much better than that of the other policies, which validates Corollary \ref{coro6_1}. 

\ifreport
\begin{figure}
\centering \includegraphics[width=0.6\textwidth]{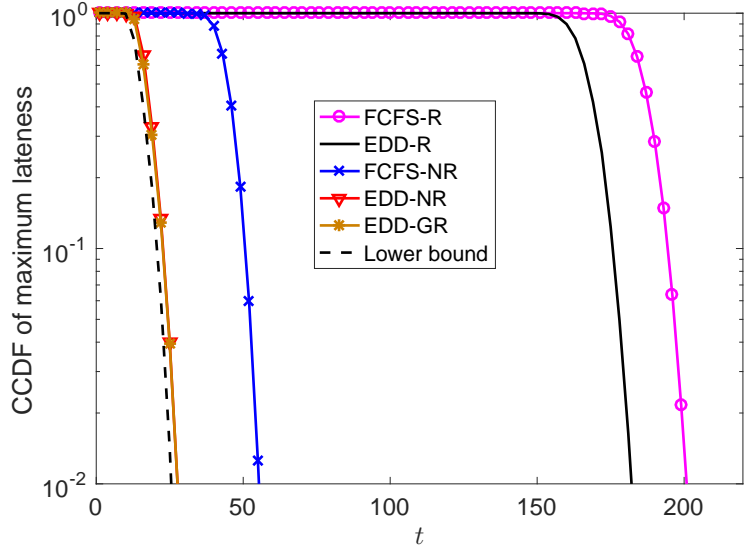} \caption{Complementary CDF of the maximum lateness $\Pr[L_{\max} (\bm{C}(\pi))>t|\mathcal{I}]$ versus $t$ in a distributed queueing system with per-task data locality constraints.}
\label{fig7} \vspace{-0.cm}
\end{figure}
\fi
\subsection{Distributed Queueing Systems}


Next, we provide some numerical results for the delay performance of replications in distributed queueing systems with data locality constraints. The inter-arrival time of the jobs $T_i=a_{i+1}-a_i$ is exponentially distributed for even $i$; and is zero for odd $i$. The number of incoming jobs is $n=100$. The due time $d_i$ of job $i$ is $a_i$ or $a_i+50$ with equal probability. The traffic intensity is set as $\rho = 0.8$.  

Figure \ref{fig7} evaluates the complementary CDF of maximum lateness $\Pr[L_{\max} (\bm{C}(\pi))$ $>t|\mathcal{I}]$ versus $t$ in a distributed queueing system with per-task data locality constraints. The system has $g=2$ groups of servers, each consisting of 3 servers. The size of each sub-job $k_{ih}$ is  $1$ or $10$ with equal probability. The distributions of task service times and cancellation overheads of one group of servers are the same with those in Fig. \ref{fig1}, and the distributions of task service times and cancellation overheads of the other  group of servers are the same with those in Fig. \ref{fig2_exp}. 
The ``Lower bound'' curve is generated by using $\Pr[L_{\max}(\bm{V}(\text{EDD-GR}))>t|\mathcal{I}]$, which, according to Theorem \ref{thm_dis4}, is a lower bound of the optimum delay performance. We can observe that The delay performance of  policy EDD-GR is close to the lower bound curve. In addition, policy EDD-NR and EDD-GR have similar performance. This is because exponential distribution is also NBU. Hence, both groups of servers have NBU service time distributions.

\ifreport
\begin{figure}
\centering \includegraphics[width=0.6\textwidth]{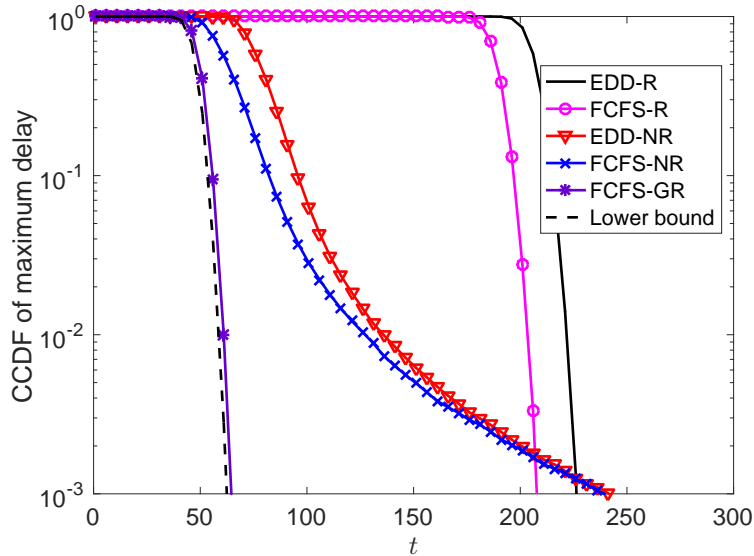} \caption{Complementary CDF of the maximum delay $\Pr[D_{\max} (\bm{C}(\pi))>t|\mathcal{I}]$ versus $t$ in a distributed queueing system with per-task data locality constraints.}
\label{fig8} \vspace{-0.cm}
\end{figure}
\fi
Figure \ref{fig8} depicts the complementary CDF of maximum delay $\Pr[D_{\max} (\bm{C}(\pi))$ $>t|\mathcal{I}]$ versus $t$ in a distributed queueing system with per-task data locality constraints. This system has $g=3$ groups of servers, each consisting of 3 servers. The size of each sub-job $k_{ih}$ is  $1$ or $10$ with equal probability. The distributions of task service times and cancellation overheads of two groups of servers are the same with those in Fig. \ref{fig7}, and the distributions of task service times and cancellation overheads of the third group of servers are the same with those in Fig. \ref{fig2}. 
The ``Lower bound'' curve is generated by using $\Pr[D_{\max}(\bm{V}(\text{FCFS-GR}))>t|\mathcal{I}]$, which, according to Corollary \ref{thm_dis5}, is a lower bound of the optimum delay performance. We can see that the delay performance of  policy FCFS-GR is quite close to the lower bound curve and is much better than that of the other policies. 

%% file: conclusion.tex
\section{Conclusion}\label{sec_conclusion}
This paper presented a comprehensive study on delay-optimal scheduling of batch jobs with replications in multi-server  systems. A number of low-complexity scheduling policies are developed and are proven to be (near) delay-optimal in a stochastic ordering sense
for minimizing three general classes of delay metrics among all causal and non-preemptive policies. The key  tools in our proofs are new sample-path conditions for comparing the delay performance of different policies. 
These sample-path conditions do not need to specify the queueing system model and hence can potentially be  applied to obtain (near) delay-optimal results for other scheduling systems. An interesting topic for future research is to develop an analytical framework to design (near) delay-optimal scheduling of replications and coding under general service time distributions (that go beyond NBU and NWU). In addition, service time correlation across the servers has significant influence on the delay performance of replications and coding, which requires further investigation.


%% file: sample_path.tex
\section{Sample-path  Method}\label{sec_proofmain}

We propose a unified sample-path method to prove the theorems in Section \ref{sec_analysis}. This method contains three steps: 

\begin{itemize}
\item[\emph{1.}] \textbf{Sample-path Orderings:} We first introduce several  sample-path orderings (Propositions \ref{ordering_1}-\ref{ordering_3_1} and Corollaries \ref{ordering_4}-\ref{ordering_4_1}).
Each of these sample-path orderings can be used to obtain a delay inequality for comparing the delay performance (i.e., average delay, maximum lateness, or maximum delay) of different policies.
\item[\emph{2.}]  \textbf{Sufficient Conditions for Sample-path Orderings:} As we have mentioned, each task has many replication modes, which require different amounts of service time. In order to minimize delay, the scheduler needs to choose efficient replication modes to execute the tasks as fast as possible. Motivated by this, we introduce two work-efficiency orderings to compare the efficiency of task executions in different policies. By combining these  work-efficiency orderings with appropriate priority rules (i.e., FUT, EDD, FCFS) for job services, we obtain several sufficient conditions  (Propositions \ref{lem1_0}-\ref{lem2}) of the sample-path orderings in \emph{Step 1}. In addition, if more than one sufficient conditions are simultaneously satisfied, we are able to obtain delay inequalities for comparing more general classes of delay metrics achieved by different policies (Propositions \ref{lem1_1_0}-\ref{coro2}).

\item[\emph{3.}]  \textbf{Coupling Arguments:} We use coupling arguments to prove that for NBU and NWU task service times and appropriate replication rules, the work-efficiency orderings are satisfied in the sense of stochastic ordering (Lemmas \ref{lem_coupling}-\ref{lem_coupling_3}). By combining this with the priority rules (i.e., FUT, EDD, FCFS), we are able to prove the sufficient conditions in \emph{Step 2} in the sense of stochastic ordering. 
By this, the main results of this paper are proven. 
\end{itemize}

This sample-path method is quite general. In particular, \emph{Step 1} and \emph{Step 2} do not need to specify the queueing system model, and can be potentially used for establishing (near) delay optimality results in other systems. 

\subsection{Step 1: Sample-path Orderings}
\label{sec_sufficient}

We first propose several sample-path orderings to compare the delay performance of different scheduling policies. Let us first define the system state of any policy $\pi\in\Pi$. 

\begin{definition}\label{def_state_thm3}
At any time instant $t\in[0,\infty)$, the system state of policy $\pi$ is specified by a pair of $n$-dimensional vectors $\bm{\xi}_{\pi}(t)=(\xi_{1,\pi}(t),\ldots,\xi_{n,\pi}(t))$ and $\bm{\gamma}_{\pi}(t)=(\gamma_{1,\pi}(t),\ldots,$ $\gamma_{n,\pi}(t))$ with non-negative components, where $n$ is the total number of jobs and can be either finite or infinite. The components of $\bm{\xi}_{\pi}(t)$ and $\bm{\gamma}_{\pi}(t)$ are interpreted as follows:
If job $i$ is present in the system at time $t$, then $\xi_{i,\pi}(t)$ is the number of \emph{remaining} tasks (which are either stored in the queue or being executed by the servers) of job $i$, and $\gamma_{i,\pi}(t)$ is the number of \emph{unassigned} tasks (which are stored in the queue and not being executed by any server) of job $i$; if job $i$ is not present in the system at time $t$ (i.e., job $i$ has not arrived at the system or has departed from the system), then $\xi_{i,\pi}(t)=\gamma_{i,\pi}(t)=0$. Hence, for all $i=1,\ldots,n$, $\pi\in\Pi$, and $t\in[0,\infty)$
\begin{align}\label{def_relation}
0\leq \gamma_{i,\pi}(t)\leq \xi_{i,\pi}(t)\leq k_i.
\end{align}

Let $\{\bm{\xi}_{\pi}(t),\bm{\gamma}_{\pi}(t),t\in[0,\infty)\}$ denote the state process of policy $\pi$ in a probability space $(\Omega,\mathcal{F},P)$, which is assumed to be right-continuous. The realization of the state process on a sample path $\omega\in \Omega$ can be expressed as $\{\bm{\xi}_{\pi}(\omega,t),$ $\bm{\gamma}_{\pi}(\omega,t),t\in[0,\infty)\}$. To ease the notational burden, we will omit $\omega$ henceforth and reuse $\{\bm{\xi}_{\pi}(t),\bm{\gamma}_{\pi}(t),t\in[0,\infty)\}$ to denote the realization of the state process on a sample path. Because the system starts to operate at time $t=0$,  there is no job in the system before time $t=0$. Hence, 
$\bm{\xi}_{\pi}(0^-) = \bm{\gamma}_{\pi}(0^-)=\bm{0}$. 
\end{definition}

%





The following proposition provides one condition \eqref{eq_ordering_1_1} for comparing the average delay $D_{\text{avg}}(\bm{c}(\pi))$ of different policies on a sample path, which was firstly introduced in \cite{Smith78} to prove the optimality of the preemptive SRPT policy for minimizing the average delay in single-server scheduling problems.
 \begin{proposition} \label{ordering_1} 
For any given job parameters $\mathcal{I}$ and a sample path of two policies $P,\pi\in\Pi$, 
if \footnote{In majorization theory \cite{Marshall2011},  \eqref{eq_ordering_1_1} is equivalent to ``$\bm{\xi}_{\pi}(t)$ is weakly supermajorized by $\bm{\xi}_{P}(t)$, i.e., $\bm{\xi}_{\pi}(t)\prec^{\text{w}}\bm{\xi}_{P}(t)$''.}  
\begin{eqnarray}\label{eq_ordering_1_1}
\sum_{i=j}^n {\xi}_{[i],P}(t)\leq \sum_{i=j}^n {\xi}_{[i],\pi}(t),~\forall~j = 1,2,\ldots,n
\end{eqnarray}
holds for all $t\in[0,\infty)$, where ${\xi}_{[i],\pi}(t)$ is the $i$-th largest component of $\bm{\xi}_{\pi}(t)$, then  
\begin{align}\label{eq_ordering_1_1_2}
c_{(i)}(P)\leq c_{(i)}(\pi),~\forall~i=1,2,\ldots,n,
\end{align}
where $c_{(i)}(\pi)$ is the $i$-th smallest component of $\bm{c}(\pi)$.\footnote{In other words, $c_{(i)}(\pi)$ is the earliest time in policy $\pi$ by which $i$ jobs have been completed in policy $\pi$.} Hence,
\emph{
\begin{align}\label{eq_ordering_1_2}
D_{\text{avg}}(\bm{c}(P))\leq D_{\text{avg}}(\bm{c}(\pi)).
\end{align}
}
\end{proposition}

\begin{proof}
Suppose that there are $l$ unfinished jobs at time $t$ in policy $\pi$, then $\sum_{i=l+1}^n {\xi}_{[i],\pi}(t)=0$. By \eqref{eq_ordering_1_1}, we get $\sum_{i=l+1}^n {\xi}_{[i],P}(t)=0$ and hence there are at most $l$ unfinished jobs in policy $P$. In other words, there are at least as many unfinished jobs in policy $\pi$ as in policy $P$ at any time $t\in[0,\infty)$. This implies \eqref{eq_ordering_1_1_2}, because the sequence of job arrival times $a_1,a_2,\ldots, a_n$ are invariant under any policy. In addition, \eqref{eq_ordering_1_2} follows from \eqref{eq_ordering_1_1_2}, which completes the proof.
\end{proof}

The sample-path ordering  \eqref{eq_ordering_1_1} is quite insightful. According to Proposition \ref{ordering_1}, if \eqref{eq_ordering_1_1} holds for all policies $\pi\in\Pi$ and all sample paths $\omega\in \Omega$, then policy $P$ is sample-path delay-optimal for minimizing the average delay $D_{\text{avg}}(\bm{C}(\pi))$. 
Interestingly, Proposition \ref{ordering_1} is also necessary: If \eqref{eq_ordering_1_1} does not hold at some time $t$, then one can construct an arrival process after time $t$ such that \eqref{eq_ordering_1_1_2} and \eqref{eq_ordering_1_2} do not hold  \cite{Smith78}. 

The sample-path ordering  \eqref{eq_ordering_1_1} has been successfully used in single-server scheduling problems \cite{Smith78}. However, it cannot be directly applied in multi-server scheduling problems. 
In the sequel, we consider an alternative method to relax the sample-path ordering \eqref{eq_ordering_1_1} and seek for near delay optimality.

 \begin{proposition} \label{ordering_2} 
For any given job parameters $\mathcal{I}$ and a sample path of two policies $P,\pi\in\Pi$, if 
\begin{eqnarray}\label{eq_ordering_2_1}
\sum_{i=j}^n {\gamma}_{[i],P}(t)\leq \sum_{i=j}^n {\xi}_{[i],\pi}(t),~\forall~j = 1,2,\ldots,n
\end{eqnarray}
holds for all $t\in[0,\infty)$, where ${\gamma}_{[i],\pi}(t)$ is the $i$-th largest component of $\bm{\gamma}_{\pi}(t)$, then 
\begin{align}\label{eq_ordering_2_2}
v_{(i)}(P)\leq c_{(i)}(\pi),~\forall~i=1,2,\ldots,n,
\end{align}
where $v_{(i)}(P)$ is the $i$-th smallest component of $\bm{v}(P)$.\footnote{In other words, $v_{(i)}(P)$ is the earliest time in policy $P$ that there exist $i$ jobs whose tasks have all started service.}
Hence, 
\emph{
\begin{align}\label{eq_ordering_2_3}
D_{\text{avg}} (\bm{v}(P))\leq D_{\text{avg}}(\bm{c}(\pi)).
\end{align}}
\end{proposition}


\begin{proof}
See Appendix \ref{app_lem2}. 
\end{proof}

Hence, by relaxing the sample-path ordering \eqref{eq_ordering_1_1} as \eqref{eq_ordering_2_1}, a relaxed delay inequality \eqref{eq_ordering_2_3} is obtained which can be used to compare the average delay of policy $P$ and policy $\pi$ in a near-optimal sense. 



Similarly, two sample-path orderings are developed in the following two lemmas to compare the maximum lateness $L_{\max}(\cdot)$ achieved by different policies.




 \begin{proposition} \label{ordering_3} 
For any given job parameters $\mathcal{I}$ and a sample path of two policies $P,\pi\in\Pi$, if 
\begin{eqnarray}\label{eq_ordering_3_1}
\sum_{i: d_i \leq \tau} {\xi}_{i,P}(t)\leq \sum_{i: d_i \leq \tau} {\xi}_{i,\pi}(t),~\forall~\tau\in[0,\infty)
\end{eqnarray}
holds for all $t\in[0,\infty)$, then 
\emph{
\begin{align}\label{eq_ordering_3_2}
L_{\max} (\bm{c}(P))\leq L_{\max}(\bm{c}(\pi)).
\end{align}}
\end{proposition}
\begin{proof}
See Appendix \ref{app_lem3}. 
\end{proof}
\begin{proposition} \label{ordering_3_1} 
For any given job parameters $\mathcal{I}$ and a sample path of two policies $P,\pi\in\Pi$, if 
\begin{eqnarray}\label{eq_ordering_3_3}
\sum_{i: d_i \leq \tau} {\gamma}_{i,P}(t)\leq \sum_{i: d_i \leq \tau} {\xi}_{i,\pi}(t),~\forall~\tau\in[0,\infty)
\end{eqnarray}
holds for all $t\in[0,\infty)$, then 
\emph{
\begin{align}\label{eq_ordering_3_4}
L_{\max} (\bm{v}(P))\leq L_{\max}(\bm{c}(\pi)).
\end{align}} 
\end{proposition}

\begin{proof}
See Appendix \ref{app_lem3}. 
\end{proof}

%
%
%

If $d_i=a_i$ for all $i$,  the maximum lateness $L_{\max}(\cdot)$ reduces to the {maximum delay} $D_{\max}(\cdot)$. Hence, we can obtain 

 \begin{corollary} \label{ordering_4} 
For any given job parameters $\mathcal{I}$ and a sample path of two policies $P,\pi\in\Pi$, if 
\begin{eqnarray}\label{eq_ordering_4_1}
\sum_{i: a_i \leq \tau} {\xi}_{i,P}(t)\leq \sum_{i: a_i \leq \tau} {\xi}_{i,\pi}(t),~\forall~\tau\in[0,\infty)
\end{eqnarray}
holds for all $t\in[0,\infty)$, then 
\emph{
\begin{align}\label{eq_ordering_4_2}
D_{\max} (\bm{c}(P))\leq D_{\max}(\bm{c}(\pi)).
\end{align}}
\end{corollary}
\begin{corollary} \label{ordering_4_1} 
For any given job parameters $\mathcal{I}$ and a sample path of two policies $P,\pi\in\Pi$, if
\begin{eqnarray}\label{eq_ordering_4_3}
\sum_{i: a_i \leq \tau} {\gamma}_{i,P}(t)\leq \sum_{i: a_i \leq \tau} {\xi}_{i,\pi}(t),~\forall~\tau\in[0,\infty)
\end{eqnarray}
holds for all $t\in[0,\infty)$, then 
\emph{
\begin{align}\label{eq_ordering_4_4}
D_{\max} (\bm{v}(P))\leq D_{\max}(\bm{c}(\pi)).
\end{align}}
\end{corollary}

The proofs of Corollary \ref{ordering_4} and Corollary \ref{ordering_4_1} are omitted, because they follow directly from Proposition \ref{ordering_3} and Proposition \ref{ordering_3_1} by setting $d_i=a_i$ for all $i$. Nonetheless, due to the importance of the maximum delay metric, Corollary \ref{ordering_4} and Corollary \ref{ordering_4_1} are of independent interests.

The sample-path orderings in Propositions \ref{ordering_1}-\ref{ordering_3_1} and Corollaries \ref{ordering_4}-\ref{ordering_4_1} are of similar forms. Their distinct features are 
\begin{itemize} 
\item In the sample-path orderings \eqref{eq_ordering_1_1} and \eqref{eq_ordering_2_1} corresponding to the average delay $D_{\text{avg}}(\cdot)$, the summations are taken over the jobs with the fewest remaining/unassigned tasks;
 
\item In the sample-path orderings \eqref{eq_ordering_3_1} and \eqref{eq_ordering_3_3} corresponding to the maximum lateness $L_{\max}(\cdot)$, the summations are taken over the jobs with the earliest due times;
 
\item In the sample-path orderings \eqref{eq_ordering_4_1} and \eqref{eq_ordering_4_3} corresponding to the maximum delay $D_{\max}(\cdot)$, the summations are taken over the jobs with the earliest arrival times.
\end{itemize} 
These features are tightly related to the priority rules for minimizing the corresponding delay metrics: The priority rule for minimizing the average delay $D_{\text{avg}}(\cdot)$ is 
 FUT first; the priority rule for minimizing the maximum lateness $L_{\max}(\cdot)$ is EDD first; the priority rule for minimizing the maximum delay $D_{\max}(\cdot)$ is FCFS. Hence, \emph{the summations in these sample-path orderings are taken over the high priority jobs}. This is the key insight behind these sample-path orderings.

A number of popular sample-path methods --- such as forward induction, backward induction, and interchange arguments  \cite{Liu1995} --- have been successfully used to establish delay optimality results in single-server scheduling problems \cite{Schrage68,Jackson55,Baccelli:1993}. However, it is challenging to directly generalize these methods and characterize the sub-optimal delay gap from the optimum when delay optimality is essentially difficult to achieve. 
On the other hand, the sample-path orderings \eqref{eq_ordering_1_1}, \eqref{eq_ordering_2_1}, \eqref{eq_ordering_3_1}, \eqref{eq_ordering_3_3}, \eqref{eq_ordering_4_1}, and \eqref{eq_ordering_4_3} provide an interesting unified framework for  sample-path delay comparisons towards both delay optimality and near delay optimality. 
To the best of our knowledge, except for \eqref{eq_ordering_1_1} developed in \cite{Smith78},  the sample-path orderings \eqref{eq_ordering_2_1}, \eqref{eq_ordering_3_1}, \eqref{eq_ordering_3_3}, \eqref{eq_ordering_4_1}, and \eqref{eq_ordering_4_3} have not appeared before.

\subsection{Step 2: Sufficient Conditions for Sample-path Orderings}

In \emph{Step 2}, we will introduce several sufficient conditions for the sample-path orderings \eqref{eq_ordering_1_1}, \eqref{eq_ordering_2_1}, \eqref{eq_ordering_3_1}, \eqref{eq_ordering_3_3}, \eqref{eq_ordering_4_1}, and \eqref{eq_ordering_4_3}. In addition, we will also develop sample-path sufficient conditions for comparing more general delay metrics in $\mathcal{D}_{\text{sym}}$, $\mathcal{D}_{\text{Sch-1}}$, and $\mathcal{D}_{\text{Sch-2}}$. 

\subsubsection{ Work-efficiency Orderings}

In traditional queueing systems without replications, the service delay is largely governed by the work conservation law (or its generalizations): At any time, the expected total amount of time for completing the jobs in the queue is invariant among all work-conserving policies \cite{Leonard_Kleinrock_book,Jose2010,Gittins:11}. However, this work conservation law does not hold in queueing systems with replications.
In particular, each task has many replication modes (i.e., it can be replicated on different sets of servers and at different time instants), which require different amounts of service time. In order to minimize delay, the scheduler needs to choose efficient replication modes to execute the tasks. Motivated by this, we introduce an ordering to compare the efficiency of task executions in different policies. We call it \emph{work-efficiency ordering}.

Let $k_{\text{sum}}=\sum_{i=1}^n k_i$ denote the total number of tasks of all jobs. Define $\bm{T}_\pi=(T_{1,\pi},\ldots,$ $ T_{k_{\text{sum}},\pi})$ as the sequence of task completion times in policy $\pi$ where $T_{1,\pi}\leq \ldots\leq T_{k_{\text{sum}},\pi}$. 
Let $\bm{t}_\pi=(t_{1,\pi},\ldots, t_{k_{\text{sum}},\pi})$ denote the realization of $\bm{T}(\pi)$ on a sample path.
\begin{definition} \label{def_order} \emph{Work-Efficiency Ordering:}
For given job parameters $\mathcal{I}$ and a sample path of two policies $P,\pi\in\Pi$, policy $P$ is said to be \emph{more work-efficient than} policy $\pi$, if
\begin{align}\label{eq_work_efficient_order_1_1}
\bm{t}_P \leq\bm{t}_\pi.
\end{align}
\end{definition}

The key idea of this work-efficiency ordering is \emph{to complete tasks as early as possible.} This idea was  used  to study delay-optimal replications in \cite{Borst2003,Righter2008,Kim2010} where each job has a single task, i.e., $k_1 = \cdots=k_n=1$.

In some scenarios, the above work-efficiency ordering is not satisfied, but it is possible to establish the following alternative form of work-efficiency ordering:
\begin{figure}
\centering 
\includegraphics[width=0.4\textwidth]{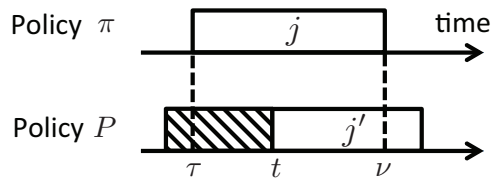} \caption{Illustration of the weak work-efficiency ordering, where the service duration of a task (i.e., the time duration since one copy of the task starts service until one copy of the task is completed) is indicated by a rectangle.
Task $j$ starts service at time $\tau$ and completes service at time $\nu$ in policy $\pi$, and one corresponding task $j'$ starts service at  time $t\in[\tau,\nu]$ in policy $P$.}
\label{Work_Efficiency_Ordering} 
\end{figure} 

\begin{definition} \label{def_order} \emph{Weak Work-efficiency Ordering:}
For any given job parameters $\mathcal{I}$ and a sample path of two policies $P,\pi\in\Pi$, policy $P$ is said to be \emph{weakly more work-efficient than} policy $\pi$, if the following assertion is true:
For each task $j$ executed in policy $\pi$, if
\begin{itemize}
\item[1.] In policy $\pi$, task $j$ starts service at time $\tau$ and completes service at time $\nu$ ($\tau\leq \nu$), 
\item[2.] In policy $P$, the queue is not empty (there exist unassigned tasks in the queue) during $[\tau,\nu]$, 
\end{itemize}
then there exists one corresponding task $j'$ in policy $P$ which starts service during $[\tau,\nu]$. \end{definition}
An illustration of this weak work-efficiency ordering is provided in Fig. \ref{Work_Efficiency_Ordering}. Notice that this weak work-efficient ordering requires the service starting time of task $j'$ in policy $P$ to be within the service duration of its corresponding task $j$ in policy $\pi$. This is a key feature that will be used later to establish near delay optimality. 

We note that the weak work-efficiency ordering does not follow from the work-efficiency ordering. We say it is \emph{weak} in the sense that work-efficiency ordering can be used to establish delay optimality, while {weak} work-efficiency ordering can be used to establish near delay optimality. 
\subsubsection{Sufficient Conditions for  Sample-path Orderings}
Using these two work-efficiency orderings, we can obtain 
the following sufficient conditions for the sample-path ordering \eqref{eq_ordering_1_1} and \eqref{eq_ordering_2_1} associated to the average delay $D_{\text{avg}}(\cdot)$.

\begin{proposition}\label{lem1_0}
For given job parameters $\mathcal{I}$ and a sample path of two policies $P,\pi\in\Pi$, if
\begin{itemize}
\itemsep0em 
\item[1.] $k_1\leq k_2\leq \ldots\leq k_n$, 
\item[2.] Policy $P$ is {more work-efficient} than policy $\pi$,
\item[3.] In policy $P$, each task completing service is from the job with the fewest remaining tasks among all jobs with remaining tasks, 
\end{itemize} 
then \eqref{eq_ordering_1_1}-\eqref{eq_ordering_1_2} hold.  
\end{proposition}
\begin{proof}
See Appendix \ref{app_lem1_0}.
\end{proof}

\begin{proposition}\label{lem1}
For given job parameters $\mathcal{I}$ and a sample path of two policies $P,\pi\in\Pi$, if
\begin{itemize}
\itemsep0em 
\item[1.] Policy $P$ is \textbf{weakly} {more work-efficient} than policy $\pi$,
\item[2.] In policy $P$, each task starting service is from the job with the fewest unassigned tasks among all jobs with unassigned tasks, 
\end{itemize} 
then \eqref{eq_ordering_2_1}-\eqref{eq_ordering_2_3} hold.  
\end{proposition}
\begin{proof}
See Appendix \ref{app0}.
\end{proof}

Similarly, two sufficient conditions are obtained for the sample-path orderings \eqref{eq_ordering_3_1} and \eqref{eq_ordering_3_3} for comparing the maximum lateness $L_{\max}(\cdot)$ of different policies.

\begin{proposition}\label{lem2_0}
For given job parameters $\mathcal{I}$ and a sample path of two policies $P,\pi\in\Pi$, if
\begin{itemize}
\itemsep0em 
\item[1.] $d_1\leq d_2\leq \ldots\leq d_n$, 
\item[2.] Policy $P$ is {more work-efficient} than policy $\pi$,
\item[3.] In policy $P$, each task completing service is from the job with the earliest due time among  all jobs with remaining tasks, 
\end{itemize} 
then \eqref{eq_ordering_3_1} and \eqref{eq_ordering_3_2} hold.  
\end{proposition}
\begin{proof}
See Appendix \ref{app_lem2_0}.
\end{proof}

\begin{proposition}\label{lem2}
For given job parameters $\mathcal{I}$ and a sample path of two policies $P,\pi\in\Pi$, if
\begin{itemize}
\itemsep0em 
\item[1.] Policy $P$ is \textbf{weakly} {more work-efficient} than policy $\pi$,
\item[2.] In policy $P$, each task starting service is from the job with the earliest due time among all jobs with unassigned tasks, 
\end{itemize} 
then \eqref{eq_ordering_3_3} and \eqref{eq_ordering_3_4} hold. \end{proposition}

\begin{proof}
See Appendix \ref{app0_1}. 
\end{proof}

%
%
%
%
%
%

\subsubsection{More General Delay Metrics}
We now investigate more general delay metrics in $\mathcal{D}_{\text{sym}}$ and $\mathcal{D}_{\text{Sch-1}}$.
First, Proposition \ref{lem1_0} and Proposition \ref{lem1} can be directly generalized to all delay metrics in $\mathcal{D}_{\text{sym}}$.

\begin{proposition}\label{lem1_1_0}
If the conditions of Proposition \ref{lem1_0} are satisfied, then for all \emph{$f\in\mathcal{D}_{\text{sym}}$}
\begin{align}
f (\bm{c}(P))\leq f(\bm{c}(\pi)).\nonumber
\end{align}
\end{proposition}
\begin{proof}
See Appendix \ref{app_lem1_1}.
\end{proof}
\begin{proposition}\label{lem1_1}
If the conditions of Proposition \ref{lem1} are satisfied, then for all \emph{$f\in\mathcal{D}_{\text{sym}}$}
\begin{align}
f (\bm{v}(P))\leq f(\bm{c}(\pi)).\nonumber
\end{align}
\end{proposition}
\begin{proof}
See Appendix \ref{app_lem1_1}.
\end{proof}

If policy $P$ simultaneously satisfies the sufficient conditions in Proposition \ref{lem1_0} and Proposition \ref{lem2_0} (or Proposition \ref{lem1} and Proposition \ref{lem2}), we can obtain a couple of delay inequalities for comparing any delay metric in $\mathcal{D}_{\text{sym}}\cup\mathcal{D}_{\text{Sch-1}}$.

\begin{proposition}\label{coro2_0}
If the conditions of Proposition \ref{lem1_0} and Proposition \ref{lem2_0} are simultaneously satisfied, then for all \emph{$f\in\mathcal{D}_{\text{sym}}\cup\mathcal{D}_{\text{Sch-1}}$}
\begin{align}\label{eq_lem_general}
f (\bm{c}(P))\leq f(\bm{c}(\pi)).
\end{align}
\end{proposition}
\begin{proof}[Proof sketch of Proposition \ref{coro2_0}]
For any $f\in\mathcal{D}_{\text{sym}}$, \eqref{eq_ordering_1_1_2} and \eqref{eq_lem_general} follow from Proposition \ref{lem1_0} and Proposition \ref{lem1_1_0}.
For any $f\in\mathcal{D}_{\text{Sch-1}}$, we construct an $n$-dimensional vector $\bm{c}'$ and show that  
\begin{align}\label{eq_coro2_6_0}
\bm{c}(P)-\bm{d}\prec \bm{c}'- \bm{d} \leq \bm{c} (\pi)-\bm{d},
\end{align}
where the first majorization ordering in \eqref{eq_coro2_6_0} follows from the rearrangement inequality \cite[Theorem 6.F.14]{Marshall2011},\cite{Chang93rearrangement_majorization}, and the second inequality in  \eqref{eq_coro2_6_0} is proven by using  \eqref{eq_ordering_1_1_2}.
This further implies 
\begin{align}\label{eq_coro2_3_0}
\bm{c}(P)-\bm{d}\prec_{\text{w}} \bm{c} (\pi)-\bm{d}. 
\end{align} 
Using this, we can show that \eqref{eq_lem_general} holds for all $f\in\mathcal{D}_{\text{Sch-1}}$.
The details are provided in Appendix \ref{app_coro2_0}.
\end{proof}

\begin{proposition}\label{coro2}
If the conditions of Proposition \ref{lem1} and Proposition \ref{lem2} are simultaneously satisfied, then for all \emph{$f\in\mathcal{D}_{\text{sym}}\cup\mathcal{D}_{\text{Sch-1}}$}
\begin{align}\label{eq_coro2}
f (\bm{v}(P))\leq f(\bm{c}(\pi)).
\end{align}
\end{proposition}
\begin{proof}
See Appendix \ref{app_coro2}.
\end{proof}

\subsection{Step 3: Coupling Arguments}\label{sec_delay_ineq}
\subsubsection{Coupling Lemmas}
We need the following three coupling lemmas to prove our main results.

\begin{lemma}\label{lem_coupling}
Consider policy $P\in \Pi$ and any policy $\pi\in \Pi$. If (i) policy $P$ follows the LPR discipline, (ii) the task service times are NBU, independent across the servers, and {i.i.d.} across the tasks assigned to the same server, 
 then there exist policy $P_1$ and  policy $\pi_1$ satisfying the same queueing disciplines with policy $P$ and policy $\pi$, respectively,  such that 
\begin{itemize}
\itemsep0em 
\item[1.] The state process $\{\bm{\xi}_{P_1}(t),\bm{\gamma}_{P_1}(t),t\in[0,\infty)\}$ of policy $P_1$ has the same distribution with the state process $\{\bm{\xi}_{P}(t),\bm{\gamma}_{P}(t),t\in[0,\infty)\}$ of policy $P$,
\item[2.] The state process $\{\bm{\xi}_{\pi_1}(t),\bm{\gamma}_{\pi_1}(t),t\in[0,\infty)\}$ of policy $\pi_1$ has the same distribution with the state process $\{\bm{\xi}_{\pi}(t),\bm{\gamma}_{\pi}(t),t\in[0,\infty)\}$  of policy $\pi$,
\item[3.] Policy $P_1$ is \textbf{weakly} more work-efficient than policy $\pi_1$ with probability one. 
\end{itemize} 
\end{lemma}

\begin{proof}
See Appendix \ref{app1}.
\end{proof}

\begin{lemma}\label{lem_coupling_2}
Consider policy $P\in \Pi$ and any policy $\pi\in \Pi$. If (i) policy $P$ follows the R discipline, (ii) the task service times are NWU, independent across the servers, and {i.i.d.} across the tasks assigned to the same server, and (iii) the cancellation overhead is $\bm{O}=\bm{0}$, 
 then there exist policy $P_1$ and  policy $\pi_1$ satisfying the same queueing disciplines with policy $P$ and policy $\pi$, respectively, such that 
\begin{itemize}
\itemsep0em 
\item[1.] The state process $\{\bm{\xi}_{P_1}(t),\bm{\gamma}_{P_1}(t),t\in[0,\infty)\}$ of policy $P_1$ has the same distribution with the state process $\{\bm{\xi}_{P}(t),\bm{\gamma}_{P}(t),t\in[0,\infty)\}$ of policy $P$,
\item[2.] The state process $\{\bm{\xi}_{\pi_1}(t),\bm{\gamma}_{\pi_1}(t),t\in[0,\infty)\}$ of policy $\pi_1$ has the same distribution with the state process $\{\bm{\xi}_{\pi}(t),\bm{\gamma}_{\pi}(t),t\in[0,\infty)\}$  of policy $\pi$,
\item[3.] Policy $P_1$ is  more work-efficient than policy $\pi_1$ with probability one. 
\end{itemize} 
\end{lemma}

\begin{proof}
See Appendix \ref{app2}.
\end{proof}

\begin{lemma}\label{lem_coupling_3}
Consider policy $P\in \Pi$ and any policy $\pi\in \Pi$. If (i) policy $P$ follows the R discipline, (ii) the task service times are exponential, independent across the servers, and {i.i.d.} across the tasks assigned to the same server, and (iii) the cancellation overhead is $\bm{O}=\bm{0}$, 
 then there exist policy $P_1$ and  policy $\pi_1$ satisfying the same queueing disciplines with policy $P$ and policy $\pi$, respectively, such that 
\begin{itemize}
\itemsep0em 
\item[1.] The state process $\{\bm{\xi}_{P_1}(t),\bm{\gamma}_{P_1}(t),t\in[0,\infty)\}$ of policy $P_1$ has the same distribution with the state process $\{\bm{\xi}_{P}(t),\bm{\gamma}_{P}(t),t\in[0,\infty)\}$ of policy $P$,
\item[2.] The state process $\{\bm{\xi}_{\pi_1}(t),\bm{\gamma}_{\pi_1}(t),t\in[0,\infty)\}$ of policy $\pi_1$ has the same distribution with the state process $\{\bm{\xi}_{\pi}(t),\bm{\gamma}_{\pi}(t),t\in[0,\infty)\}$  of policy $\pi$,
\item[3.] Policy $P_1$ is \textbf{weakly} more work-efficient than policy $\pi_1$ with probability one. 
\end{itemize} 
\end{lemma}

\begin{proof}
See Appendix \ref{app2_1}.
\end{proof}

We note that Theorem 6.B.3 in \cite{StochasticOrderBook} plays an important role in the proofs of Lemmas \ref{lem_coupling}-\ref{lem_coupling_3}: Because the task service times are \emph{independent} across the servers and \emph{i.i.d.} across the tasks assigned to the same server, we only need the NBU/NWU assumption and Theorem 6.B.3 in \cite{StochasticOrderBook}, instead of invoking the stronger likelihood ratio ordering  as in 
\cite{Chang93rearrangement_majorization}, \cite[Theorem 6.B.15]{StochasticOrderBook}, to prove these coupling lemmas.

\subsubsection{Proofs of the Main Results}

Now, we are ready to prove the main results. 

\begin{proof}[Proof of Theorem \ref{thm1}]
 According to lemma \ref{lem_coupling}, for any policy $\pi\in\Pi$, there exist two state processes $\{\bm{\xi}_{\text{FUT-LPR}_1}(t), $ $\bm{\gamma}_{\text{FUT-LPR}_1}(t),t\in[0,\infty)\}$ and $\{\bm{\xi}_{\pi_1}(t),\bm{\gamma}_{\pi_1}(t),t\in[0,\infty)\}$ of policy FUT-LPR$_1$ and  policy $\pi_1$, such that (i) the state process $\{\bm{\xi}_{\text{FUT-LPR}_1}(t),\bm{\gamma}_{\text{FUT-LPR}_1}(t),t\in[0,\infty)\}$ of policy FUT-LPR$_1$  has the same distribution with the state process $\{\bm{\xi}_{\text{FUT-LPR}}(t),\bm{\gamma}_{\text{FUT-LPR}}(t),t\in[0,\infty)\}$ of policy FUT-LPR, (ii) the state process  $\{\bm{\xi}_{\pi_1}(t),\bm{\gamma}_{\pi_1}(t),t\in[0,\infty)\}$ of policy   $\pi_1$ has the same distribution with the state process $\{\bm{\xi}_{\pi}(t),\bm{\gamma}_{\pi}(t),t\in[0,\infty)\}$ of policy $\pi$, and (iii) policy FUT-LPR$_1$ is weakly more work-efficient than policy $\pi_1$ with probability one.

By (iii), the scheduling decisions of policy FUT-LPR$_1$, and Proposition \ref{lem1_1}, for all $f\in\mathcal{D}_{\text{sym}}$
\begin{align}
\Pr[f(\bm{V}(\text{FUT-LPR}_1)) \leq f(\bm{C}(\pi_1)) |\mathcal{I}]=1. \nonumber
\end{align}
By (i), $f(\bm{V}(\text{FUT-LPR}_1))$ has the same distribution with $f(\bm{V}(\text{FUT-LPR}))$. By (ii), $f(\bm{C}(\pi_1))$ has the same distribution with $f(\bm{C}(\pi))$.
Using the property of stochastic ordering \cite[Theorem 1.A.1]{StochasticOrderBook}, we can obtain \eqref{eq_delaygap1}. This completes the proof.
\end{proof}

\begin{proof}[Proof of Theorem \ref{lem7_NBU}]
After Theorem \ref{thm1} is established, Theorem \ref{lem7_NBU} is proven in Appendix \ref{app_lem7}.
\end{proof}

\begin{proof}[Proof of Theorem \ref{thm2}]
 According to lemma \ref{lem_coupling_2}, for any policy $\pi\in\Pi$, there exist two state processes $\{\bm{\xi}_{\text{FUT-R}_1}(t), $ $\bm{\gamma}_{\text{FUT-R}_1}(t),t\in[0,\infty)\}$ and $\{\bm{\xi}_{\pi_1}(t),\bm{\gamma}_{\pi_1}(t),t\in[0,\infty)\}$ of policy FUT-R$_1$ and  policy $\pi_1$, such that (i) the state process $\{\bm{\xi}_{\text{FUT-LPR}_1}(t),\bm{\gamma}_{\text{FUT-R}_1}(t),t\in[0,\infty)\}$ of policy FUT-R$_1$  has the same distribution with the state process $\{\bm{\xi}_{\text{FUT-R}}(t),$ $\bm{\gamma}_{\text{FUT-R}}(t),t\in[0,\infty)\}$ of policy FUT-R, (ii) the state process  $\{\bm{\xi}_{\pi_1}(t),\bm{\gamma}_{\pi_1}(t),t\in[0,\infty)\}$ of policy   $\pi_1$ has the same distribution with the state process $\{\bm{\xi}_{\pi}(t),\bm{\gamma}_{\pi}(t),t\in[0,\infty)\}$ of policy $\pi$, and (iii) policy FUT-R$_1$ is more work-efficient than policy $\pi_1$ with probability one.

In policy FUT-R$_1$, each task completing service is from the job with the fewest remaining tasks among all jobs with remaining tasks.
By (iii) and Proposition \ref{lem1_1_0}, for all $f\in\mathcal{D}_{\text{sym}}$
\begin{align}
\Pr[f(\bm{C}(\text{FUT-R}_1)) \leq f(\bm{C}(\pi_1)) |\mathcal{I}]=1. \nonumber
\end{align}
By (i), $f(\bm{C}(\text{FUT-R}_1))$ has the same distribution with $f(\bm{C}$ $(\text{FUT-R}))$. By (ii), $f(\bm{C}(\pi_1))$ has the same distribution with $f(\bm{C}(\pi))$.
Then, by the property of stochastic ordering \cite[Theorem 1.A.1]{StochasticOrderBook}, we can obtain \eqref{eq_delaygap3}. This completes the proof.
\end{proof}

\begin{proof}[Proof of Theorem \ref{thm2_exp}]
According to lemma \ref{lem_coupling_3}, for any policy $\pi\in\Pi$, there exist two state processes $\{\bm{\xi}_{\text{FUT-R}_1}(t), $ $\bm{\gamma}_{\text{FUT-R}_1}(t),t\in[0,\infty)\}$ and $\{\bm{\xi}_{\pi_1}(t),\bm{\gamma}_{\pi_1}(t),t\in[0,\infty)\}$ of policy FUT-R$_1$ and  policy $\pi_1$, such that (i) the state process $\{\bm{\xi}_{\text{FUT-LPR}_1}(t),$ $\bm{\gamma}_{\text{FUT-R}_1}(t),t\in[0,\infty)\}$ of policy FUT-R$_1$  has the same distribution with the state process $\{\bm{\xi}_{\text{FUT-R}}(t),$ $\bm{\gamma}_{\text{FUT-R}}(t),t\in[0,\infty)\}$ of policy FUT-R, (ii) the state process  $\{\bm{\xi}_{\pi_1}(t),\bm{\gamma}_{\pi_1}(t),$ $t\in[0,\infty)\}$ of policy   $\pi_1$ has the same distribution with the state process $\{\bm{\xi}_{\pi}(t),\bm{\gamma}_{\pi}(t),t\in[0,\infty)\}$ of policy $\pi$, and (iii) policy FUT-R$_1$ is weakly more work-efficient than policy $\pi_1$ with probability one.

By (iii), the scheduling decisions of policy FUT-R$_1$, and Proposition \ref{lem1_1}, for all $f\in\mathcal{D}_{\text{sym}}$
\begin{align}
\Pr[f(\bm{V}(\text{FUT-R}_1)) \leq f(\bm{C}(\pi_1)) |\mathcal{I}]=1. \nonumber
\end{align}
By (i), $f(\bm{V}(\text{FUT-R}_1))$ has the same distribution with $f(\bm{V}$ $(\text{FUT-R}))$. By (ii), $f(\bm{C}(\pi_1))$ has the same distribution with $f(\bm{C}(\pi))$.
Using the property of stochastic ordering \cite[Theorem 1.A.1]{StochasticOrderBook}, we can obtain \eqref{eq_delaygap3_exp}. This completes the proof.

In addition, \eqref{eq_gap_NWU} is proven in Appendix \ref{app_lem7_NWU}. This completes the proof.
\end{proof}

\begin{proof}[Proof of Theorem \ref{thm3}]
By replacing policy FUT-LPR, policy FUT-R$_1$,  and Proposition \ref{lem1_1}  in the proof of Theorem \ref{thm1} with policy EDD-LPR, policy EDD-R$_1$,  and Proposition \ref{lem2}, respectively,
Theorem \ref{thm3} is proven.
\end{proof}

\begin{proof}[Proof of Theorem \ref{coro_thm3_1}]
If $k_1=\ldots=k_n=1$, each job has only one task. Hence, the job with the earliest due time among all jobs with unassigned tasks is also one job with the fewest unassigned tasks. 

If $d_1\leq d_2\leq \ldots\leq d_n$, $k_1\leq k_2 \leq \ldots\leq k_n$, in policy EDD-LPR, each task starting service is from the job with the earliest due time among all jobs with unassigned tasks, which is also the job with the fewest unassigned tasks among all jobs with unassigned tasks. 

By this and replacing policy FUT-LPR, policy FUT-LPR$_1$, and Proposition \ref{lem1_1}  in the proof of Theorem \ref{thm1} with policy EDD-LPR, policy EDD-LPR$_1$, and Proposition \ref{coro2}, respectively,
Theorem \ref{coro_thm3_1} is proven.
\end{proof}

\begin{proof}[Proof of Theorem \ref{thm4}]
By replacing policy FUT-R, policy FUT-R$_1$, and Proposition \ref{lem1_1_0}  in the proof of Theorem \ref{thm2} with policy EDD-R, policy EDD-R$_1$, and Proposition \ref{lem2_0}, respectively,
Theorem \ref{thm4} is proven.
\end{proof}

\begin{proof}[Proof of Theorem \ref{coro4_1}]
If $d_1\leq d_2\leq \ldots\leq d_n$, $k_1\leq k_2 \leq \ldots\leq k_n$, in policy EDD-R, each task completing service is from the job with the earliest due time among all jobs with remaining tasks, which is also the job with the fewest remaining tasks among all jobs with remaining tasks. 

By this and replacing policy FUT-R, policy FUT-R$_1$, and Proposition \ref{lem1_1_0}  in the proof of Theorem \ref{thm2} with policy EDD-R, policy EDD-R$_1$, and Proposition \ref{coro2_0}, respectively,
Theorem \ref{coro4_1} is proven.
\end{proof}

\begin{proof}[Proof of Theorem \ref{thm4_exp}]
By replacing policy FUT-R, policy FUT-R$_1$, and Proposition \ref{lem1_1}  in the proof of Theorem \ref{thm2_exp} with policy EDD-R, policy EDD-R$_1$, and Proposition \ref{lem2}, respectively,
Theorem \ref{thm4_exp} is proven.
\end{proof}

\begin{proof}[Proof of Theorem \ref{coro_thm3_1_exp}]
If $k_1=\ldots=k_n=1$, each job has only one task. Hence, the job with the earliest due time among all jobs with unassigned tasks is also one job with the fewest unassigned tasks. 

If $d_1\leq d_2\leq \ldots\leq d_n$, $k_1\leq k_2 \leq \ldots\leq k_n$, in policy EDD-R, each task starting service is from the job with the earliest due time among all jobs with unassigned tasks, which is also the job with the fewest unassigned tasks among all jobs with unassigned tasks. 

By this and replacing policy FUT-R, policy FUT-R$_1$, and Proposition \ref{lem1_1}  in the proof of Theorem \ref{thm2_exp} with policy EDD-R, policy EDD-R$_1$, and Proposition \ref{coro2}, respectively,
Theorem \ref{coro_thm3_1_exp} is proven.
\end{proof}

%% file: appendices_lem2.tex
\section{Proof of Proposition~2}\label{app_lem2}
Let $j$ be any integer chosen from $\{1,\ldots,n\}$, and $y_{j}$ be the number of jobs that have arrived by the time  $c_{(j)}(\pi)$, where $y_{j}\geq j$. Because $j$ jobs are completed by the time $c_{(j)}(\pi)$ in policy $\pi$, 
there are exactly $(y_{j}-j)$ incomplete jobs in the system at time $c_{(j)}(\pi)$. By the definition of the system state, we have $\xi_{[i],\pi}(c_{(j)}(\pi))=0$ for $i=y_{j}-j+1,\ldots,n$. Hence,
\begin{align}
\sum_{i=y_{j}-j+1}^n\xi_{[i],\pi}(c_{(j)}(\pi))=0.\nonumber
\end{align}
Combining this with \eqref{eq_ordering_2_1}, yields that policy $P$ satisfies 
\begin{align}\label{eq_proof_1}
\sum_{i=y_{j}-j+1}^n \gamma_{[i],P}\left(c_{(j)}(\pi)\right)\leq 0.
\end{align}
Next, the definition of the system state tells us that $ \gamma_{i,P}\left(t\right)\geq0$ holds for all $i=1,\ldots,n$ and $t\geq0$. Hence, we have
\begin{align}\label{eq_proof_2}
\gamma_{[i],P}\left(c_{(j)}(\pi)\right)=0, ~\forall~i=y_j-j+1,\ldots,n.
\end{align}
Therefore, there are at most $y_j-j$ jobs which have unassigned tasks at time  $c_{(j)}(\pi)$ in policy $P$.
Because  the sequence of job arrival times $a_1,a_2,\ldots,a_n$ are invariant under any policy, $y_j$ jobs have arrived by the time  $c_{(j)}(\pi)$ in policy $P$. Thus, there are at least $j$ jobs which have no unassigned tasks  at the time  $c_{(j)}(\pi)$  in policy $P$, which can be equivalently expressed as
\begin{align}\label{eq_proof_222}
v_{(j)}(P)\!  \leq c_{(j)}(\pi).
\end{align}
Because $j$ is arbitrarily chosen, \eqref{eq_proof_222} holds for all $j=1,\ldots,n$, which is exactly \eqref{eq_ordering_2_2}. In addition, \eqref{eq_ordering_2_3} follows from \eqref{eq_ordering_2_2}, which completes the proof.

%% file: appendices_lem3.tex
\section{Proofs of Propositions~3-4}\label{app_lem3}
\begin{proof}[Proof of Proposition~\ref{ordering_3}]
Let $w_i$ be the index of the job associated with the job completion time $c_{(i)}(P)$. In order to prove \eqref{eq_ordering_3_2}, it is sufficient to show that for each $j=1,2,\ldots,n$,
\begin{align}\label{eq_Condition_3_12_thm3}
c_{w_j}(P)-d_{w_j} \leq \max_{i=1,2,\ldots,n }[c_{i}(\pi)-d_i].
\end{align}
We prove \eqref{eq_Condition_3_12_thm3} by contradiction. \emph{For this, let us assume that 
\begin{align}\label{eq_proof_9_thm3}
c_{i}(\pi)< c_{w_j}(P) 
\end{align}
holds for all job $i$ satisfying $a_i\leq c_{w_j}(P)$ and $d_i\leq d_{w_j}$.} That is, if job $i$ arrives before time $c_{w_j}(P)$ and its due time is no later than $d_{w_j}$, then  job $i$ is completed before time $c_{w_j}(P)$ in policy $\pi$.
Define  
\begin{align}\label{eq_def}
\tau_j=\max_{i:a_i\leq c_{w_j}(P),d_i\leq d_{w_j}} c_{i}(\pi).
\end{align} 
According to \eqref{eq_proof_9_thm3} and \eqref{eq_def}, we can obtain
\begin{align}\label{eq_proof_19_thm3}
\tau_j<c_{w_j}(P).
\end{align}

On the other hand, \eqref{eq_def} tells us that all job  $i$ satisfying $d_i\leq d_{w_j}$ and $a_i\leq c_{w_j}(P)$ are completed by time $\tau_j$ in policy $\pi$. 
By this,
the system state of policy $\pi$ satisfies  
\begin{align}
\sum_{i:d_i\leq d_{w_j}}\xi_{i,\pi}(\tau_j) =0.\nonumber
\end{align}
Combining this with \eqref{eq_ordering_3_1}, yields
\begin{align}\label{eq_proof_1_thm3}
\sum_{i:d_i\leq d_{w_j}} \xi_{i,P}(\tau_j)\leq 0.
\end{align}
Further, the definition of the system state tells us that $\xi_{i,P}(t)\geq0$ for all $i=1,\ldots,n$ and $t\geq0$. Using this and \eqref{eq_proof_1_thm3}, we get that job $w_j$ satisfies
\begin{align}
\xi_{w_j,P}(\tau_j)= 0.\nonumber
\end{align}
That is, all tasks of job $w_j$ are completed by time $\tau_j$ in policy $P$. Hence, $c_{w_j}(P)\leq \tau_j$, where contradicts with \eqref{eq_proof_19_thm3}. Therefore, there exists at least one job $i$ satisfying the conditions 
$a_i\leq c_{w_j}(P)$, $d_i\leq d_{w_j}$, and $c_{w_j}(P) \leq c_{i}(\pi)$. This can be equivalently expressed as
\begin{align}\label{eq_proof_3_thm3}
c_{w_j}(P) \leq \max_{i:a_i\leq c_{w_j}(P),d_i\leq d_{w_j} }c_{i}(\pi).
\end{align}
Hence, for each $j=1,2,\ldots,n$,
\begin{align}
c_{w_j}(P)-d_{w_j} &\leq \max_{i:a_i\leq c_{w_j}(P),d_i\leq d_{w_j} }c_{i}(\pi)-d_{w_j} \nonumber\\
& \leq \max_{i:a_i\leq c_{w_j}(P),d_i\leq d_{w_j} }[c_{i}(\pi)-d_{i}]\nonumber\\
&\leq \max_{i=1,2,\ldots,n }[c_{i}(\pi)-d_i].\nonumber
\end{align}
This implies \eqref{eq_ordering_3_2}. Hence, Proposition~\ref{ordering_3} is proven.
\end{proof}
The proof of Proposition~\ref{ordering_3_1} is almost identical to that of 
Proposition~\ref{ordering_3}, and hence is not repeated here. The only difference is that $c_{w_j}(P)$ and $\bm{\xi}_{P}(\tau_j)$ in the proof of Proposition~\ref{ordering_3} should be replaced by $v_{w_j}(P)$ and $\bm{\gamma}_{P}(\tau_j)$, respectively.

%% file: appendices_lem1_0.tex

\section{Proof of Proposition 5} \label{app_lem1_0}

The following two lemmas are needed to prove Proposition \ref{lem1_0}:  

\begin{lemma}\cite[Lemmas 1-2]{Smith78}\label{lem_non_prmp1_0}
Suppose that under policy $P$, $\{\bm{\xi}_{P}',\bm{\gamma}_{P}'\}$ is obtained by completing $b_P$ tasks in the system whose state is $\{\bm{\xi}_{P},\bm{\gamma}_{P}\}$. Further, suppose that under policy $\pi$, $\{\bm{\xi}_{\pi}',\bm{\gamma}_{\pi}'\}$ is obtained by completing $b_\pi$ tasks in the system whose state is $\{\bm{\xi}_{\pi},\bm{\gamma}_{\pi}\}$.
If $b_P\geq b_\pi$, policy $P$ satisfies Condition 3 of Proposition \ref{lem1_0}, and 
\begin{eqnarray}\label{eq_non_prmp_41_0}
\sum_{i=j}^n {\xi}_{[i],P}\leq \sum_{i=j}^n {\xi}_{[i],\pi}, ~j = 1,2,\ldots,n,\nonumber
\end{eqnarray}
then
\begin{eqnarray}\label{eq_non_prmp_40_0}
\sum_{i=j}^n {\xi}_{[i],P}'\leq \sum_{i=j}^n {\xi}_{[i],\pi}', ~j = 1,2,\ldots,n.
\end{eqnarray}
\end{lemma}

\begin{lemma}\cite[Lemma 3]{Smith78}\label{lem_non_prmp2_0}
Suppose that, under policy $P$, $\{\bm{\xi}_{P}',\bm{\gamma}_{P}'\}$ is obtained by adding a job with $b$ tasks to the system whose state is $\{\bm{\xi}_{P},\bm{\gamma}_{P}\}$. Further, suppose that, under policy $\pi$, $\{\bm{\xi}_{\pi}',\bm{\gamma}_{\pi}'\}$ is obtained by adding a job with $b$ tasks to the system whose state is $\{\bm{\xi}_{\pi},\bm{\gamma}_{\pi}\}$.
If
\begin{eqnarray}
\sum_{i=j}^n {\xi}_{[i],P}\leq \sum_{i=j}^n {\xi}_{[i],\pi}, ~j = 1,2,\ldots,n,\nonumber
\end{eqnarray}
then
\begin{eqnarray}
\sum_{i=j}^n {\xi}_{[i],P}'\leq \sum_{i=j}^n {\xi}_{[i],\pi}', ~j = 1,2,\ldots,n.\nonumber
\end{eqnarray}
\end{lemma}

We now use Lemma \ref{lem_non_prmp1_0} and Lemma \ref{lem_non_prmp2_0} to prove Proposition \ref{lem1_0}.
\ifreport
\begin{proof}[Proof of Proposition \ref{lem1_0}]
\else
\begin{proof}[of Proposition \ref{lem1_0}]
\fi

Because policy $P$ is more work-efficient than policy $\pi$, the sequence of task completion times in policy $P$ are  smaller  than those in policy $\pi$, i.e., 
\begin{align}\label{eq_something}
(t_{1,P},\ldots, t_{k_{\text{sum}},P}) \leq(t_{1,\pi},\ldots, t_{k_{\text{sum}},\pi}). 
\end{align} 

We modify the task completion times on the sample-path of policy $P$ as follows:  For each $i=1,\ldots, k_{\text{sum}}$, if a task of job $j_i$ is completed at time $t_{i,P}$ on the original sample-path of policy $P$, then on the modified sample-path of policy $P$, the same task of job $j_i$ is completed at time $t_{i,\pi}$. This modification satisfies the following three claims: 
\begin{itemize}
\item[1.] According to \eqref{eq_something}, the task completion times of policy $P$ are postponed after the modification;
\item[2.] The order of completed tasks in policy $P$ remains the same before and after the modification;
\item[3.] The task completion times on the sample-path of policy $\pi$ and on the modified sample-path of policy $P$ are identical.
\end{itemize}

Let $\hat{\bm{\xi}}_{P}(t) =(\hat{\xi}_{1,P}(t),\ldots,\hat{\xi}_{n,P}(t))$ and $\hat{\bm{\gamma}}_{P}(t) =(\hat{\gamma}_{1,P}(t),$ $\ldots,\hat{\gamma}_{n,P}(t))$ denote the system state on the modified sample-path of policy $P$. 
From Claims 1 and 2, we can get ${\xi}_{i,P}(t)  \leq \hat{\xi}_{i,P}(t) $ for all $t\geq 0$ and $i=1,\ldots,n$. Therefore, for all $t\in[0,\infty)$
\begin{align}\label{eq_lem1_0_proof_1}
\sum_{i=j}^n {\xi}_{[i],P}(t)\leq \sum_{i=j}^n \hat{\xi}_{[i],P}(t), ~i=1,2,\ldots,n.
\end{align}

Next, we compare policy $\pi$ with the modified sample-path of policy $P$. 
According to Claim 1, Claim 2, and $k_1\leq k_2\leq \ldots\leq k_n$, each task completing service on the modified sample-path of policy $P$ is still from the job with the fewest remaining tasks among all jobs with remaining tasks. That is, Condition 3 of Proposition \ref{lem1_0} is satisfied on the modified sample-path of policy $P$, which is required by Lemma \ref{lem_non_prmp1_0}.

Because $\hat{\bm{\xi}}_{P}(0) = {\bm{\xi}}_{\pi}(0) =\bm{0}$, by using Claim 3, Lemma \ref{lem_non_prmp1_0}, and Lemma \ref{lem_non_prmp2_0}, and taking an induction on the job arrival events and task completion events over time, we can obtain for all $t\in[0,\infty)$
\begin{align}\label{eq_lem1_0_proof_2}
\sum_{i=j}^n \hat{\xi}_{[i],P}(t)\leq \sum_{i=j}^n {\xi}_{[i],\pi}(t), ~i=1,2,\ldots,n.
\end{align}
Combining \eqref{eq_lem1_0_proof_1} and \eqref{eq_lem1_0_proof_2}, yields \eqref{eq_ordering_1_1}. Then, \eqref{eq_ordering_1_1_2} and \eqref{eq_ordering_1_2} follow from Proposition \ref{ordering_1}, which completes the proof. \end{proof}

%% file: appendices_version2.tex
\section{Proof of Proposition 6} \label{app0}

The following two lemmas are needed to prove Proposition \ref{lem1}:  

\begin{lemma}\label{lem_non_prmp1}
Suppose that under policy $P$, $\{\bm{\xi}_{P}',\bm{\gamma}_{P}'\}$ is obtained by allocating $b_P$ unassigned tasks to the servers in the system whose state is $\{\bm{\xi}_{P},\bm{\gamma}_{P}\}$. Further, suppose that under policy $\pi$, $\{\bm{\xi}_{\pi}',\bm{\gamma}_{\pi}'\}$ is obtained by completing $b_\pi$ tasks in the system whose state is $\{\bm{\xi}_{\pi},\bm{\gamma}_{\pi}\}$.
If $b_P\geq b_\pi$, condition 2 of Proposition \ref{lem1} is satisfied in policy $P$, and
\begin{eqnarray}\label{eq_non_prmp_41}
\sum_{i=j}^n {\gamma}_{[i],P}\leq \sum_{i=j}^n {\xi}_{[i],\pi}, ~\forall~j = 1,2,\ldots,n,\nonumber
\end{eqnarray}
then
\begin{eqnarray}\label{eq_non_prmp_40}
\sum_{i=j}^n {\gamma}_{[i],P}'\leq \sum_{i=j}^n {\xi}_{[i],\pi}', ~\forall~j = 1,2,\ldots,n.
\end{eqnarray}
\end{lemma}

\begin{proof}
If $\sum_{i=j}^n {\gamma}_{[i],P}'=0$, then the inequality \eqref{eq_non_prmp_40} follows naturally. 
If $\sum_{i=j}^n {\gamma}_{[i],P}'>0$, then there exist  unassigned tasks which have not been assigned to any server. 
In policy $P$, each task allocated to the servers is from the job with the minimum positive ${\gamma}_{i,P}$. Hence,
$\sum_{i=j}^n {\gamma}_{[i],P}'=\sum_{i=j}^n {\gamma}_{[i],P} - b_P \leq \sum_{i=j}^n {\xi}_{[i],\pi} - b_\pi \leq \sum_{i=j}^n {\xi}_{[i],\pi}'$.
\end{proof}

\begin{lemma}\label{lem_non_prmp2}
Suppose that, under policy $P$, $\{\bm{\xi}_{P}',\bm{\gamma}_{P}'\}$ is obtained by adding a job with $b$ tasks to the system whose state is $\{\bm{\xi}_{P},\bm{\gamma}_{P}\}$. Further, suppose that, under policy $\pi$, $\{\bm{\xi}_{\pi}',\bm{\gamma}_{\pi}'\}$ is obtained by adding a job with $b$ tasks to the system whose state is $\{\bm{\xi}_{\pi},\bm{\gamma}_{\pi}\}$.
If
\begin{eqnarray}
\sum_{i=j}^n {\gamma}_{[i],P}\leq \sum_{i=j}^n {\xi}_{[i],\pi}, ~\forall~j = 1,2,\ldots,n,\nonumber
\end{eqnarray}
then
\begin{eqnarray}
\sum_{i=j}^n {\gamma}_{[i],P}'\leq \sum_{i=j}^n {\xi}_{[i],\pi}', ~\forall~j = 1,2,\ldots,n.\nonumber
\end{eqnarray}
\end{lemma}

\begin{proof}
Without loss of generalization, we suppose that after the job arrival, $b$ is the $l$-th largest component of $\bm{\gamma}_{P}'$ and the $m$-th largest component of $\bm{\xi}_{\pi}'$, i.e., $\gamma'_{[l],P} = \xi'_{[m],\pi}= b$. We consider the following four cases:

{Case 1}: $l<j, m<j$. We have $\sum_{i=j}^n {\gamma}_{[i],P}' =\sum_{i=j-1}^n {\gamma}_{[i],P} \leq \sum_{i=j-1}^n {\xi}_{[i],\pi}= \sum_{i=j}^n {\xi}_{[i],\pi}'$.

{Case 2}: $l<j, m\geq j$. We have $\sum_{i=j}^n {\gamma}_{[i],P}' =\sum_{i=j-1}^n {\gamma}_{[i],P} \leq b + \sum_{i=j}^n {\gamma}_{[i],P} \leq b + \sum_{i=j}^n {\xi}_{[i],\pi} = \sum_{i=j}^n {\xi}_{[i],\pi}'$.

{Case 3}: $l\geq j, m<j$. We have $\sum_{i=j}^n {\gamma}_{[i],P}' = b + \sum_{i=j}^n {\gamma}_{[i],P} \leq \sum_{i=j-1}^n {\gamma}_{[i],P} \leq \sum_{i=j-1}^n {\xi}_{[i],\pi} = \sum_{i=j}^n {\xi}_{[i],\pi}'$.

{Case 4}: $l\geq j, m\geq j$. We have $\sum_{i=j}^n {\gamma}_{[i],P}' = b + \sum_{i=j}^n {\gamma}_{[i],P} \leq b + \sum_{i=j}^n {\xi}_{[i],\pi} = \sum_{i=j}^n {\xi}_{[i],\pi}'$.
\end{proof}

We now use Lemma \ref{lem_non_prmp1} and Lemma \ref{lem_non_prmp2} to prove Proposition \ref{lem1}.
\ifreport
\begin{proof}[Proof of Proposition \ref{lem1}]
\else
\begin{proof}[of Proposition \ref{lem1}]
\fi

\begin{figure*}
\centering 
\includegraphics[width=0.8\textwidth]{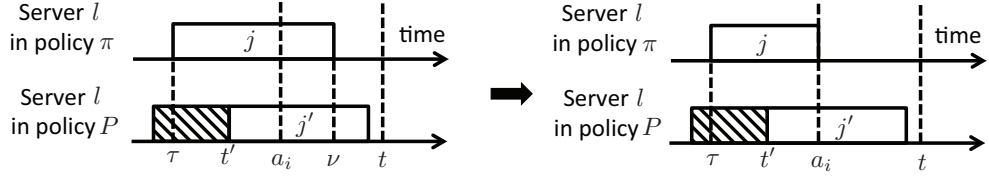} \caption{Illustration of the modification of task completion times in policy $\pi$: If in policy $\pi$, task $j$ starts execution at time $\tau\in[0,a_i]$ and completes execution at time $\nu\in(a_i,t]$, and in policy $P$, task $j'$ starts execution at time $t'\in[0,a_i]$, then the completion time of task $j$ is changed from $\nu$ to $a_i^-$ in policy $\pi$.}
\label{fig_modification} 
\end{figure*} 

Assume that no task is completed at the job arrival times $a_i$ for $i=1,\ldots,n$. This does not lose any generality, because if a task is completed at time $t_j=a_i$, Proposition \ref{lem1} can be proven by first proving for the case $t_j=a_i+\epsilon$ and then taking the limit $\epsilon\rightarrow 0$. 
We prove \eqref{eq_ordering_2_1} by induction. 

\emph{Step 1: We will show that \eqref{eq_ordering_2_1} holds during $[0,a_2)$.\footnote{Note that $a_1=0$.}} 

Because $\bm{\xi}_{P}(0^-) = \bm{\gamma}_{P}(0^-) =\bm{\xi}_{\pi}(0^-) = \bm{\gamma}_{\pi}(0^-)=\bm{0}$, \eqref{eq_ordering_2_1} holds at time $0^-$. Job 1 arrives at time $a_1=0$. By Lemma \ref{lem_non_prmp2}, \eqref{eq_ordering_2_1} holds at time $0$. Let $t$ be an arbitrarily chosen time during $(0,a_{2})$. Suppose that  $b_\pi$ tasks start execution and also complete execution during $[0,t]$ in policy $\pi$. We need to consider two cases: 

Case 1: The queue is not empty (there exist unassigned tasks in the queue) during $[0,t]$ in policy $P$. By the weak work-efficiency ordering condition, no fewer than $b_\pi$ tasks start execution during $[0,t]$ in policy $P$. Because \eqref{eq_ordering_2_1} holds at time $0$, by Lemma \ref{lem_non_prmp1}, \eqref{eq_ordering_2_1} also holds at time $t$. 

Case 2: The queue is empty (all tasks in the system are in service) by time $t'\in[0,t]$ in policy $P$. Because $t\in(0,a_{2})$ and there is no task arrival during $(0,a_2)$, there is no task arrival during $(t',t]$. Hence, it must hold that all tasks in the system are in service at time $t$. Then, the system state of policy $P$ satisfies $\sum_{i=j}^n {\gamma}_{[i],P}(t)=0$ for all $j=1,2,\ldots,n$ at time $t$. Hence, \eqref{eq_ordering_2_1} holds at time $t$.

In summary of these two cases, \eqref{eq_ordering_2_1} holds for all $t\in[0,a_2)$.

\emph{Step 2: Assume that for some integer $i\in\{2,\ldots, n\}$, the conditions of Proposition \ref{lem1} imply that \eqref{eq_ordering_2_1} holds for all $t\in[0,a_i)$. We will prove that the conditions of Proposition \ref{lem1} imply that  \eqref{eq_ordering_2_1} holds for all $t\in[0,a_{i+1})$.}

Let $t$ be an arbitrarily chosen time during $(a_i,a_{i+1})$. We modify the task completion times in policy $\pi$ as follows: For each pair of corresponding task $j$ and task $j'$ mentioned in the definition of the weak work-efficiency ordering, if 
\begin{itemize}
\item In policy $\pi$, task $j$ starts execution at time $\tau\in[0,a_i]$ and completes execution at time $\nu\in(a_i,t]$,
\item In policy $P$, the queue is not empty (there exist unassigned tasks in the queue) during $[\tau,\nu]$,
\item In policy $P$, the corresponding task $j'$ starts execution at time $t'\in[0,a_i]$,
\end{itemize}
then the completion time of task $j$ is modified from $\nu$ to $a_i^-$ in policy $\pi$, as illustrated in Fig. \ref{fig_modification}.

This modification satisfies the following three claims:
\begin{itemize}
\item[1.] The system state of policy $\pi$ at time $t$ remains the same before and after this modification;
\item[2.] Policy $P$ is still weakly more work-efficient than policy $\pi$ after this modification;
\item[3.] If $b_\pi$ tasks complete execution during $[a_i,t]$ on the modified sample path of policy $\pi$, and  the queue is not empty (there exist unassigned tasks in the queue) during $[a_i,t]$ in policy $P$, then no fewer than $b_\pi$ tasks start execution during $[a_i,t]$ in policy $P$.
\end{itemize}
We now prove these three claims. Claim 1 follows from the fact that the tasks completed during $[0,t]$ remain the same before and after this modification. It is easy to prove Claim 2 by checking the definition of work-efficiency ordering. For Claim 3, notice that if a task $j$ starts execution and completes execution during $[a_i,t]$ on the modified sample path of policy $\pi$, then by Claim 2, its corresponding task $j'$ must start execution during $[a_i,t]$ in policy $P$. On the other hand, if a task $j$ starts execution during $[0,a_i]$ and completes execution during $[a_i,t]$ on the modified sample path of policy $\pi$, then by the modification, its corresponding task $j'$ must start execution during $[a_i,t]$ in policy $P$. By combining these two cases, Claim 3 follows.

We use these three claims to prove the  statement of \emph{Step 2}. According to Claim 2, policy $P$ is weakly more work-efficient than policy $\pi$ after the modification. By the assumption of \emph{Step 2}, \eqref{eq_ordering_2_1} holds during $[0,a_i)$ for the modified sample path of policy $\pi$. Job $j$ arrives at time $a_i$. By Lemma \ref{lem_non_prmp2}, \eqref{eq_ordering_2_1} holds at time $a_i$ for the modified sample path of policy $\pi$. Suppose that $b_\pi$ tasks complete execution during $[a_i,t]$ on the modified sample path of policy $\pi$.
We need to consider two cases: 

Case 1: The queue is not empty (there exist unassigned tasks in the queue) during $[a_i,t]$ in policy $P$. By Claim 3, no fewer than $b_\pi$ tasks start execution during $[a_i,t]$ in policy $P$. Because \eqref{eq_ordering_2_1} holds at time $a_i$, by Lemma \ref{lem_non_prmp1}, \eqref{eq_ordering_2_1} also holds at time $t$ for the modified sample path of policy $\pi$. 

Case 2: The queue is empty (all tasks in the system are in service) at  time $t'\in[a_i,t]$ in policy $P$. Because $t\in(a_i,a_{i+1})$ and $t'\in[a_i,t]$, there is no task arrival during $(t',t]$. Hence, it must hold that all tasks in the system are in service at time $t$. Then, the system state of policy $P$ satisfies $\sum_{i=j}^n {\gamma}_{[i],P}(t)=0$ for all $j=1,2,\ldots,n$  at time $t$. Hence, \eqref{eq_ordering_2_1} holds at time $t$ for the modified sample path of policy $\pi$.

In summary of these two cases, \eqref{eq_ordering_2_1} holds at time $t$ for the modified sample path of policy $\pi$. By Claim 1, the system state of policy $\pi$ at time $t$ remains the same before and after this modification. Hence, \eqref{eq_ordering_2_1} holds at time $t$ for the original sample path of policy $\pi$. Therefore, if  the assumption of \emph{Step 2} is true, then \eqref{eq_ordering_2_1} holds for all $t\in[0,a_{i+1})$. 

By induction, \eqref{eq_ordering_2_1} holds at time $t\in[0,\infty)$. Then, \eqref{eq_ordering_2_2} and \eqref{eq_ordering_2_3} follow from Proposition \ref{ordering_2}.
This completes the proof.\end{proof}

%% file: appendices_lem2_0.tex
\section{Proof  of Proposition 7} \label{app_lem2_0}

The proof of Proposition \ref{lem2_0} requires the following two lemmas:

\begin{lemma}\label{lem_non_prmp1_thm3_0}
Suppose that, in policy $P$, $\{\bm{\xi}_P',\bm{\gamma}_P'\}$ is obtained by completing $b_P$ tasks in the system whose state is $\{\bm{\xi}_P,\bm{\gamma}_P\}$. Further, suppose that, in policy $\pi$, $\{\bm{\xi}_\pi',\bm{\gamma}_\pi'\}$ is obtained by completing $b_\pi$ tasks in the system whose state is $\{\bm{\xi}_\pi,\bm{\gamma}_\pi\}$.
If $b_P\geq b_\pi$, condition 3 of Proposition \ref{lem2_0} is satisfied in policy $P$, and
\begin{eqnarray}
\sum_{i:d_i\leq\tau} \xi_{i,P}\leq \sum_{i:d_i\leq\tau} \xi_{i,\pi}, ~\tau\in[0,\infty),\nonumber
\end{eqnarray}
then
\begin{eqnarray}\label{eq_non_prmp_40_thm3_0}
\sum_{i:d_i\leq\tau} \xi_{i,P}'\leq \sum_{i:d_i\leq\tau} \xi_{i,\pi}', ~\tau\in[0,\infty).\end{eqnarray}
\end{lemma}

\begin{proof}
If $\sum_{i:d_i\leq\tau} \xi_{i,P}'=0$, then the inequality \eqref{eq_non_prmp_40_thm3} follows naturally. 
If $\sum_{i:d_i\leq\tau} \xi_{i,P}'>0$, then there exist some remaining tasks. 
In policy $P$, each task completing service is from the job with the earliest due time. Hence,
$\sum_{i:d_i\leq\tau} \xi_{i,P}'=\sum_{i:d_i\leq\tau} \xi_{i,P} - b_P \leq \sum_{i:d_i\leq\tau} \xi_{i,\pi} -b_\pi \leq \sum_{i:d_i\leq\tau} \xi_{i,\pi}'$.
\end{proof}

\begin{lemma}\label{lem_non_prmp2_thm3_0}
Suppose that under policy $P$, $\{\bm{\xi}_P',\bm{\gamma}_P'\}$ is obtained by adding a job with $b$ tasks and due time $d$ to the system whose state is $\{\bm{\xi}_P,\bm{\gamma}_P\}$. Further, suppose that under policy $\pi$, $\{\bm{\xi}_\pi',\bm{\gamma}_\pi'\}$ is obtained by adding a job with $b$ tasks and due time $d$ to the system whose state is $\{\bm{\xi}_\pi,\bm{\gamma}_\pi\}$.
If
\begin{eqnarray}
\sum_{i:d_i\leq\tau} \xi_{i,P}\leq \sum_{i:d_i\leq\tau} \xi_{i,\pi}, ~\tau\in[0,\infty),\nonumber
\end{eqnarray}
then
\begin{eqnarray}
\sum_{i:d_i\leq\tau} \xi_{i,P}'\leq \sum_{i:d_i\leq\tau} \xi_{i,\pi}', ~\tau\in[0,\infty).\nonumber\end{eqnarray}
\end{lemma}

\begin{proof}
If $d\leq\tau$, then
$\sum_{i:d_i\leq\tau} \xi_{i,P}'\leq \sum_{i:d_i\leq\tau} \xi_{i,P} + b\leq \sum_{i:d_i\leq\tau} \xi_{i,\pi}+b \leq \sum_{i:d_i\leq\tau} \xi_{i,\pi}'$.

If $d>\tau$, then
$\sum_{i:d_i\leq\tau} \xi_{i,P}'\leq \sum_{i:d_i\leq\tau} \xi_{i,P} \leq \sum_{i:d_i\leq\tau} \xi_{i,\pi} \leq \sum_{i:d_i\leq\tau} \xi_{i,\pi}'$.
\end{proof}

The proof of Proposition \ref{lem2_0} is almost identical with that of Proposition \ref{lem1_0}, and hence is not repeated here. The only difference is that Lemma \ref{lem_non_prmp1_0} and Lemma \ref{lem_non_prmp2_0} in the proof of Proposition \ref{lem1_0} should be replaced by Lemma \ref{lem_non_prmp1_thm3_0} and Lemma \ref{lem_non_prmp2_thm3_0}, respectively.

%% file: appendices_2version2.tex

\section{Proof  of Proposition 8} \label{app0_1}



The proof of Proposition \ref{lem2} requires the following two lemmas:

\begin{lemma}\label{lem_non_prmp1_thm3}
Suppose that, in policy $P$, $\{\bm{\xi}_P',\bm{\gamma}_P'\}$ is obtained by allocating $b_P$ unassigned tasks to the servers in the system whose state is $\{\bm{\xi}_P,\bm{\gamma}_P\}$. Further, suppose that, in policy $\pi$, $\{\bm{\xi}_\pi',\bm{\gamma}_\pi'\}$ is obtained by completing $b_\pi$ tasks in the system whose state is $\{\bm{\xi}_\pi,\bm{\gamma}_\pi\}$.
If $b_P\geq b_\pi$, condition 2 of Proposition \ref{lem2} is satisfied in policy $P$, and
\begin{eqnarray}
\sum_{i:d_i\leq\tau} \gamma_{i,P}\leq \sum_{i:d_i\leq\tau} \xi_{i,\pi}, ~\tau\in[0,\infty),\nonumber
\end{eqnarray}
then
\begin{eqnarray}\label{eq_non_prmp_40_thm3}
\sum_{i:d_i\leq\tau} \gamma_{i,P}'\leq \sum_{i:d_i\leq\tau} \xi_{i,\pi}', ~\tau\in[0,\infty).\end{eqnarray}
\end{lemma}

\begin{proof}
If $\sum_{i:d_i\leq\tau} \gamma_{i,P}'=0$, then the inequality \eqref{eq_non_prmp_40_thm3} follows naturally. 
If $\sum_{i:d_i\leq\tau} \gamma_{i,P}'>0$, then there exist some unassigned tasks in the queue. 
In policy $P$, each task allocated to the servers is from the job with the earliest due time. Hence,
$\sum_{i:d_i\leq\tau} \gamma_{i,P}'=\sum_{i:d_i\leq\tau} \gamma_{i,P} - b_P \leq \sum_{i:d_i\leq\tau} \xi_{i,\pi} -b_\pi \leq \sum_{i:d_i\leq\tau} \xi_{i,\pi}'$.
\end{proof}

\begin{lemma}\label{lem_non_prmp2_thm3}
Suppose that under policy $P$, $\{\bm{\xi}_P',\bm{\gamma}_P'\}$ is obtained by adding a job with $b$ tasks and due time $d$ to the system whose state is $\{\bm{\xi}_P,\bm{\gamma}_P\}$. Further, suppose that under policy $\pi$, $\{\bm{\xi}_\pi',\bm{\gamma}_\pi'\}$ is obtained by adding a job with $b$ tasks and due time $d$ to the system whose state is $\{\bm{\xi}_\pi,\bm{\gamma}_\pi\}$.
If
\begin{eqnarray}
\sum_{i:d_i\leq\tau} \gamma_{i,P}\leq \sum_{i:d_i\leq\tau} \xi_{i,\pi}, ~\tau\in[0,\infty),\nonumber
\end{eqnarray}
then
\begin{eqnarray}
\sum_{i:d_i\leq\tau} \gamma_{i,P}'\leq \sum_{i:d_i\leq\tau} \xi_{i,\pi}', ~\tau\in[0,\infty).\nonumber\end{eqnarray}
\end{lemma}

\begin{proof}
If $d\leq\tau$, then
$\sum_{i:d_i\leq\tau} \gamma_{i,P}'\leq \sum_{i:d_i\leq\tau} \gamma_{i,P} + b\leq \sum_{i:d_i\leq\tau} \xi_{i,\pi}+b \leq \sum_{i:d_i\leq\tau} \xi_{i,\pi}'$.

If $d>\tau$, then
$\sum_{i:d_i\leq\tau} \gamma_{i,P}'\leq \sum_{i:d_i\leq\tau} \gamma_{i,P} \leq \sum_{i:d_i\leq\tau} \xi_{i,\pi} \leq \sum_{i:d_i\leq\tau} \xi_{i,\pi}'$.
\end{proof}

The proof of Proposition \ref{lem2} is almost identical with that of Proposition \ref{lem1}, and hence is not repeated here. The only difference is that Lemma \ref{lem_non_prmp1} and Lemma \ref{lem_non_prmp2} in the proof of Proposition \ref{lem1} should be replaced by Lemma \ref{lem_non_prmp1_thm3} and Lemma \ref{lem_non_prmp2_thm3}, respectively.


%% file: appendices_coro2_0.tex
\section{Proofs of Propositions 9-10} \label{app_lem1_1}

\ifreport
\begin{proof}[Proof of Proposition \ref{lem1_1_0}]
\else
\begin{proof}[of Proposition \ref{lem1_1_0}]
\fi
We have proven that \eqref{eq_ordering_1_1_2} holds under the conditions of Proposition \ref{lem1_1_0}.
Note that \eqref{eq_ordering_1_1_2} can be equivalently expressed in the following vector form:
\begin{align}
\bm{c}_{\uparrow} (P)\leq \bm{c}_{\uparrow} (\pi).\nonumber
\end{align}
Because any $f\in\mathcal{D}_{\text{sym}}$ is a symmetric and increasing function, we can obtain
\begin{align}
&f(\bm{c} (P))=f(\bm{c}_{\uparrow} (P)) \nonumber\\
\leq_{}&f(\bm{c}_{\uparrow} (\pi))= f(\bm{c} (\pi)).\nonumber
\end{align} 
This completes the proof.
\end{proof}

The proof of Proposition \ref{lem1_1} is almost identical with that of 
Proposition \ref{lem1_1_0}, and hence is not repeated here. The only difference is that $\bm{c} (P)$ in the proof of Proposition \ref{lem1_1_0} should be replaced by $\bm{v} (P)$, respectively.

\section{Proof of Proposition 11}\label{app_coro2_0}
In the proof of Proposition \ref{coro2_0}, we need to use the following rearrangement inequality: 
\begin{lemma} \cite[Theorem 6.F.14]{Marshall2011}\label{lem_rearrangement}
Consider two $n$-dimensional vectors $(x_1,\dots,x_n)$ and $(y_1,\ldots,y_n)$. If $(x_i- x_j)(y_i-y_j)\leq0$ for two indices $i$ and $j$ where $1\leq i<j\leq n$, then 
\begin{align}
&(x_1\!-\!y_1,\ldots,x_j\!-\!y_i,\ldots,x_i\!-\!y_j,\ldots,x_n\!-\!y_n)\nonumber\\
\prec& (x_1\!-\!y_1,\ldots,x_i\!-\!y_i,\ldots,x_j\!-\!y_j,\ldots,x_n\!-\!y_n).\nonumber
\end{align}
\end{lemma}

\ifreport
\begin{proof}[Proof of Proposition \ref{coro2_0}]
\else
\begin{proof}[of Proposition \ref{coro2_0}]
\fi

For  $f\in\mathcal{D}_{\text{sym}}$, \eqref{eq_ordering_1_1_2} and  \eqref{eq_lem_general} follow from Proposition \ref{lem1_0} and Proposition \ref{lem1_1_0}.

For  $f\in\mathcal{D}_{\text{Sch-1}}$, \eqref{eq_lem_general}  is proven in 3 steps, which are described as follows:

\emph{Step 1: We will show that}
\begin{align}\label{eq_coro2_0_3}
\bm{c}(P)-\bm{d}\prec_{\text{w}} \bm{c} (\pi)-\bm{d}.
\end{align} 
According to Eq. (1.A.17) and Theorem 5.A.9 of \cite{Marshall2011},
it is sufficient to show that there exists an $n$-dimensional vector $\bm{c}'$ such that 
\begin{align}\label{eq_coro2_0_6}
\bm{c}(P)-\bm{d}\prec \bm{c}'- \bm{d} \leq \bm{c} (\pi)-\bm{d}.
\end{align}
Vector $\bm{c}'$ is constructed as follows: First, $\bm{c}'$ is a rearrangement (or permutation) of the vector $\bm{c}(P)$, which can be equivalently expressed as
\begin{align}\label{eq_coro2_0_4}
c_{(i)}'=c_{(i)}(P),~i=1,\ldots,n.
\end{align}
Second, for each $j=1,\ldots,n$, if the completion time $c_j(\pi)$ of job $j$  is the $i$-th smallest component of $\bm{c}(\pi)$, i.e., 
\begin{align}
c_j(\pi)= c_{(i)}(\pi),
\end{align}
then $c_j'$ associated with job $j$ is the $i$-th smallest component of $\bm{c}'$, i.e., 
\begin{align}\label{eq_coro2_0_5}
c_j' = c_{(i)}'.
\end{align}
Combining \eqref{eq_ordering_1_1_2} and \eqref{eq_coro2_0_4}-\eqref{eq_coro2_0_5}, yields
\begin{align}
c_j' = c_{(i)}'=c_{(i)}(P) \leq c_{(i)}(\pi) = c_j(\pi)\nonumber
\end{align}
for $j=1,\ldots,n$. This implies $\bm{c}' \leq \bm{c} (\pi)$, and hence the second inequality in \eqref{eq_coro2_0_6} is proven.

The remaining task is to prove the first inequality in \eqref{eq_coro2_0_6}. 
First, consider the case that the due times $d_1,\ldots, d_n$ of the $n$ jobs are  different from each other. 
The vector  $\bm{c}(P)$ can be obtained from  $\bm{c}'$ by the following procedure: For each $j=1,\ldots,n$, define a set 
\begin{align}
S_j = \{i: a_i\leq c_{j}(P), d_i < d_{j}\}.
\end{align}
If there exists two jobs $i$ and $j$ which satisfy $i\in S_j$ and $c_i' > c_{j}'$, we interchange the components $c_i'$ and $c_{j}'$ in vector $\bm{c}'$. Repeat this interchange operation, until such two jobs $i$ and $j$ satisfying  $i\in S_j$ and $c_i' > c_{j}'$ cannot be found. Therefore, at the end of this procedure, if job $i$ arrives before $c_{j}(P)$ and job $i$ has an earlier due time than job $j$, then $c_i' < c_{j}'$, which is  satisfied by policy $P$. Therefore, the vector  $\bm{c}(P)$ is obtained at the end of this procedure.
In each interchange operation of this procedure, $(c_i' - c_{j}')(d_i - d_{j})\leq0$ is satisfied before the interchange of $c_i'$ and $c_{j}'$. By Lemma \ref{lem_rearrangement} and the transitivity of the ordering of majorization, we can obtain $\bm{c}(P)-\bm{d}\prec \bm{c}'- \bm{d}$, which is the first inequality in \eqref{eq_coro2_0_6}. 

Next, consider the case that two jobs $i$ and $j$ have identical due time $d_i = d_j$. Hence, $(v_i' - v_{j}')(d_i - d_{j})=0$. In this case, the service order of job $i$ and job $j$ are indeterminate in policy $P$. Nonetheless, by Lemma \ref{lem_rearrangement}, the service order of job $i$ and job $j$ does not affect the first inequality in \eqref{eq_coro2_0_6}. Hence, the first inequality in \eqref{eq_coro2_0_6} holds even when $d_i = d_j$.
 
Finally, \eqref{eq_coro2_0_3} follows from \eqref{eq_coro2_0_6}.

\emph{Step 3: We use \eqref{eq_coro2_0_3} to prove Proposition \ref{coro2_0}.}
For any $f\in\mathcal{D}_{\text{Sch-1}}$, $f(\bm{x}+\bm{d})$ is increasing and Schur convex. 
According to Theorem 3.A.8 of \cite{Marshall2011}, for all $f\in\mathcal{D}_{\text{Sch-1}}$
\begin{align}
&f(\bm{c}(P)) \nonumber\\
= &f[(\bm{c}(P)-\bm{d})+\bm{d}] \nonumber\\
\leq& f[(\bm{c}(\pi)-\bm{d})+\bm{d}] \nonumber\\
=& f(\bm{c} (\pi)).\nonumber
\end{align} 
This completes the proof.
\end{proof}

\section{Proof of Proposition 12}\label{app_coro2}
%

For  $f\in\mathcal{D}_{\text{sym}}$, \eqref{eq_ordering_2_2} and  \eqref{eq_coro2} follow from Proposition \ref{lem1} and Proposition \ref{lem1_1}.

For  $f\in\mathcal{D}_{\text{Sch-1}}$, \eqref{eq_coro2} is proven in 3 steps, which are described as follows:



\emph{Step 1: We will show that}
\begin{align}\label{eq_coro2_3}
\bm{v}(P)-\bm{d}\prec_{\text{w}} \bm{c} (\pi)-\bm{d}.
\end{align} 
According to Eq. (1.A.17) and Theorem 5.A.9 of \cite{Marshall2011},
it is sufficient to show that there exists an $n$-dimensional vector $\bm{v}'$ such that 
\begin{align}\label{eq_coro2_6}
\bm{v}(P)-\bm{d}\prec \bm{v}'- \bm{d} \leq \bm{c} (\pi)-\bm{d}.
\end{align}
Vector $\bm{v}'$ is constructed as follows: First, the components of the vector $\bm{v}'$ is a rearrangement (or permutation) of the components of the vector $\bm{v}(P)$, which can be equivalently expressed as
\begin{align}\label{eq_coro2_4}
v_{(i)}'=v_{(i)}(P),~\forall~i=1,\ldots,n.
\end{align}
Second, for each $j=1,\ldots,n$, if the completion time $c_j(\pi)$ of job $j$  is the $i$-th smallest component of $\bm{c}(\pi)$, i.e., 
\begin{align}
c_j(\pi)= c_{(i)}(\pi),
\end{align}
then $v_j'$ associated with job $j$ is the $i$-th smallest component of $\bm{v}'$, i.e., 
\begin{align}\label{eq_coro2_5}
v_j' = v_{(i)}'.
\end{align}
Combining \eqref{eq_ordering_2_2} and \eqref{eq_coro2_4}-\eqref{eq_coro2_5}, yields
\begin{align}
v_j' = v_{(i)}'=v_{(i)}(P) \leq c_{(i)}(\pi) = c_j(\pi)\nonumber
\end{align}
for $j=1,\ldots,n$. This implies $\bm{v}' \leq \bm{c} (\pi)$, and hence the second inequality in \eqref{eq_coro2_6} is proven.

The remaining task is to prove the first inequality in \eqref{eq_coro2_6}. 
First, consider the case that the due times $d_1,\ldots, d_n$ of the $n$ jobs are  different from each other. 
The vector  $\bm{v}(P)$ can be obtained from  $\bm{v}'$ by the following procedure: For each $j=1,\ldots,n$, define a set 
\begin{align}
S_j = \{i: a_i\leq v_{j}(P), d_i < d_{j}\}.
\end{align}
If there exist two jobs $i$ and $j$ which satisfy $i\in S_j$ and $v_i' > v_{j}'$, we interchange the components $v_i'$ and $v_{j}'$ in vector $\bm{v}'$. Repeat this interchange operation, until such two jobs $i$ and $j$ satisfying  $i\in S_j$ and $v_i' > v_{j}'$ cannot be found. Therefore, at the end of this procedure, if job $i$ arrives before $v_{j}(P)$ and job $i$ has an earlier due time than job $j$, then $v_i' < v_{j}'$, which is exactly the priority rule of job service satisfied by policy $P$. Therefore, the vector  $\bm{v}(P)$ is obtained at the end of this procedure.
In each interchange operation of this procedure, $(v_i' - v_{j}')(d_i - d_{j})\leq0$ is satisfied before the interchange of $v_i'$ and $v_{j}'$. By Lemma \ref{lem_rearrangement} and the transitivity of the ordering of majorization, we can obtain $\bm{v}(P)-\bm{d}\prec \bm{v}'- \bm{d}$, which is the first inequality in \eqref{eq_coro2_6}. 

Next, consider the case that two jobs $i$ and $j$ have identical due time $d_i = d_j$. Hence, $(v_i' - v_{j}')(d_i - d_{j})=0$. In this case, the service order of job $i$ and job $j$ are indeterminate in policy $P$. Nonetheless, by Lemma \ref{lem_rearrangement}, the service order of job $i$ and job $j$ does not affect the first inequality in \eqref{eq_coro2_6}. Hence, the first inequality in \eqref{eq_coro2_6} holds even when $d_i = d_j$.
 
Finally, \eqref{eq_coro2_3} follows from \eqref{eq_coro2_6}.

\emph{Step 3: We use \eqref{eq_coro2_3} to prove Proposition \ref{coro2}.}
For any $f\in\mathcal{D}_{\text{Sch-1}}$, $f(\bm{x}+\bm{d})$ is increasing and Schur convex. 
According to Theorem 3.A.8 of \cite{Marshall2011}, for all $f\in\mathcal{D}_{\text{Sch-1}}$, we have
\begin{align}
&f(\bm{v}(P)) \nonumber\\
= &f[(\bm{v}(P)-\bm{d})+\bm{d}] \nonumber\\
\leq& f[(\bm{c}(\pi)-\bm{d})+\bm{d}] \nonumber\\
=& f(\bm{c} (\pi)).\nonumber
\end{align} 
This completes the proof.

%% file: appendices_Theorem1.tex

\begin{figure}
\centering
\includegraphics[width=0.5\textwidth]{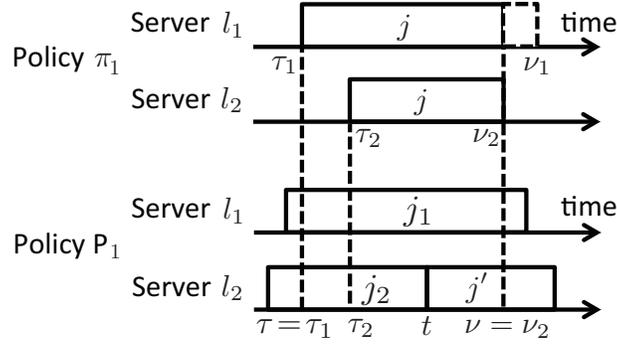} \caption{Illustration of the weak work-efficiency ordering between policy $\pi_1$ and policy $P$$_1$. In policy $\pi_1$, two copies of task $j$ are replicated on the server $l_1$ and server $l_2$ at time $\tau_1$ and $\tau_2$, where $\tau=\min\{\tau_1,\tau_2\}$. Server $l_2$ completes one copy of task $j$ at time $\nu$, server $l_1$ cancels its redundant copy of task $j$ at time $\nu$. Hence, the service duration of task $j$ is $[\tau,\nu]$ in policy $\pi_1$. In policy $P$$_1$, at least one of the servers $l_1$ and $l_2$ becomes idle before time $\nu$. In this example, server $l_2$ becomes idle at time $t\in[\tau,\nu]$ and a new task $j'$ starts execution on server $l_2$ at time $t$. Hence, the weak work-efficiency ordering is satisfied.} 
\label{fig_theorem_proof_1}
\end{figure}

\section{Proof of Lemma \ref{lem_coupling}}\label{app1}
We will need the following lemma:
\begin{lemma}\label{lem_independent}
Suppose that $X_1, \ldots, X_m$ are non-negative independent random variables, $\chi_1, \ldots, \chi_m$ are arbitrarily given non-negative constants, $R_l = [X_l - \chi_l | X_l > \chi_l]$ for  $l=1,\ldots,m$, then $R_1, \ldots, R_m$  are  mutually independent.
\end{lemma}
\begin{proof}
For all constants $t_l \geq 0$, $l=1,\dots m$, we have
\begin{align}
&\Pr[R_l >t_l, l=1,\ldots,m] \nonumber\\
=& \Pr[X_l - \chi_l >t_l, l=1,\ldots,m| X_l > \chi_l, l=1,\ldots,m]\nonumber\\
=&\frac{\Pr[X_l   >t_l + \chi_l, l=1,\ldots,m]}{\Pr[X_l > \chi_l, l=1,\ldots,m]}\nonumber\\
=& \frac{\prod_{l=1}^m \Pr[X_l   >t_l + \chi_l]}{\prod_{l=1}^m \Pr[X_l   > \chi_l]}\nonumber\\
=& \prod_{l=1}^m \Pr[X_l - \chi_l >t_l| X_l > \chi_l]\nonumber\\
=&\prod_{l=1}^m \Pr[R_l >t_l].
\end{align}
Hence, $R_1, \ldots, R_m$  are mutually independent.
\end{proof}

\begin{proof}[Proof of Lemma \ref{lem_coupling}]
We use coupling to prove Lemma \ref{lem_coupling}: We construct two policies $P_1$ and $\pi_1$ such that policy $P_1$
satisfies the same queueing discipline with policy $P$, and policy $\pi_1$ satisfies the same queueing discipline with policy $\pi$. Hence, policy $P_1$ is work-conserving.
The task and job completion times of policy $P_1$ (policy $\pi_1$) have the same distribution with those of policy $P$ (policy $\pi$). Because the state process is determined by the job parameters $\mathcal{I}$ and the task/job completion events, the state process $\{\bm{\xi}_{P_1}(t),\bm{\gamma}_{P_1}(t),t\in[0,\infty)\}$ of policy $P_1$ has the same distribution with the state process $\{\bm{\xi}_{P}(t),\bm{\gamma}_{P}(t),t\in[0,\infty)\}$ of policy $P$, and the state process $\{\bm{\xi}_{\pi_1}(t),\bm{\gamma}_{\pi_1}(t),t\in[0,\infty)\}$ of policy $\pi_1$ has the same distribution with the state process $\{\bm{\xi}_{\pi}(t),\bm{\gamma}_{\pi}(t),t\in[0,\infty)\}$  of policy $\pi$.

Next, we show that policy $P_1$ and policy $\pi_1$ can be constructed such that policy $P_1$ is weakly more work-efficient than policy $\pi_1$ with probability one.
Let us consider any task $j$ executed in policy $\pi_1$. 
As illustrated in Fig. \ref{fig_theorem_proof_1}, suppose that 
$u$ copies of task $j$ are replicated on the servers $l_1,l_2,\dots,l_u$ at the time instants $\tau_1,\tau_2,\ldots, \tau_u$ in policy $\pi_1$, where $\tau=\min_{w=1,\ldots,u} \tau_w$.\footnote{If $u=1$, there is no replication.}~In addition, suppose that server $l_w$ will complete processingits copy of task $j$ at time $\tau_w$ if there is no cancellation. Then, one of these $u$ servers will complete one copy of task $j$ at time $\nu = \min_{w=1,\ldots,u} \nu_w$, which is the earliest among these $u$ servers.
Hence, task $j$ starts service at time $\tau$ and completes service at time $\nu$ in policy $\pi_1$. Suppose that the queue is not empty (there exist unassigned tasks) during $[\tau,\nu]$ in policy $P$$_1$.
Because policy $P$$_1$ is work-conserving, all servers are busy during $[\tau,\nu]$ in policy $P$$_1$. 
In policy $P_1$, let $\tau_w+R_{l_w}$ be the earliest time that server $l_w$ becomes available to process a new task after time $\tau_w$. We will show that policy $P_1$ can be constructed such that  for all $w=1,\ldots,u$, 
\begin{align}\label{eq_sample_earlier}
\tau_w+R_{l_w}\leq \nu_w
\end{align} holds with probability one. Let $X_{l}$ denote the task service time of server $l$ and $O_{l}$ denote the cancellation delay overhead of server $l$. 
We need to consider three cases:



\emph{Case 1:} In policy $P_1$, server $l_w$ is processing task $j_w$ at time $\tau_w$, and will keep processing task $j_w$ until it completes task $j_w$ at time $\tau_w + R_{l_w}$. Suppose that server $l_w$ has spent $\chi_{l_w}$ ($\chi_{l_w}\geq0$) seconds on task $j_w$ by time $\tau_w$ in policy $P$$_1$. Then, the CCDF of $R_{l_w}$  is given by
\begin{align}\label{eq_couplecase1}
\Pr[R_{l_w}>s] = \Pr[X_{l_w}-\chi_{l_w}>s|X_{l_w}> \chi_{l_w} ].
\end{align}
Because the task service times are  NBU, we can obtain that for all $s\geq 0$ 
\begin{align}\label{eq_couplecase1_1}
\Pr[X_{l_w}-\chi_{l_w}>s|X_{l_w}> \chi_{l_w} ] \leq \Pr[X_{l_w}>s].
\end{align}
By combining \eqref{eq_couplecase1} and \eqref{eq_couplecase1_1}, we obtain
\begin{align}\label{lem_FUT_NR_1_0}
R_{l_w} \leq_{\text{st}}X_{l_w}.
\end{align}
By Theorem 1.A.1 of \cite{StochasticOrderBook}, policy $P_1$ can be constructed such that \eqref{eq_sample_earlier} always holds in \emph{Case 1}. 

%

\emph{Case 2:} In policy $P_1$, server $l_w$ is processing task $j_w$ at time $\tau_w$ and will keep processing task $j_w$ until another server completes a copy of task $j_w$; then server $l_w$ will cancel its redundant copy of task $j_w$ and will complete the cancellation operation at time $\tau_w + R_{l_w}$. 
Suppose that server $l_w$ has spent $\chi_{l_w}$($\chi_{l_w}\geq0$)  seconds on processing task $j_w$ by time $\tau_w$ in policy $P$$_1$. In addition, suppose that server $l_w$ will spend an additional $\zeta_{l_w}$($\zeta_{l_w}\geq0$)  seconds on processing task $j_w$ after time $\tau_w$  in policy $P$$_1$, before starting to cancel task $j_w$. Then, 
\begin{align}\label{eq_couplecase2_1}
R_{l_w} = \zeta_{l_w} + O_{l_w},
\end{align}
where $O_{l_w}$ the cancellation delay overhead of server $l_w$.

First, let us consider the case that task $j_w$ is not cancelled. Suppose that in this case, server $l_w$ will complete processing task $j_w$ at time $\tau_w + R'_{l_w}$. As shown in \emph{Case 1}, the service of task $j_w$ can be constructed such that 
\begin{align}\label{eq_couplecase2_2}
\tau_w + R'_{l_w}\leq \nu_w
\end{align}
always holds. 

Second, in \emph{Case 2}, task $j_w$ is cancelled at $\tau_w + \zeta_{l_w}$.
According to the LPR discipline, task cancellation only happen when the time to cancel the task is  shorter than the remaining service time to complete the task in the hazard rate ordering. Hence, for all $t\geq 0$
\begin{align}
\Pr[O_{l_w}>t] \leq \Pr[X_{l_w}-\chi_{l_w}-\zeta_{l_w}>t|X_{l_w}> \chi_{l_w}+\zeta_{l_w} ].\nonumber
\end{align}
By Theorem 1.A.1 of \cite{StochasticOrderBook}, policy $P_1$ can be constructed such that 
\begin{align}\label{eq_couplecase2_4}
O_{l_w} \leq R'_{l_w}-\zeta_{l_w}
\end{align}
always holds. By combining \eqref{eq_couplecase2_1}-\eqref{eq_couplecase2_4},
 policy $P_1$ can be constructed such that \eqref{eq_sample_earlier} always holds in \emph{Case 2}.


\emph{Case 3:} In policy $P_1$, server $l_w$ is cancelling a redundant copy of task $j_w$ at time $\tau_w$, and will complete the cancellation operation at time $\tau_w + R_{l_w}$. Suppose that  in policy $P$$_1$, server $l_w$ has spent $\chi_{l_w}$($\chi_{l_w}\geq0$) seconds on processing task $j_w$ before starting to cancel task $j_w$, and server $l_w$ has spent $\zeta_{l_w}$($\zeta_{l_w}\geq0$)  seconds on cancelling task $j_w$ by time $\tau_w$. Hence, for all $s\geq0$
\begin{align}\label{eq_couplecase3}
\Pr[R_{l_w}>s] = \Pr[O_{l_w}-\zeta_{l_w}>s|O_{l_w}> \zeta_{l_w} ].
\end{align}
According to the LPR discipline, task cancellation only happen when the time to cancel the task is  shorter than the remaining service time to complete the task in the hazard rate ordering. Hence, for all $s\geq0$
\begin{align}
&\Pr[O_{l_w}-\zeta_{l_w}>s|O_{l_w}> \zeta_{l_w} ] \nonumber\\
\leq& \Pr[X_{l_w}-\chi_{l_w}-\zeta_{l_w}>s|X_{l_w}> \chi_{l_w}+\zeta_{l_w} ].
\end{align}
Finally, because the task service times are  NBU, for all $s\geq0$
\begin{align}
&\Pr[X_{l_w}-\chi_{l_w}-\zeta_{l_w}>s|X_{l_w}> \chi_{l_w}+\zeta_{l_w} ]\nonumber\\
\leq & \Pr[X_{l_w}>s].\label{eq_couplecase3_1}
\end{align}
By combining \eqref{eq_couplecase3}-\eqref{eq_couplecase3_1}, \eqref{lem_FUT_NR_1_0} follows. By Theorem 1.A.1 of \cite{StochasticOrderBook}, policy $P_1$ can be constructed such that \eqref{eq_sample_earlier} always holds in \emph{Case 3}.

By Lemma \ref{lem_independent}, $R_{l_1},\ldots, R_{l_u}$ are mutually independent. Hence, policy $P_1$ can be constructed such that \eqref{eq_sample_earlier} holds for all $w=1,\ldots,u$ with probability one. Therefore, in policy $P$$_1$ there exists at least one of the server $l_1,\ldots,l_u$, say server $l_v$, that completes processing or cancelling a task and becomes available to process a new task before time $\nu= \min_{w=1,\ldots,u} \nu_w$. Let $t\in[\tau,\nu]$ denote the time that server $l_v$ becomes available to process a new task in policy $P_1$. Because server $l_v$ is kept busy during $[\tau,\nu]$, a new task, say task $j'$, will start execution  on server $l_v$ at time $t$ in policy $P_1$. Since the queue is not empty (there exist unassigned tasks in the queue) during $[\tau,\nu]$, according to the LPR discipline, task $j'$ cannot be a replicated copy of a task that has been assigned to some server before time $t$. Hence, task $j'$ starts service at time $t\in[\tau,\nu]$. 

In the above coupling arguments, conditioned on every possible realization of policy $P_1$ and policy $\pi_1$ before the service of task $j$ starts, we can construct the service of task $j$ in policy $\pi_1$ and the service of the corresponding task $j'$ in policy $P_1$ such that the requirement of weak work-efficiency ordering is satisfied for this pair of tasks. Next, following the proof of \cite[Theorem 6.B.3]{StochasticOrderBook}, one can continue this procedure to progressively construct the service of all tasks in policy $\pi_1$ and policy $P_1$. By this, we obtain that policy $P$$_1$ is weakly more work-efficient than policy $\pi_1$ with probability one, which completes the proof.
%
\end{proof}

\section{Proof of Lemma~2} \label{app2}

Define $\bm{s}=(s_{1},\ldots,$ $ s_{k_{\text{sum}}})$ as the sequence of task arrival times where $s_1\leq \ldots\leq s_{k_{\text{sum}}}$. Hence, $\bm{s}$ is unique determined by the job arrival times $a_1,\ldots,a_n$ and job sizes $k_1,\ldots,k_n$ which are included in the job parameters $\mathcal{I}$.
Recall that $\bm{T}_\pi=(T_{1,\pi},\ldots,$ $ T_{k_{\text{sum}},\pi})$ is the sequence of task completion times in policy $\pi$ where $T_{1,\pi}\leq \ldots\leq T_{k_{\text{sum}},\pi}$. 
We will show that \emph{for all $\pi\in\Pi$}
\begin{align}\label{eq_thm2_proof1}
[\bm{T}_{P}|\mathcal{I}] \leq_{\text{st}} [\bm{T}_\pi|\mathcal{I}].
\end{align}

We prove \eqref{eq_thm2_proof1} by using Theorem 6.B.3 of \cite{StochasticOrderBook}. Consider the first task completion time $T_{1,\pi}$. Job 1 arrives at time $a_1=0$. Note that any policy $\pi\in\Pi$ is non-preemptive. If policy $\pi$ is  work-conserving, then
\begin{align}
[T_{1,\pi}|\mathcal{I}]=\min_{l=1,\ldots,m} X_l,\nonumber
\end{align} 
otherwise, if policy $\pi$ is non-work-conserving, then
\begin{align}
[T_{1,\pi}|\mathcal{I}]\geq\min_{l=1,\ldots,m} X_l,\nonumber
\end{align} 
because of the possibility of server idleness. Since policy $P$ is work-conserving, we can obtain that for all $\pi\in\Pi$
\begin{align}\label{eq_thm2_condition1_2}
[T_{1,P}|\mathcal{I}]=\min_{l=1,\ldots,m} X_l\leq[T_{1,\pi}|\mathcal{I}].
\end{align} 

Next, consider the  evolution from $T_{j,\pi}$ to $T_{j+1,\pi}$. For any work-conserving policy $\pi\in\Pi$, we can obtain
\begin{eqnarray}\label{eq_thm2_2}
T_{j+1,\pi}\!=\max\{s_{j+1},T_{j,\pi}\} + \!\min_{l=1,\ldots,m} R_{j,l,\pi},
\end{eqnarray}
where $R_{j,l,\pi}$ is the remaining service time for server $l$ to complete the task being executed at time $\max\{s_{j+1},T_{j,\pi}\}$. 

Because the task service times are independent across the servers and the CCDF $\bar{F}$  is absolutely continuous, the probability for any two servers to complete their tasks at the same time is zero. 
Therefore, in policy $P$, when a task copy is completed on a server, the remaining $m-1$ replicated copies of this task are still being processed on the other servers; these replicated task copies are cancelled immediately and $m$ replicated copies of a new task are assigned to the servers. Suppose that server $l$ has spent $\tau_{l,\pi}$ ($\tau_{l,\pi}\geq0$) seconds on processing a task by time $\max\{s_{j+1},T_{j,\pi}\}$ in policy $\pi$. 
Then, in policy $P$, 
$\tau_{l,P}=0$ for $l=1,\ldots,m$. Hence, CCDF of $R_{j,l,P}$ is given by
\begin{align}\label{eq_thm2_7}
&\Pr\Big[ R_{j,l,P}>t\Big]= \Pr\Big[X_l>t \Big].
\end{align}
If $\pi$ is a work-conserving policy, then we have $\tau_l\geq0$ for $l=1,\ldots,m$. Hence, 
Hence, CCDF of $R_{j,l,\pi}$ is given by
\begin{align}\label{eq_thm2_8}
&\Pr\Big[R_{j,l,\pi}>t\Big]= \Pr\Big[(X_l-\tau_l)>t\Big|X_l> \tau_l \Big].
\end{align} 
Because the task service times are independent NWU, by \eqref{eq_thm2_7} and \eqref{eq_thm2_8},  for $l=1,\ldots,m$
\begin{align}
 R_{j,l,P}\leq_{\text{st}} R_{j,l,\pi}.\nonumber
\end{align}
According to Lemma \ref{lem_independent}, $R_{j,1,\pi},\ldots, R_{j,m,\pi}$ are mutual independent. Hence, using Theorem 6.B.16(b) of \cite{StochasticOrderBook}, yields 
\begin{align}\label{eq_min_thm2}
\min_{l=1,\ldots,m} R_{j,l,P}\leq_{\text{st}} \min_{l=1,\ldots,m} R_{j,l,\pi}.
\end{align}
Combining \eqref{eq_thm2_2}, \eqref{eq_min_thm2}, and the fact that $s_j$ is uniquely determined by $\mathcal{I}$, it follows that for all work-conserving policy $\pi\in\Pi$
\begin{align}\label{eq_thm2_4}
[T_{j+1,P} | \mathcal{I},  T_{j,P}= t_j] \leq_{\text{st}} [T_{j+1,\pi} | \mathcal{I},  T_{j,\pi} = t_j'] \nonumber\\
\text{whenever}~ t_j \leq t_j',j=1,2,\ldots
\end{align}
If policy $\pi$ is non-work-conserving, \eqref{eq_thm2_2} becomes
\begin{eqnarray}\label{eq_thm2_2_1}
T_{j+1,\pi}\!\geq\max\{s_{j+1},T_{j,\pi}\} + \!\min_{l=1,\ldots,m} R_{j,l,\pi},\nonumber
\end{eqnarray}
because of the possibility of server idleness. In this case, \eqref{eq_thm2_4} still holds. Hence, \eqref{eq_thm2_4} holds for all $\pi\in\Pi$.
Then, substituting \eqref{eq_thm2_condition1_2} and \eqref{eq_thm2_4} into Theorem 6.B.3 of \cite{StochasticOrderBook}, yields
\begin{align}
[(T_{1,P},\ldots,T_{j,P})|\mathcal{I}] \leq_{\text{st}} [(T_{1,\pi},\ldots,T_{j,\pi})|\mathcal{I}],~\forall~\pi\in\Pi.\nonumber
\end{align}
Hence, \eqref{eq_thm2_proof1} is proven. According to Theorem 6.B.1 of \cite{StochasticOrderBook}, this is equivalent to Lemma \ref{lem_coupling_2}. This completes the proof.


\section{Proof of Lemma~\ref{lem_coupling_3}} \label{app2_1}We use coupling to prove Lemma \ref{lem_coupling_3}: We construct two policies $P_1$ and $\pi_1$ such that policy $P_1$
satisfies the same queueing discipline with policy $P$, and policy $\pi_1$ satisfies the same queueing discipline with policy $\pi$. Hence, policy $P_1$ is work-conserving.
The task and job completion times of policy $P_1$ (policy $\pi_1$) have the same distribution with those of policy $P$ (policy $\pi$). Because the state process is determined by the job parameters $\mathcal{I}$ and the task/job completion events, the state process $\{\bm{\xi}_{P_1}(t),\bm{\gamma}_{P_1}(t),t\in[0,\infty)\}$ of policy $P_1$ has the same distribution with the state process $\{\bm{\xi}_{P}(t),\bm{\gamma}_{P}(t),t\in[0,\infty)\}$ of policy $P$, and the state process $\{\bm{\xi}_{\pi_1}(t),\bm{\gamma}_{\pi_1}(t),t\in[0,\infty)\}$ of policy $\pi_1$ has the same distribution with the state process $\{\bm{\xi}_{\pi}(t),\bm{\gamma}_{\pi}(t),t\in[0,\infty)\}$  of policy $\pi$.

Next, we show that policy $P_1$ and policy $\pi_1$ can be constructed such that policy $P_1$ is weakly more work-efficient than policy $\pi_1$ with probability one.
Let us consider any task $j$ executed in policy $\pi_1$. 
As illustrated in the upper part of Fig. \ref{fig_theorem_proof_1}, suppose that 
$u$ copies of task $j$ are replicated on the servers $l_1,l_2,\dots,l_u$ at the time instants $\tau_1,\tau_2,\ldots, \tau_u$ in policy $\pi_1$, where $\tau=\min_{w=1,\ldots,u} \tau_w$.\footnote{If $u=1$, there is no replication.} 
In addition, suppose that server $l_w$ will complete processingits copy of task $j$ at time $\tau_w$ if there is no cancellation. Then, one of these $u$ servers will complete one copy of task $j$ at time $\nu = \min_{w=1,\ldots,u} \nu_w$, which is the earliest among these $u$ servers.
Hence, task $j$ starts service at time $\tau$ and completes service at time $\nu$ in policy $\pi_1$. Suppose that the queue is not empty (there exist unassigned tasks) during $[\tau,\nu]$ in policy $P$$_1$,
we will  show that policy $P$$_1$ can be constructed such that \emph{there exists one corresponding task $j'$ which starts service during $[\tau,\nu]$.}




Because policy $P$$_1$ is work-conserving and there exist unassigned tasks at any time during $[\tau,\nu]$, all servers are busy during $[\tau,\nu]$ in policy $P$$_1$. Suppose that in policy $P_1$, $\tau_w+R_{l_w}$ is the earliest time that one task is completed on server $l_w$ after time $\tau_w$. Let $X_{l_w}$ denote the task service time of server $l_w$ which follows an exponential distribution. Because  exponential distributions are memoryless, $R_{l_w}$  follows the same exponential distribution, i.e.,
\begin{align}\label{lem_FUT_NR_1_0_0}
R_{l_w} =_{\text{st}} X_{l_w}.
\end{align}
Because the task service times are independent across the servers, by Lemma \ref{lem_independent}, $R_{l_1},\ldots, R_{l_u}$ are mutually independent. 
By Theorem 6.B.16(b) of \cite{StochasticOrderBook}, we can obtain 
\begin{align}\label{lem_FUT_NR_11}
\min_{w=1,\ldots,u}R_{l_w} =_{\text{st}} \min_{w=1,\ldots,u}X_{l_w}.
\end{align}
In policy $\pi_1$, server $l_w$ starts to process task $j$ at time $\tau_w$ for $w=1,\ldots,u$, until one of the servers $l_1,\ldots,l_u$, say server $l_v$, completes task $j$ at time $\nu$. 
According to \eqref{lem_FUT_NR_1_0_0}, \eqref{lem_FUT_NR_11}, and Theorem 1.A.1 of \cite{StochasticOrderBook}, policy $P$$_1$ and $\pi_1$ can be coupled such that in policy $P$$_1$, server $l_v$ completes a task exactly at time $\nu$. 
Since the queue is not empty (there exist unassigned tasks in the queue) during $[\tau,\nu]$, according to the R discipline, $m$ replicated copies of a new task, say task $j'$, will be assigned to the $m$ servers at time $\nu$. Hence, task $j'$ starts service at time $\nu\in[\tau,\nu]$.

In the above coupling arguments, conditioned on every possible realization of policy $P_1$ and policy $\pi_1$ before the service of task $j$ starts in policy $\pi_1$, we can construct the service of task $j$ in policy $\pi_1$ and the service of the corresponding task $j'$ in policy $P_1$ such that the requirement of weak work-efficiency ordering is satisfied for this pair of tasks. 
Next, following the proof of \cite[Theorem 6.B.3]{StochasticOrderBook}, one can continue the above procedure to progressively construct the service of all tasks in policy $\pi_1$ and policy $P_1$. By this, we obtain that policy $P$$_1$ is weakly more work-efficient than policy $\pi_1$ with probability one, which completes the proof.
\section{Proof of Theorem 2}\label{app_lem7}
Let us consider  ${C}_i(\text{FUT-LPR})-V_i(\text{FUT-LPR})$.
At time $V_i(\text{FUT-LPR})$, all tasks of job $i$ are completed or under service.
if $k_i>m$,  then job $i$ has at most $m$ incomplete tasks that are under service at time $V_i(\text{FUT-LPR})$; if $k_i\leq m$, then job $i$ has at most $k_i$ incomplete tasks that are under service  at time $V_i(\text{FUT-LPR})$. Therefore, in policy FUT-LPR, no more than $k_i\wedge m = \min\{k_i,m\}$ tasks of job $i$ are completed during the time interval $[V_i(\text{FUT}$ $\text{-LPR}), {C}_i(\text{FUT-LPR})]$.

Suppose that at time $V_i(\text{FUT-LPR})$, a set of servers  $\mathcal{S}_i\subseteq\{1,\ldots,m\}$ are processing the tasks of job $i$, which satisfies $|\mathcal{S}_i|\leq k_i \wedge m$ and there is no replications in the set of servers $\mathcal{S}_i$. Note that if some servers in $\{1,\ldots,m\}/\mathcal{S}_i$ are processing the replicated task copies of job $i$, the delay gap that we will obtain will be even smaller.   

Let $\chi_l$ denote the amount of time that server $l\in \mathcal{S}_i$ has spent on executing a task of job $i$ by time $V_i(\text{FUT-LPR})$ in policy FUT-LPR. Let $R_{l}$ denote the remaining service time of server $l\in \mathcal{S}_i$ for executing this task after time $V_i(\text{FUT-LPR})$. Then, $R_l$ can be expressed as $R_l = [X_l - \chi_l | X_l > \chi_l]$. Because the $X_l$'s are independent NBU random variables with mean $\mathbb{E}[X_l]=1/\mu_l$, for all realizations of $\chi_l$
\begin{align}
[R_l|\chi_l] \leq_{\text{st}} X_l,~\forall~l\in \mathcal{S}_i.\nonumber
\end{align}
In addition, Theorem 3.A.55 of \cite{StochasticOrderBook} tells us that
\begin{align}
X_l \leq_{\text{icx}} Z_l,~\forall~l\in \mathcal{S}_i,\nonumber
\end{align}
where $\leq_{\text{icx}}$ is the increasing convex order defined in \cite[Chapter 4]{StochasticOrderBook} and the $Z_l$'s are independent exponential random variables with mean $\mathbb{E}[Z_l]=\mathbb{E}[X_l] =\mu_l$. Hence, 
\begin{align}
[R_l|\chi_l] \leq_{\text{icx}} Z_l,~\forall~l\in \mathcal{S}_i.\nonumber
\end{align}
Lemma \ref{lem_independent} tells us that 
the $R_l$'s are conditional independent for any given realization of $\{\chi_l,l\in \mathcal{S}_i\}$. Hence, by Corollary 4.A.16 of \cite{StochasticOrderBook}, for all realizations of $\mathcal{S}_i$ and $\{\chi_l,l\in \mathcal{S}_i\}$
\begin{align}\label{eq_exp_order}
\big[\max_{l\in \mathcal{S}_i} R_l \big| \mathcal{S}_i, \{\chi_l, l\in \mathcal{S}_i\}\big] \leq_{\text{icx}} \big[\max_{l\in \mathcal{S}_i} Z_l\big| \mathcal{S}_i\big].
\end{align}

Then,
\begin{align}
&\mathbb{E}[{C}_i(\text{FUT-LPR})-V_i(\text{FUT-LPR})|\mathcal{S}_i ,\{\chi_l, l\in \mathcal{S}_i\}]\nonumber\\
\leq & \mathbb{E}\!\left[ \max_{l\in \mathcal{S}_i} R_l \bigg| \mathcal{S}_i ,\{\chi_l, l\in \mathcal{S}_i\}\right]\nonumber\\
\leq&\mathbb{E}\!\left[ \max_{l\in \mathcal{S}_i} Z_l \bigg| \mathcal{S}_i \right]\label{eq_gap_condition3}\\
\leq &\mathbb{E}\!\left[\max_{l=1,\ldots,k_i \wedge m} Z_l\right]\label{eq_gap_condition4}\\
\leq &\sum_{l=1}^{k_i \wedge m} \frac{1}{\sum_{j=1}^l \mu_j},\label{eq_gap_condition5}
\end{align}
where 
\eqref{eq_gap_condition3} is due to
\eqref{eq_exp_order} and Eq. (4.A.1) of \cite{StochasticOrderBook},
\eqref{eq_gap_condition4} is due to $\mu_1\leq \ldots\leq \mu_M$, $|\mathcal{S}_i|\leq k_i \wedge m$, and the fact that $\max_{l=1,\ldots,k_i \wedge m} Z_l$ is independent of $\mathcal{S}_i$, and \eqref{eq_gap_condition5} is due to the property of exponential distributions. Because $\mathcal{S}_{i}$ and $\{\chi_{l}, l\in \mathcal{S}_{i}\}$
 are random variables which are determined by the job parameters $\mathcal{I}$, taking the conditional expectation for given $\mathcal{I}$ in \eqref{eq_gap_condition5}, yields
\begin{align} 
\mathbb{E}\!\left[ {C}_i(\text{FUT-LPR})-V_i(\text{FUT-LPR}) | \mathcal{I}\right] \leq \sum_{l=1}^{k_i \wedge m} \frac{1}{\sum_{j=1}^l \mu_j}.\nonumber
\end{align} 
By taking the average over all $n$ jobs, the first inequality of \eqref{eq_gap} is proven.
In addition, it is known that for each $k=1,2,\ldots,$
\begin{align}
\sum_{l=1}^k \frac{1}{l} \leq \ln(k)+1.\nonumber
\end{align}
By this, the second inequality of \eqref{eq_gap} holds. This completes the proof.

%% file: appendices_3.tex
\section{Proof of (22)}\label{app_lem7_NWU}  
Consider the time difference ${C}_i(\text{FUT-R})-{V}_i(\text{FUT-R})$.
In policy FUT-R, all servers are allocated to process $m$ replicated copies of a task from the job with the fewest unassigned tasks. Hence, at time ${V}_i(\text{FUT-R})$, one task of job $i$ are being processed by all $m$ servers. 
Because the $X_l$'s are independent exponential random variables with mean $\mathbb{E}[X_l]=1/\mu_l$, 
by Theorem 3.A.55 of \cite{StochasticOrderBook}, we can obtain
\begin{align}
&\mathbb{E}[{C}_i(\text{FUT-R})-{V}_i(\text{FUT-R})|\mathcal{I}]\nonumber\\
\leq &\mathbb{E}\!\left[\min_{l=1,\ldots,m} X_l\bigg|\mathcal{I}\right]\label{eq_gap_NWU_condition1}\\
=&\mathbb{E}\!\left[ \min_{l=1,\ldots,m} X_l\right]\label{eq_gap_NWU_condition2}\\
\leq &\frac{1}{\sum_{l=1}^m \mu_l}.\label{eq_gap_NWU_condition5}
\end{align}
where \eqref{eq_gap_NWU_condition1} is because one task of job $i$ are being processed by all $m$ servers at time ${V}_i(\text{FUT-R})$, \eqref{eq_gap_NWU_condition2} is because $X_l$ is independent of $\mathcal{I}$, 
and \eqref{eq_gap_NWU_condition5} is due to the property of exponential distributions. By this,  \eqref{eq_gap_NWU} is proven.

\section{Proof of Theorem \ref{thm_dis4}}\label{app_dis_EDD} 
By Theorem \ref{thm3}, Theorem \ref{thm4}, Theorem \ref{thm4_exp}, and the fact that $V_{ih}(\text{EDD-GR})\leq C_{ih}(\text{EDD-GR})$, we obtain that for all $h=1,\ldots,g$, $\pi\in\Pi$,
and $\mathcal{I}$
\begin{align}
\Big[\max_{i=1,\ldots,n}[V_{ih}(\text{EDD-GR})-d_i]\Big|\mathcal{I}\Big] \leq_{\text{st}} \Big[\max_{i=1,\ldots,n}[C_{ih}(\pi)-d_i]\Big|\mathcal{I}\Big]. \nonumber
\end{align}
In addition, according to \eqref{eq_sub_job_relation}, we can get  
\begin{align}
&L_{\max}(\bm{V}(\text{EDD-GR})) = \max_{i=1,\ldots,n}[V_{i}(\text{EDD-GR})-d_i] = \max_{h=1,\ldots, g} \max_{i=1,\ldots,n}[V_{ih}(\text{EDD-GR})-d_i], \nonumber\\
&L_{\max}(\bm{C}(\pi)) = \max_{i=1,\ldots,n}[C_{i}(\pi)-d_i] = \max_{h=1,\ldots, g} \max_{i=1,\ldots,n}[C_{ih}(\pi)-d_i].\nonumber
\end{align}
Then, by using Theorem 6.B.16.(b) of \cite{StochasticOrderBook} and the independence of the service across the server groups, it follows that
\begin{align}
\Big[\max_{i=1,\ldots,n}[V_{i}(\text{EDD-GR})-d_i]\Big|\mathcal{I}\Big] \leq_{\text{st}} \Big[\max_{i=1,\ldots,n}[C_{i}(\pi)-d_i]\Big|\mathcal{I}\Big]. \nonumber
\end{align}
By this, Theorem \ref{thm_dis4} is proven.

\section{Proof of Theorem \ref{thm2_dist}}\label{app_thm2_dist} 

For any policy $\pi\in\Pi$, suppose that policy FUT-GR$_1$  (policy $\pi_1$) satisfies the same queueing discipline with policy FUT-GR (policy $\pi$).
By  the proof arguments of  Theorems \ref{thm1}-\ref{thm2_exp} and the fact that $V_{ih}(\text{FUT-GR}_1)\leq C_{ih}(\text{FUT-GR}_1)$, policy FUT-GR$_1$ and policy $\pi_1$ can be coupled such that 
\begin{align}\label{eq_thm2_dist1}
V_{(i),h}(\text{FUT-GR}_1)\leq C_{(i),h}(\pi_1) 
\end{align}
holds with probability one for $h=1,\ldots,g$ and $i=1,\ldots,n$. 
Under per-job data locality constraints, for each job $i$ there exists $u(i)\in\{1,\ldots,g\}$ such that $C_{i,u(i)} (\pi_1) = C_{i} (\pi_1) $, $V_{i,u(i)} (\text{FUT-GR}_1) = V_{i} (\text{FUT-GR}_1) $, and $C_{i,h} (\pi_1) = V_{i,h} (\text{FUT-GR}_1) = 0$ for all $h\neq u(i)$. By this and \eqref{eq_thm2_dist1}, we can obtain
\begin{align}
V_{(i)}(\text{FUT-GR}_1)\leq C_{(i)}(\pi_1)\nonumber
\end{align}
holds with probability one for $i=1,\ldots,n$.
Then, Because any $f\in\mathcal{D}_{\text{sym}}$ is a symmetric and increasing function, we can obtain
\begin{align}
&f(\bm{V} (\text{FUT-GR}_1))=f(\bm{V}_{\uparrow} (\text{FUT-GR}_1)) \nonumber\\
\leq_{}&f(\bm{C}_{\uparrow} (\pi_1))= f(\bm{C} (\pi_1)).\nonumber
\end{align} 
holds with probability one.
Then, by the property of stochastic ordering \cite[Theorem 1.A.1]{StochasticOrderBook}, we can obtain \eqref{eq_delaygap2_dist}. This completes the proof.

\section{Proof of Theorem \ref{thm3_dist}}\label{app_thm3_dist} 

If each server group $h$ and its sub-job parameters $\mathcal{I}_h$ satisfy the conditions of Theorem \ref{coro_thm3_1}, Theorem \ref{coro4_1}, or Theorem \ref{coro_thm3_1_exp}, then we can obtain
\begin{itemize}
\item[1.] The job with the earliest due time among all jobs with unassigned tasks is also the jobs with fewest unassigned tasks,

\item[2.] Each server group $h$ and its sub-job parameters $\mathcal{I}_h$  satisfy the conditions of Theorem \ref{thm1}, Theorem \ref{thm2}, or Theorem \ref{thm2_exp}.
\end{itemize}
Then, by using Theorem \ref{thm2_dist}, yields that \eqref{eq_delaygap3_dist} holds for all $f\in\mathcal{D}_{\text{sym}}$. 

Next, we consider the delay metrics in $\mathcal{D}_{\text{Sch-1}}$.
For any policy $\pi\in\Pi$, suppose that policy EDD-GR$_1$  (policy $\pi_1$) satisfies the same queueing discipline  with policy EDD-GR (policy $\pi$).
By using the proof arguments of Theorem \ref{coro_thm3_1}, Theorem \ref{coro4_1}, and Theorem \ref{coro_thm3_1_exp}, and the fact that $V_{ih}(\text{EDD-GR}_1)\leq C_{ih}(\text{EDD-GR}_1)$, policy EDD-GR$_1$ and policy $\pi_1$ can be coupled such that 
\begin{align}
\bm{V}_h(\text{EDD-GR}_1)-\bm{d}\prec_{\text{w}} \bm{C}_h (\pi_1)-\bm{d}\nonumber
\end{align}
holds with probability one for each $h=1,\ldots,g$. Then, Theorem 5.A.7 of \cite{Marshall2011} tells us that
\begin{align}\label{eq_thm3_dist}
(\bm{V}_1(\text{EDD-GR}_1)-\bm{d}, \ldots,\bm{V}_g(\text{EDD-GR}_1)-\bm{d})\prec_{\text{w}} (\bm{C}_1 (\pi_1)-\bm{d},\ldots,\bm{C}_g (\pi_1)-\bm{d})
\end{align}
holds with probability one. Under per-job data locality constraints, for each job $i$ there exists $u(i)\in\{1,\ldots,g\}$ such that $C_{i,u(i)} (\pi_1) = C_{i} (\pi_1) $, $V_{i,u(i)} (\text{EDD-GR}_1) = V_{i} (\text{EDD-GR}_1) $, and $C_{i,h} (\pi_1) = V_{i,h} (\text{FUT-GR}_1) $ $= 0$ for all $h\neq u(i)$.  By this and \eqref{eq_thm3_dist}, we get that
\begin{align}
\bm{V}(\text{EDD-GR}_1)-\bm{d}\prec_{\text{w}} \bm{C} (\pi_1)-\bm{d}\nonumber
\end{align}
holds with probability one.
In addition, by Theorem 3.A.8 of \cite{Marshall2011},
\begin{align}
f(\bm{V}(\text{EDD-GR}_1))\leq f(\bm{C}(\pi_1)) \nonumber
\end{align}
holds with probability one for all $f\in\mathcal{D}_{\text{Sch-1}}$.
Then, by the property of stochastic ordering \cite[Theorem 1.A.1]{StochasticOrderBook}, we can obtain \eqref{eq_delaygap3_dist}. This completes the proof.

%% file: report.bbl
\begin{thebibliography}{10}
\providecommand{\url}[1]{#1}
\csname url@samestyle\endcsname
\providecommand{\newblock}{\relax}
\providecommand{\bibinfo}[2]{#2}
\providecommand{\BIBentrySTDinterwordspacing}{\spaceskip=0pt\relax}
\providecommand{\BIBentryALTinterwordstretchfactor}{4}
\providecommand{\BIBentryALTinterwordspacing}{\spaceskip=\fontdimen2\font plus
\BIBentryALTinterwordstretchfactor\fontdimen3\font minus
  \fontdimen4\font\relax}
\providecommand{\BIBforeignlanguage}[2]{{%
\expandafter\ifx\csname l@#1\endcsname\relax
\typeout{** WARNING: IEEEtran.bst: No hyphenation pattern has been}%
\typeout{** loaded for the language `#1'. Using the pattern for}%
\typeout{** the default language instead.}%
\else
\language=\csname l@#1\endcsname
\fi
#2}}
\providecommand{\BIBdecl}{\relax}
\BIBdecl

\bibitem{Linden2006}
G.~Linden, http://glinden.blogspot.com/2006/11/marissa-mayer-at-web-20.html/.

\bibitem{LindenStanfordtalk}
------, ``Make data useful,''
  http://www.gduchamp.com/media/StanfordDataMining.2006-11-28.pdf, Stanford
  CS345 Talk, 2006.

\bibitem{Farber2006}
D.~Farber, http://www.zdnet.com/article/googles-marissa-mayer-speed-wins/.

\bibitem{Martin2007}
R.~Martin,
  http://www.informationweek.com/wall-streets-quest-to-process-data-at-the-speed-of-light/d/d-id/1054287?

\bibitem{mapreduce}
J.~Dean and S.~Ghemawat, ``{MapReduce}: Simplified data processing on large
  clusters,'' in \emph{USENIX OSDI}, Dec. 2004, pp. 137--150.

\bibitem{Dean:2013:Tail}
J.~Dean and L.~A. Barroso, ``The tail at scale,'' \emph{Commun. ACM}, vol.~56,
  no.~2, pp. 74--80, Feb. 2013.

\bibitem{Ghare:2004}
G.~D. Ghare and S.~T. Leutenegger, ``Improving speedup and response times by
  replicating parallel programs on a {SNOW},'' in \emph{JSSPP}, 2004.

\bibitem{Cirne2007213}
W.~Cirne, F.~Brasileiro, D.~Paranhos, L.~W. Goes, and W.~Voorsluys, ``On the
  efficacy, efficiency and emergent behavior of task replication in large
  distributed systems,'' \emph{Parallel Computing}, vol.~33, no.~3, pp. 213 --
  234, 2007.

\bibitem{vulimiri13latency}
A.~Vulimiri, P.~B. Godfrey, R.~Mittal, J.~Sherry, S.~Ratnasamy, and S.~Shenker,
  ``Low latency via redundancy,'' in \emph{ACM CoNEXT}, 2013.

\bibitem{ShengboInfocom}
S.~Chen, Y.~Sun, U.~Kozat, L.~Huang, P.~Sinha, G.~Liang, X.~Liu, and N.~B.
  Shroff, ``When queueing meets coding: Optimal-latency data retrieving scheme
  in storage clouds,'' in \emph{IEEE INFOCOM}, 2014.

\bibitem{Wang:2014}
D.~Wang, G.~Joshi, and G.~Wornell, ``Efficient task replication for fast
  response times in parallel computation,'' in \emph{ACM Sigmetrics}, 2014.

\bibitem{Wang2015}
------, ``Using straggler replication to reduce latency in large-scale parallel
  computing,'' in \emph{ACM SIGMETRICS Workshop on Distributed Cloud
  Computing}, 2015.

\bibitem{shah-Allerton-2013}
N.~B. Shah, K.~Lee, and K.~Ramchandran, ``When do redundant requests reduce
  latency?'' in \emph{Allerton Conference}, 2013.

\bibitem{Kristen2015}
K.~Gardner, S.~Zbarsky, S.~Doroudi, M.~Harchol-Balter, E.~Hyyti\"{a}, and
  A.~Scheller-Wolf, ``Queueing with redundant requests: First exact analysis,''
  in \emph{ACM Sigmetrics}, 2015.

\bibitem{Lee-Allerton-2015}
K.~Lee, R.~Pedarsani, and K.~Ramchandran, ``On scheduling redundant requests
  with cancellation overheads,'' in \emph{Allerton Conference}, 2015.

\bibitem{Ananthanarayanan11}
G.~Ananthanarayanan, A.~Ghodsi, S.~Shenker, and I.~Stoica, ``Why let resources
  idle? aggressive cloning of jobs with dolly,'' in \emph{USENIX HotCloud},
  2011.

\bibitem{Ananthanarayanan13}
------, ``Effective straggler mitigation: Attack of the clones,'' in
  \emph{USENIX NSDI}, 2013.

\bibitem{Liang2013_2}
G.~Liang and U.~Kozat, ``{TOFEC}: Achieving optimal throughput-delay trade-off
  of cloud storage using erasure codes,'' in \emph{IEEE INFOCOM}, 2014.

\bibitem{taskplacement2011}
V.~Chudnovsky, R.~Rifaat, J.~Hellerstein, B.~Sharma, and C.~Das, ``Modeling and
  synthesizing task placement constraints in google compute clusters,'' in
  \emph{Symposium on Cloud Computing}, 2011.

\bibitem{Sparrow:2013}
K.~Ousterhout, P.~Wendell, M.~Zaharia, and I.~Stoica, ``Sparrow: Distributed,
  low latency scheduling,'' in \emph{ACM SOSP}, 2013, pp. 69--84.

\bibitem{Leonardi:1997}
S.~Leonardi and D.~Raz, ``Approximating total flow time on parallel machines,''
  in \emph{ACM STOC}, 1997.

\bibitem{Weiss:1992}
G.~Weiss, ``Turnpike optimality of {Smith's} rule in parallel machines
  stochastic scheduling,'' \emph{Math. Oper. Res.}, vol.~17, no.~2, pp.
  255--270, May 1992.

\bibitem{Weiss:1995}
------, ``On almost optimal priority rules for preemptive scheduling of
  stochastic jobs on parallel machines,'' \emph{Advances in Applied
  Probability}, vol.~27, no.~3, pp. 821--839, 1995.

\bibitem{Dacre1999}
M.~Dacre, K.~Glazebrook, and J.~Niño-Mora, ``The achievable region approach to
  the optimal control of stochastic systems,'' \emph{Journal of the Royal
  Statistical Society: Series B (Statistical Methodology)}, vol.~61, no.~4, pp.
  747--791, 1999.

\bibitem{Ying2015}
L.~Ying, R.~Srikant, and X.~Kang, ``The power of slightly more than one sample
  in randomized load balancing,'' in \emph{IEEE INFOCOM}, 2015.

\bibitem{Stolyar_heavy2004}
A.~L. Stolyar, ``Maxweight scheduling in a generalized switch: State space
  collapse and workload minimization in heavy traffic,'' \emph{The Annals of
  Applied Probability}, vol.~14, no.~1, pp. 1--53, 2004.

\bibitem{Leonard_Kleinrock_book}
L.~Kleinrock, \emph{Queueing Systems}.\hskip 1em plus 0.5em minus 0.4em\relax
  John Wiley and Sons, 1975, vol. 1\& 2.

\bibitem{Jose2010}
J.~Nino-Mora, ``Conservation laws and related applications,'' in \emph{Wiley
  Encyclopedia of Operations Research and Management Science}.\hskip 1em plus
  0.5em minus 0.4em\relax John Wiley \& Sons, Inc., 2010.

\bibitem{Gittins:11}
J.~C. Gittins, K.~Glazebrook, and R.~Weber, \emph{Multi-armed Bandit Allocation
  Indices}, 2nd~ed.\hskip 1em plus 0.5em minus 0.4em\relax Wiley, Chichester,
  NY, 2011.

\bibitem{Zaharia:2008}
M.~Zaharia, A.~Konwinski, A.~D. Joseph, R.~Katz, and I.~Stoica, ``Improving
  {MapReduce} performance in heterogeneous environments,'' in \emph{USENIX
  OSDI}, 2008.

\bibitem{Ananthanarayanan:2010}
G.~Ananthanarayanan, S.~Kandula, A.~Greenberg, I.~Stoica, Y.~Lu, B.~Saha, and
  E.~Harris, ``Reining in the outliers in map-reduce clusters using {Mantri},''
  in \emph{USENIX OSDI}, 2010.

\bibitem{Dremel2010}
S.~Melnik, A.~Gubarev, J.~J. Long, G.~Romer, S.~Shivakumar, M.~Tolton, and
  T.~Vassilakis, ``Dremel: Interactive analysis of web-scale datasets,'' in
  \emph{VLDB}, 2010.

\bibitem{Ananthanarayanan2014}
G.~Ananthanarayanan, M.~C.-C. Hung, X.~Ren, I.~Stoica, A.~Wierman, and M.~Yu,
  ``{GRASS}: Trimming stragglers in approximation analytics,'' in \emph{USENIX
  NSDI}, 2014.

\bibitem{vulimiri12latency}
A.~Vulimiri, O.~Michel, P.~B. Godfrey, and S.~Shenker, ``More is less: Reducing
  latency via redundancy,'' in \emph{ACM HotNets}, 2012.

\bibitem{addShengboOriginal}
G.~Liang and U.~Kozat, ``{FAST CLOUD}: Pushing the envelope on delay
  performance of cloud storage with coding,'' \emph{IEEE/ACM Trans.
  Networking}, Dec. 2014.

\bibitem{DTN-delay}
S.~Jain, M.~Demmer, R.~Patra, and K.~Fall, ``Using redundancy to cope with
  failures in a delay tolerant network,'' in \emph{ACM SIGCOMM}, 2005.

\bibitem{Li-Allerton-2015}
S.~Li, M.~A. Maddah-Ali, and A.~S. Avestimehr, ``Coded {MapReduce},'' in
  \emph{Allerton Conference}, 2015.

\bibitem{Lee-NIPS-2015}
K.~Lee, M.~Lam, R.~Pedarsani, D.~Papailiopoulos, and K.~Ramchandran, ``Speeding
  up distributed machine learning using codes,'' in \emph{NIPS workshop on
  Machine Learning Systems}, 2015.

\bibitem{huang-isit-2012}
L.~Huang, S.~Pawar, H.~Zhang, and K.~Ramchandran, ``Codes can reduce queueing
  delay in data centers,'' in \emph{IEEE ISIT}, 2012.

\bibitem{Joshi-2012}
G.~Joshi, Y.~Liu, and E.~Soljanin, ``On the delay-storage trade-off in content
  download from coded distributed storage systems,'' \emph{IEEE J. Sel. Areas
  Commun.}, vol.~32, pp. 989--997, May 2014.

\bibitem{shah-mdsq-2012}
N.~B. Shah, K.~Lee, and K.~Ramchandran, ``The {MDS} queue: Analysing latency
  performance of codes,'' in \emph{IEEE ISIT}, 2014.

\bibitem{Kumar2014}
\BIBentryALTinterwordspacing
A.~Kumar, R.~Tandon, and T.~C. Clancy, ``On the latency of erasure-coded cloud
  storage systems,'' 2014. [Online]. Available:
  \url{http://arxiv.org/abs/1405.2833}
\BIBentrySTDinterwordspacing

\bibitem{BinInfoCom2016}
B.~Li, A.~Ramamoorthy, and R.~Srikant, ``Mean-field-analysis of coding versus
  replication in cloud storage systems,'' in \emph{IEEE INFOCOM}, 2016.

\bibitem{Gardner2016}
K.~Gardner, S.~Zbarsky, M.~Harchol-Balter, and A.~Scheller-Wolf, ``The power of
  d choices for redundancy,'' in \emph{ACM Sigmetrics}, 2016.

\bibitem{Gardner2016_2}
K.~Gardner, M.~Harchol-Balter, and A.~Scheller-Wolf, ``A better model for job
  redundancy: Decoupling server slowdown and job size,'' in \emph{IEEE
  MASCOTS}, Sept 2016, pp. 1--10.

\bibitem{Borst2003}
S.~Borst, O.~Boxma, J.~Groote, and S.~Mauw, ``\BIBforeignlanguage{English}{Task
  allocation in a multi-server system},''
  \emph{\BIBforeignlanguage{English}{Journal of Scheduling}}, vol.~6, no.~5,
  pp. 423--436, 2003.

\bibitem{Righter2008}
G.~Koole and R.~Righter, ``\BIBforeignlanguage{English}{Resource allocation in
  grid computing},'' \emph{\BIBforeignlanguage{English}{Journal of
  Scheduling}}, vol.~11, no.~3, pp. 163--173, 2008.

\bibitem{Kim2010}
Y.~Kim, R.~Righter, and R.~Wolff, ``Grid scheduling with {NBU} service times,''
  \emph{Operations Research Letters}, vol.~38, no.~6, pp. 502 -- 504, 2010.

\bibitem{Joshi2015}
G.~Joshi, E.~Soljanin, and G.~Wornell, ``Efficient redundancy techniques for
  latency reduction in cloud systems,'' \url{http://arxiv.org/abs/1508.03599},
  2015.

\bibitem{Sun2015}
Y.~Sun, Z.~Zheng, C.~E. Koksal, K.-H. Kim, and N.~B. Shroff, ``Provably delay
  efficient data retrieving in storage clouds,'' in \emph{IEEE INFOCOM}, 2015.

\bibitem{Marshall2011}
A.~W. Marshall, I.~Olkin, and B.~C. Arnold, \emph{Inequalities: Theory of
  Majorization and Its Applications}, 2nd~ed.\hskip 1em plus 0.5em minus
  0.4em\relax Springer, 2011.

\bibitem{StochasticOrderBook}
M.~Shaked and J.~G. Shanthikumar, \emph{Stochastic Orders}.\hskip 1em plus
  0.5em minus 0.4em\relax Springer, 2007.

\bibitem{Googletrace2012}
C.~Reiss, A.~Tumanov, G.~R. Ganger, R.~H. Katz, and M.~A. Kozuch,
  ``Heterogeneity and dynamicity of clouds at scale: Google trace analysis,''
  in \emph{ACM SoCC}, 2012.

\bibitem{C3:2015}
L.~Suresh, M.~Canini, S.~Schmid, and A.~Feldmann, ``C3: Cutting tail latency in
  cloud data stores via adaptive replica selection,'' in \emph{USENIX NSDI},
  Oakland, CA, May 2015, pp. 513--527.

\bibitem{Arnold99}
B.~Arnold, E.~Castillo, and J.~M. Sarabia, \emph{Conditional Specifications of
  Statistical Models}.\hskip 1em plus 0.5em minus 0.4em\relax Springer-Verlag
  New York, Inc., 1999.

\bibitem{Christodoulopoulos2008}
K.~Christodoulopoulos, V.~Gkamas, and E.~Varvarigos, ``Statistical analysis and
  modeling of jobs in a grid environment,'' \emph{Journal of Grid Computing},
  vol.~6, no.~1, pp. 77--101, 2008.

\bibitem{mitzenmacher2001power}
M.~Mitzenmacher, ``The power of two choices in randomized load balancing,''
  \emph{IEEE Transactions on Parallel and Distributed Systems}, vol.~12,
  no.~10, pp. 1094--1104, 2001.

\bibitem{vvedenskaya1996queueing}
N.~D. Vvedenskaya, R.~L. Dobrushin, and F.~I. Karpelevich, ``Queueing system
  with selection of the shortest of two queues: An asymptotic approach,''
  \emph{Problemy Peredachi Informatsii}, vol.~32, no.~1, pp. 20--34, 1996.

\bibitem{Borg2015}
A.~Verma, L.~Pedrosa, M.~R. Korupolu, D.~Oppenheimer, E.~Tune, and J.~Wilkes,
  ``Large-scale cluster management at {Google} with {Borg},'' in
  \emph{EuroSys}, Bordeaux, France, 2015.

\bibitem{michael2012book}
M.~L. Pinedo, \emph{Scheduling: Theory, Algorithms, and Systems}, 4th~ed.\hskip
  1em plus 0.5em minus 0.4em\relax Springer, 2012.

\bibitem{Kushner2004}
H.~J. Kushner and P.~A. Whiting, ``Convergence of proportional-fair sharing
  algorithms under general conditions,'' \emph{IEEE Transactions on Wireless
  Communications}, vol.~3, no.~4, pp. 1250--1259, July 2004.

\bibitem{Schrage68}
L.~Schrage, ``A proof of the optimality of the shortest remaining processing
  time discipline,'' \emph{Operations Research}, vol.~16, pp. 687--690, 1968.

\bibitem{Smith78}
D.~R. Smith, ``A new proof of the optimality of the shortest remaining
  processing time discipline,'' \emph{Operations Research}, vol.~16, pp.
  197--199, 1978.

\bibitem{Liu1995}
Z.~Liu, P.~Nain, and D.~Towsley, ``Sample path methods in the control of
  queues,'' \emph{Queueing Systems}, vol.~21, no.~3, pp. 293--335.

\bibitem{Jackson55}
J.~R. Jackson, ``Scheduling a production line to minimize maximum tardiness,''
  management Science Research Report, University of California, Los Angeles,
  CA, 1955.

\bibitem{Baccelli:1993}
F.~Baccelli, Z.~Liu, and D.~Towsley, ``Extremal scheduling of parallel
  processing with and without real-time constraints,'' \emph{J. ACM}, vol.~40,
  no.~5, pp. 1209--1237, Nov. 1993.

\bibitem{Chang93rearrangement_majorization}
C.-S. Chang and D.~D. Yao, ``Rearrangement, majorization and stochastic
  scheduling,'' \emph{Math. of Oper. Res}, 1993.

\bibitem{ParetoDis}
B.~C. Arnold, ``\BIBforeignlanguage{English}{Pareto and generalized pareto
  distributions},'' in \emph{\BIBforeignlanguage{English}{Modeling Income
  Distributions and Lorenz Curves}}, ser. Economic Studies in Equality, Social
  Exclusion and Well-Being, D.~Chotikapanich, Ed.\hskip 1em plus 0.5em minus
  0.4em\relax Springer New York, 2008, vol.~5, pp. 119--145.

\end{thebibliography}
